\documentclass[12pt]{article}
\usepackage[paper=a4paper,left=2.5cm,right=2.5cm,bottom=3.5cm,footskip=1.5cm,top=3.5cm,headheight=2cm,headsep=1.5cm]{geometry}
\usepackage{amsmath}
\usepackage{amsfonts}
\usepackage{verbatim}
\usepackage{color}
\usepackage{theorem}
\usepackage{mathtools}

\usepackage{graphicx}
\usepackage[normal]{subfigure}
\definecolor{cblack}{rgb}{0,0,0}
\definecolor{cverde}{rgb}{0,.5,0}
\definecolor{light}{rgb}{.8,.8,.8}
\definecolor{dark}{rgb}{.5,.5,.5}
\definecolor{auburn}{rgb}{0.43,0.21,0.1}
\definecolor{airforceblue}{rgb}{0.36,0.54,0.66}

\newcommand{\bA}{\boldsymbol{A}}

\newcommand{\bB}{\boldsymbol{B}}
\newcommand{\be}{\boldsymbol{e}}

\newcommand{\bF}{\boldsymbol{F}}
\newcommand{\bH}{\boldsymbol{H}}

\newcommand{\bK}{\boldsymbol{K}}
\newcommand{\bL}{\boldsymbol{L}}
\newcommand{\bm}{\boldsymbol{m}}

\newcommand{\bT}{\boldsymbol{T}}
\newcommand{\bO}{\boldsymbol{\Omega}}

\newcommand{\bzs}{\boldsymbol{\sigma}}
\newcommand{\C}{\mathbb{C}}
\newcommand{\CA}{\mathcal{A}}
\newcommand{\CB}{\mathcal{B}}
\newcommand{\CC}{\mathcal{C}}
\newcommand{\CD}{\mathcal{D}}

\newcommand{\CP}{\mathcal{P}}

\newcommand{\dn}{{\text{\tiny dn}}}
\newcommand{\ds}{\displaystyle}

\newcommand{\R}{\mathbb{R}}

\newcommand{\tcm}{\textcolor{magenta}}

\newcommand{\up}{{\text{\tiny up}}}
\newcommand{\za}{\alpha}

\newcommand{\zl}{\lambda}

\newcommand{\zs}{\sigma}
\newcommand{\zS}{\Sigma}

\newcommand{\deff}{\stackrel{\text{\tiny def}}{=}}

\newtheorem{theorem}{\textbf{Theorem}}

{\theorembodyfont{\rmfamily}}

\newtheorem{proposition}[theorem]{\textbf{Proposition}}
\newtheorem{remark}[theorem]{\textbf{Remark}}
\newtheorem{assumption}[theorem]{\textbf{Assumption}}

\newenvironment{proof}
{\begin{trivlist} \item [\hskip\labelsep\quad\textsl{Proof.}
\hspace{0.5 em}]}{\hfill$\tcm{\Box}$ \end{trivlist}}

\makeatletter
\renewcommand*{\@fnsymbol}[1]{\ensuremath{\ifcase#1\or a\or b\or *\or
   \mathsection\or \mathparagraph\or \|\or **\or \dagger\dagger
   \or \ddagger\ddagger \else\@ctrerr\fi}}
\makeatother 

\numberwithin{theorem}{section}

\numberwithin{equation}{section}

\allowdisplaybreaks

\author{
F. Demontis\footnote{Dipartimento di Matematica e Informatica, Universit\`a degli Studi di Cagliari, Viale Merello 92, 09121 Cagliari, Italy},~~
S. Lombardo\footnote{Department of Mathematics, Physics and Electrical Engineering, Northumbria University, Newcastle upon Tyne, NE1 8ST, United Kingdom} \footnote{Department of Mathematical Sciences, School of Science, Loughborough University, Loughborough, LE11 3TU, United Kingdom},\\ ~~
M. Sommacal\footnotemark[2] \footnote{Corresponding author. e-mail: \texttt{matteo.sommacal@northumbria.ac.uk}},~~
C. van der Mee\footnotemark[1],~~
and F. Vargiu\footnotemark[1]
}
\title{Effective Generation of Closed-form Soliton Solutions of the Continuous Classical Heisenberg Ferromagnet Equation}

\begin{document}
\date{}
\maketitle
\thispagestyle{plain}

\begin{abstract}
The non-topological, stationary and propagating, soliton solutions of the classical continuous Heisenberg ferromagnet equation are investigated. A general, rigorous formulation of the Inverse Scattering Transform for this equation is presented, under less restrictive conditions than the Schwartz class hypotheses and naturally incorporating the non-topological character of the solutions. Such formulation is based on a new triangular representation for the Jost solutions, which in turn allows an immediate computation of the asymptotic behaviour of the scattering data for large values of the spectral parameter, consistently improving on the existing theory. A new, general, explicit multi-soliton solution formula, amenable to computer algebra, is obtained by means of the matrix triplet method, producing all the soliton solutions (including breather-like and multipoles), and allowing their classification and description.


\end{abstract}

\section{Introduction}\label{sec:1}

Recently, a number of theoretical and mostly experimental advancements have ignited renewed interest toward the study of propagating (non-to\-po\-lo\-gi\-cal) magnetic-droplet soliton configurations in ferromagnetic materials at the nanometer length-scale \cite{LauShaw2011}, particularly in view of potential applications to spintronics, magnonics and other future, spin-based information storage and processing technologies \cite{BaderParkin2010}.

The experimental observation of solitons and solitary waves in ferromagnetic systems has proved challenging, mainly due to the dimensions of the length-scale at which these phenomena are expected to occur \cite{IvanovKosevich1977, KosevichIvanovKovalev1990}. Nano-contact spin-torque oscillators (NC-STO) have been predicted to act as soliton creators in ultra thin, two-dimensional magnetic film with strong uniaxial, perpendicular anisotropy \cite{HoeferSilvaKeller2010, BoTiCoFiMuMaSlAk2010}. These configurations, denominated (non-topological) magnetic-droplet solitons, had been studied in \cite{IvaSte1989, PieZak1998, IvZaYa2001} and more recently in \cite{HoeSom2012} and \cite{HoeSomSil2012}. Finally, the first generation (enucleation) of a magnetic-droplet soliton in an NC-STO device has been announced in 2013 \cite{MohseniSaniPerssonNguyenChung2013}, along with the observation of interesting dynamical properties, including droplet oscillatory motion, droplet ``spinning'' and droplet ``breather'' states. This breakthrough led to further theoretical and experimental investigation (see, for instance, \cite{MacBacKen2014, MoSaDuPeNgChPoMuIaEkAk2014, ChMoSaIaDuNgPoMuEkHoAk2014, MaBoHo2014, ChMoEkDuRaSaNgRaDuAk2015, BooHoe2015, ChEkIaMoSaBoHoDuAk2016}).

In \cite{IaDuBoMoChHoAk2014}, it has been shown how, as an extended magnetic thin film is reduced to a nano-wire with a nano-contact of fixed size at its center, the observed excited modes undergo transitions from a fully localized two-dimensional droplet into a two-dimensional droplet ``edge'' mode, and then a pulsating one-dimensional droplet, linking the study of low-dimensional droplet solitons to the recent experimental discoveries.

In view of this connection, we have launched an ambitious research programme aimed at completing the investigation of non-topological, localized solutions of the underlying mathematical model of ferromagnetism at the nanometer length-scale, namely the one-dimensional, continuous Landau-Lifshitz equa\-tion, both in the ab\-sen\-ce and in the pre\-sen\-ce of (uni\-axial and bi\-axial) ani\-so\-tro\-py, see \cite{LandauLifschitz1935, Aharoni}. For this equation -- or, more precisely, for this family of models -- we aim at obtaining closed-form expressions for all (including potentially new) soliton solutions, allowing their classification and the description of their interactions. In the present paper we focus on the classical, continuous Heisenberg ferromagnet chain equation (\textit{i.e.} the one-dimensional, isotropic Landau-Lifshitz equation), which is the simplest and most fundamental of the continuous, integrable models of ferromagnetism \cite{Lak, T, ZT, Fogedby1980, Fogedby1980b}.

Let
\begin{equation}\label{eq:magnetization}
\bm:\R\times\R\to\mathbb{S}^2\,,\quad
\bm(x,t)=\sum_{j=1}^{3}m_{j}(x,t)\,\be_{j}
\end{equation}
be the magnetization vector at position $x$ and time $t$, where the vectors $\be_{j}$, $j=1,2,3$, are the standard Cartesian basis vectors for $\R^{3}$, $\mathbb{S}^2$ is the unit sphere in $\R^3$ and then $\|\bm(x, t)\|=1$. The position $x$ is taken on the real line orientated as $\be_{1}$. Then, the Heisenberg ferromagnet equation reads (in non-dimensional form):
\begin{subequations}\label{eq:HF1}
\begin{equation}\label{eq:HF1a}
\bm_t=\bm\wedge\bm_{xx},
\end{equation}
on which we impose the asymptotic condition
\begin{equation}\label{eq:HF1b}
\bm(x, t)\to\be_3 \mbox{ as } x\to\pm\infty\,.
\end{equation}
\end{subequations}
Equation \eqref{eq:HF1a} is the well-known continuous limit of the (quantum) ferromagnetic Heisenberg chain in a constant field when the wavelength of the excited modes is larger than the lattice distance (see, for instance, \cite{Aharoni} for a detailed discussion, or \cite{Fogedby1980} for a quick derivation; the effects of the discreteness of the lattice on the classical continuum limit of the Heisenberg chain are discussed in \cite{LakPorDan1980}). We assume that the constant spin field of the ground state of the Heisenberg chain points is $\be_3$. Then the boundary condition \eqref{eq:HF1b} has been chosen in analogy to the boundary condition for the uni-axial Landau-Lifshitz equation with perpendicular anisotropy (easy-axis). Finally, we observe that in the right-hand side of \eqref{eq:HF1a} one can add a term of the form $h\,\bm\wedge\be_3$, $h\in\R$, which can be scaled out by a convenient change of variables (\textit{e.g.} see \cite{Fogedby1980, HoeSom2012}).

It is well known that \eqref{eq:HF1} is integrable (see, for instance, \cite{Fogedby1980} for a brief time-line of the early original results on integrability). Localized, propagating, solitary waves (as well as periodic wave train solutions) had been derived in \cite{NakSas1974, LakRuiTho1976, Tjon}. In \cite{Lak}, Lakshmanan proved that \eqref{eq:HF1} has an infinite number of constants of motion and showed that the associated energy and current densities can be given in terms of the solutions of the nonlinear Schr\"{o}dinger equation \cite{ZS}. In \cite{T}, Takhtajan showed that \eqref{eq:HF1} admits a Lax pair representation. Let us briefly recall here that, if $V$ is a $2\times 2$ invertible matrix depending on position $x\in\R$, time $t\in\R$, and a spectral parameter $\lambda$, then (see \cite{T}) the Lax pair $\left(\bA,\bB\right)$ associated to \eqref{eq:HF1} is given by:
\begin{align}\label{eq:pair}
\left\{\begin{array}{l}
V_x=\bA\,V=[i\zl(\bm\cdot\bzs)]\,V\\
\\
V_t=\bB\,V=[-2i\zl^2(\bm\cdot\bzs)-i\zl(\bm\wedge\bm_x\cdot\bzs)]\,V\,,
\end{array}\right.
\end{align}
where $\bzs$ is the column vector with entries the Pauli matrices
\begin{equation*}
\zs_1=\begin{pmatrix}0&1\\1&0\end{pmatrix},\qquad
\zs_2=\begin{pmatrix}0&-i\\i&0\end{pmatrix},\qquad\zs_3=\begin{pmatrix}1&0\\0&-1
\end{pmatrix}\,.
\end{equation*}
Of course, the knowledge of the Lax pair for \eqref{eq:HF1} assures that the Inverse Scattering Transform (IST) (see \cite{AblSeg, CdG, FT}) can be applied to solve the initial-value problem \cite{T, ZT},
\begin{align}\label{eq:initial}
\left\{\begin{array}{l}
\bm_t=\bm\wedge\bm_{xx}\\
\\
\bm(x,0)\quad\text{known}\,.
\end{array}\right.
\end{align}
In \cite{T}, the Marchenko equations and the time dependence of the scattering data are presented, as well as the one-soliton solution and the phase and centre-of-mass shifts for a two-soliton collision. In \cite{ZT}, a gauge equivalence between \eqref{eq:HF1} and the nonlinear Schr\"{o}dinger equation is proved to exist. In \cite{Fogedby1980}, extending the results in \cite{T}, a diagonal action-angle representation of \eqref{eq:HF1} is exhibited.

The aim of the present paper is twofold. The first goal is to present a new, more general, rigorous theory for the IST. In particular, the direct scattering problem is proved to be well-posed for potentials satisfying the following conditions, which will be assumed to be valid throughout the work:
\begin{assumption}\label{HP1}
As a function of the position, the matrix $\bm(x)\cdot\bzs$ has an almost everywhere existing derivative with respect to $x$ with entries in $L^1(\R)$. %(see \cite{BianGuoLing2014})
Thus $\bm(x)\cdot\bzs$ is bounded and continuous in $x\in\R$.
\end{assumption}
\begin{assumption}\label{HP2}
The inequality $m_3(x)>-1$ holds for all $x\in\R$.
\end{assumption}
These conditions are less restrictive than the usual (see \cite{FT}) Schwartz class hypotheses. Moreover, it is worth observing that, under the first Assumption \ref{HP1}, $\bm(x)$ is absolutely continuous for $x\in\R$; thus its {\it point-wise} values make sense and it makes mathematical sense to assume that, in addition, $m_3(x)>-1$ for each $x\in\R$. Moreover, unexpectedly and rather remarkably, Assumption \ref{HP2} automatically entails the non-topological character of the solutions (see Sec. 5 in \cite{KosevichIvanovKovalev1990}), which is otherwise verified \textit{a posteriori} (solutions for which the magnetization $\bm(x)$ maps to a curve on $\mathbb{S}^{2}$ that is closed and contractible by continuous deformations to the north pole are called \textit{non-topological}, whereas solutions which map to lines on $\mathbb{S}^{2}$ connecting the two poles are called \textit{topological}).

For potentials satisfying Assumptions \ref{HP1} and \ref{HP2} we establish the analyticity properties of eigenfunctions and scattering data. In order to derive these results we define a convenient set of Jost solutions (see Section \ref{sub:a}) which enables the study of their asymptotic behaviour at large $\zl$. Then, differently from \cite{BianGuoLing2014} (where the conditions in Assumption \ref{HP1} are used for developing the IST theory for \eqref{eq:HF1} exploiting the gauge equivalence to the nonlinear Schr\"{o}dinger equation and by solving the corresponding Riemann-Hilbert problem), in our treatment the inverse scattering problem is formulated directly in terms of the Marchenko integral equations. They are obtained by using a new triangular representation of the Jost solutions (see Propositions \ref{P2} and \ref{P3} in Section \ref{sec:2}) which differs substantially from the triangular representations in \cite{T} and \cite{ZT} (\textit{e.g.}, see formulae (13) and (17) in \cite{ZT}). In fact, the triangular representations introduced in \cite{T} and \cite{ZT} (and used in the literature thereafter) feature the spectral parameter $\lambda$ as a factor multiplying the integral of the kernels, and this results in a rather involuted computation of the asymptotic behaviour for large $\zl$ of the Jost solutions (and consequently also of the scattering data), requiring the equivalence between \eqref{eq:HF1} and the nonlinear Schr\"{o}dinger equation to be obtained (see \cite{FT}). On the contrary, in \eqref{eq:3.10} the said factor $\lambda$ does not appear, and this allows us to establish directly and straightforwardly the asymptotic behaviour of the scattering data.

A further, remarkable advantage of the new triangular representation of the Jost solutions, \eqref{eq:3.10} and \eqref{eq:3.12}, lies in the fact that they can be immediately generalized for the Landau-Lifshitz equation with easy-axis anisotropy (see \cite{BK, BKK, Mik, Zong}) -- and possibly with any kind of anisotropy -- allowing us to establish the analytical properties of the associated eigenfunctions.

More generally, we believe that the ideas used in the present paper to make the direct and inverse scattering theory of the first equation of \eqref{eq:pair} rigorous can be extended to the scattering operator associated to the Landau-Lifshitz equation, thus paving the way to the construction of explicit solutions for this model via the IST (both in the uniaxial and in the biaxial versions).

The second objective of this paper is to find a general, explicit multi-soliton solution formula for (\ref{eq:HF1}). This formula -- which differs from those obtained by means of the Darboux dressing method (\textit{e.g.}, see \cite{Wang2005, SalHas2009, BianGuoLing2014, ChenWan2014}), requiring one to treat the problem of inverting an $N\times N$ matrix featuring Jost solutions as its elements -- contains and allows an immediate classification of all the reflectionless solutions, irrespective of the number and the nature of the discrete eigenvalues in the spectrum, providing their direct physical interpretation (\textit{e.g.}, explicit expressions for the speed and precession frequency, as well as the location and time of the interactions, period of the oscillations for the breather-like solutions, separation of the maxima for the creation of entangled states, etc, see Section \ref{sec:4}). Indeed, by choosing in a proper way the parameters featured by this formula and naturally linked to the spectral data, we are able to generate explicit expressions for all the solutions already known in the literature \cite{NakSas1974, LakRuiTho1976, Tjon, Wang2005, SalHas2009, BianGuoLing2014, ChenWan2014}, and, in particular, general, explicit expressions for the breather-like and multipole solutions (see Section \ref{sec:4}). As for the latter, it is important to underline here that, in principle, the existence of multipole solutions for (\ref{eq:HF1}) might be inferred from the gauge equivalence to the nonlinear Schr\"{o}dinger equation as derived in \cite{ZT}. However, although multipole solutions can be obtained using the formula in Theorem 11 in \cite{BianGuoLing2014} (where multisoliton, multibreathers, and multipole solutions are collectively called high-order solitons), in the present article we derive and exploit a general, explicit expression for multiple-pole solitons (irrespective of the number and the order of the multiple poles), that is capable of providing their immediate classification and which does not need the computation of any auxiliary parameters. On a more general note, it is worth clarifying that, even if the Riemann-Hilbert problem for Zakharov-Shabat systems as well as the reflectionless solutions of the nonlinear Schr\"{o}dinger equation have been extensively investigated for a very long time (\textit{e.g.}, see \cite{AblSeg, CdG, FT, CORBOOK, APT, D}), nonetheless the gauge equivalence \cite{ZT} does not automatically entail that from there one can easily and immediately recover a general, explicit, multi-soliton solution formulae for (\ref{eq:HF1}) (see \cite{BianGuoLing2014}).

To obtain this result we will develop the matrix triplet method, already employed to solve exactly, in the reflectionless case, several other integrable equations (\textit{e.g.}, see \cite{AktosunMee2006, DM0, DM1, DM2, DM3, DM4}). The idea of this method is to represent the Marchenko kernel as $Ce^{-x\,A}B$, where $(A, B, C)$ is a suitable matrix triplet (\ref{eq:5.1b}), in such a way that the Marchenko integral equation can be solved explicitly via separation of variables. The solutions obtained in this way will not contain anything more complicated than matrix exponentials and solutions of Lyapunov equations \cite{Dym, CORBOOK}, hence can be ``unzipped'' into lengthy expressions containing elementary functions. Moreover, for the one-soliton solution we specialize the expression obtained by using this algebraic approach in terms of the physical parameters used to characterize the one-soliton solution, \textit{i.e.} the velocity $v$ along the $x$-axis and the precessional frequency $\omega$, thereby obtaining the physical interpretation of the discrete eigenvalues and the norming constants. Furthermore, starting from \cite{ZT} we rederive the existence of a gauge transformation between the solutions of the nonlinear Schr\"{o}dinger equation and the solutions of classical Heisenberg ferromagnet equation. We postpone to future investigation the effect of such transformation when both solutions are expressed in terms of the same triplet of matrices.

Closed form solutions of the Heisenberg equation can be generated by the matrix triplet method \cite{CORBOOK}, and by the Riemann-Hilbert method \cite{APT}. The matrix triplet method is explicit in terms of matrix exponentials and inverse matrices, where a proof of the existence of the matrix inverses is available in the literature \cite{DM0}. The solution formulas are amenable to using matrix algebra methods and can be (and have been) used to test the accuracy of numerical methods to solve integrable nonlinear evolution equations \cite{FermoMeeSeatzu2016}. Also, explicit solution formulas obtained by means of the matrix triplet method for the nonlinear Schr\"{o}dinger and modified Korteweg-de Vries equations have been verified by direct substitution, disregarding entirely the IST method to derive them \cite{DM0,DM1}. On the other hand, the Riemann-Hilbert method requires one to solve a system of linear equations to determine certain parameters featured by the solution (\textit{e.g.}, see page 33 in \cite{APT}).
%and the unique solvability of such system of linear equations has not always been proved in the literature (\textit{e.g.}, see Chapter 2 in \cite{APT}).

The paper is organized as follows. In Section \ref{sec:2} we study the analyticity of the Jost solutions and scattering data and determine their time-evolution. Furthermore, we formulate the inverse scattering problem in terms of the Marchenko integral equation. In Section \ref{sec:3}, combining the IST and the matrix triplet method, we get an explicit solution formula for \eqref{eq:HF1}. Finally, in Section \ref{sec:4} we exploit the solution formula to suggest a classification of all the (reflectionless) soliton solutions, including new breather-like and multipole solutions. In Appendix \ref{sec:A} we determine the Marchenko equations by using the triangular representation introduced in Section \ref{sec:2}, and in Appendix \ref{sec:B} we give further details of the derivation of the solution formula and we provide alternative (and more explicit) formulations of it.

\section{Direct and inverse scattering theory}\label{sec:2}
In this section we focus on the direct and inverse scattering theory associated to the first of equation \eqref{eq:pair}. In particular, we study the analyticity properties and the asymptotic behaviour at large $\zl$ for the Jost solutions and the scattering data, and formulate the inverse scattering problem in terms of the Marchenko integral equations. When treating the direct and inverse scattering in Subsections \ref{sub:a}, \ref{sub:b}, and \ref{sub:c} we disregard the time variable (\textit{e.g.} $\bm(x,t)$ will be considered as a function of $x$ only and represented as $\bm(x)$). Time will be subsequently reintroduced starting from Subsection \ref{sub:d}.

\subsection{Jost solutions}\label{sub:a}
The main purpose of this subsection is the study of the asymptotic behaviour of the Jost solutions (see Theorem \ref{theo1} below). To do so we represent them by using a new triangular representation (see Proposition \ref{P2}) different to the one proposed in \cite{T, ZT}.

Let us define the {\it Jost matrices} $\Psi(x,\zl)$ and $\Phi(x,\zl)$ as those solutions of the linear eigenvalue problems $\Psi_x=\bA\,\Psi$ and $\Phi_x=\bA\,\Phi$, where $\bA$ is the Lax matrix defined in $\eqref{eq:pair}$, and satisfying the asymptotic conditions:
\begin{subequations}\label{eq:1.4}
\begin{alignat}{3}
\Psi(x,\zl)&=\begin{pmatrix}\psi(x,\zl)&\overline{\psi}(x,\zl)
\end{pmatrix}=e^{i\zl x\zs_3}[I_2+o(1)],&\qquad&x\to+\infty,\label{eq:1.4a}\\
\Phi(x,\zl)&=\begin{pmatrix}\overline{\phi}(x,\zl)&\phi(x,\zl)
\end{pmatrix}=e^{i\zl x\zs_3}[I_2+o(1)],&\qquad&x\to-\infty.\label{eq:1.4b}
\end{alignat}
\noindent with $I_2$ being the $2\times2$ identity matrix. The columns $\psi(x,\zl)$, $\overline{\psi}(x,\zl)$, $\overline{\phi}(x,\zl)$, and $\phi(x,\zl)$
%{\color{blue}
%\begin{align}\label{eq:1.4d}
%\psi(x,\lambda)=\begin{pmatrix}\psi^{up}(x,\zl)\\ \psi^{dn}(x,\zl)\end{pmatrix}\,,\,
%\overline{\psi}(x,\lambda)=\begin{pmatrix}\overline{\psi}^{up}(x,\zl)\\ \overline{\psi}^{dn}(x,\zl)\end{pmatrix}\,,\,
%\phi(x,\lambda)=\begin{pmatrix}\phi^{up}(x,\zl)\\ \phi^{dn}(x,\zl)\end{pmatrix}\,,\,
%\overline{\phi}(x,\lambda)=\begin{pmatrix}\overline{\phi}^{up}(x,\zl)\\ \overline{\phi}^{dn}(x,\zl)\end{pmatrix}\,,
%\end{align}}
are called {\it Jost functions}. As a note of caution, we warn the reader that here and thereafter the bar over a symbol does not indicate complex conjugation, which instead is indicated by means of an asterisk in superscript. In the sequel, we also use the following notations:
\begin{align}
\Psi(x,\zl)=
\begin{pmatrix}\psi^{up}(x,\zl)&\overline{\psi}^{up}(x,\zl)\\
\psi^{dn}(x,\zl)&\overline{\psi}^{dn}(x,\zl)
\end{pmatrix},\quad
\Phi(x,\zl)&=\begin{pmatrix}\overline{\phi}^{up}(x,\zl)&\phi^{up}(x,\zl)\\
\overline{\phi}^{dn}(x,\zl)&\phi^{dn}(x,\zl)
\end{pmatrix}\,.\label{eq:1.4c}
\end{align}
\end{subequations}
Then the differential equations $\Psi_x=\bA\,\Psi$ and $\Phi_x=\bA\,\Phi$ (cf. with \eqref{eq:pair}) can be written as
\begin{subequations}\label{eq:1.5}
\begin{align}
\Psi_x&=i\zl(\bm\cdot\bzs)\,\Psi,\label{eq:1.5a}\\
\Phi_x&=i\zl(\bm\cdot\bzs)\,\Phi.\label{eq:1.5b}
\end{align}
\end{subequations}
It is then easily verified that $\Psi(x,\zl)$ and $\Phi(x,\zl)$ belong to the group $SU(2)$. Indeed, any square matrix $U(x)$ solution to the differential system $U_x=W(x)\,U$, where $W(x)$ is skew-Hermitian and has zero trace, has $U^\dagger\,U$ and $\det(U)$ independent of $x\in\R$. Here and thereafter the dagger denotes the complex conjugate transpose. As a result,
\begin{subequations}\label{eq:1.6}
\begin{alignat}{3}
\Psi_{11}(x,\zl)^*&=\Psi_{22}(x,\zl),&\qquad
\Psi_{12}(x,\zl)^*&=-\Psi_{21}(x,\zl),\label{eq:1.6a}\\
\Phi_{11}(x,\zl)^*&=\Phi_{22}(x,\zl),&\qquad
\Phi_{12}(x,\zl)^*&=-\Phi_{21}(x,\zl).\label{eq:1.6b}
\end{alignat}
\end{subequations}
Since the two Jost matrices are both solutions to the same first order linear homogeneous differential system, there exists a so-called {\it transition matrix} $\bT(\zl)$, depending on $\zl$ and belonging to $SU(2)$, such that
\begin{equation}\label{eq:1.7}
\Psi(x,\zl)=\Phi(x,\zl)\,\bT(\zl),\qquad\zl\in\R.
\end{equation}
For $\zl\in\R$, we have
\begin{equation}\label{eq:transmatrix}
\bT(\zl)=\begin{pmatrix}\tau(\zl)&-\varrho(\zl)\\
\varrho(\zl)^*&\tau(\zl)^*\end{pmatrix}\,,
\end{equation}
where $|\tau(\zl)|^2+|\varrho(\zl)|^2=1$. We assume that $\tau(\zl)\neq0$ for each $\zl\in\R$, \textit{i.e.} we assume that no spectral singularities exist.

In order to formulate the Riemann-Hilbert problem we need to establish the analyticity properties as well as the asymptotic behaviour at large $\zl$ for the Jost solutions and for the coefficients $\tau(\zl)$ and $\varrho(\zl)$. To get these results, let us put $\bm^0=\bm-\be_3$. We can convert the differential systems \eqref{eq:1.5} with corresponding asymptotic conditions \eqref{eq:1.4} into the Volterra integral equations
\begin{subequations}\label{eq:2.1}
\begin{align}
&\Psi(x,\zl)=e^{i\zl x\zs_3}
-i\zl\int_x^\infty \mathrm{d}\xi\,e^{-i\zl(\xi-x)\zs_3}\,(\bm^0(\xi)\cdot\bzs)\,
\Psi(\xi,\zl),\label{eq:2.1a}\\
&\Phi(x,\zl)=e^{i\zl x\zs_3}
+i\zl\int_{-\infty}^x \mathrm{d}\xi\,e^{i\zl(x-\xi)\zs_3}\,(\bm^0(\xi)\cdot\bzs)\,
\Phi(\xi,\zl).\label{eq:2.1b}
\end{align}
\end{subequations}
As a result of Gronwall's inequality (see Appendix of \cite{DBVV}) we get for $(x,\zl)\in\R^2$
\begin{subequations}\label{eq:2.2}
\begin{align}
\|\Psi(x,\zl)\|&\le\exp\left(|\zl|\int_x^\infty \mathrm{d}\xi\,\|\bm^0(\xi)\|\right),
\label{eq:2.2a}\\
\|\Phi(x,\zl)\|&\le\exp\left(|\zl|\int_{-\infty}^x \mathrm{d}\xi\,\|\bm^0(\xi)\|\right),
\label{eq:2.2b}
\end{align}
\end{subequations}
where we have to assume that $\bm^0(x)=\bm(x)-\be_3$ has its entries in $L^1(\R)$.

Here and thereafter, let $\C^{+}$ and $\C^{-}$ denote the upper and lower half-planes, respectively, whereas $\overline{\C}^{+}=\C^+\cup\R$ and $\overline{\C}^{-}=\C^-\cup\R$ denote the closure of $\C^{+}$ and $\C^{-}$, respectively.
We can easily prove the following
\begin{proposition}\label{P1} Suppose that $\bm^0(x)=\bm(x)-\be_3$ has its entries in $L^1(\R)$. Then, the so-called Faddeev functions
%$$
%e^{-i\zl x}\,\psi^\up(x,\zl),\, e^{-i\zl x}\,\psi^\dn(x,\zl),\, e^{i\zl x}\,\phi^\up(x,\zl), \mbox{ and } e^{i\zl x}\,\phi^\dn(x,\zl)
%$$
$e^{-i\zl x}\,\psi(x,\zl)$ and $e^{i\zl x}\,\phi(x,\zl)$
are analytic in $\zl\in\C^+$ and continuous in $\zl\in\overline{\C}^{+}$, while the Faddeev functions
%$$
%e^{i\zl x}\,\overline{\psi}^\up(x,\zl),\, e^{i\zl x}\,\overline{\psi}^\dn(x,\zl),\, e^{-i\zl x}\,\overline{\phi}^\up(x,\zl), \mbox{ and } e^{-i\zl x}\,\overline{\phi}^\dn(x,\zl)
%$$
$e^{i\zl x}\,\overline{\psi}(x,\zl)$ and $e^{-i\zl x}\,\overline{\phi}(x,\zl)$
are analytic in $\zl\in\C^-$ and continuous in $\zl\in\overline{\C}^{-}$.
\end{proposition}
\begin{proof}
Writing the Volterra integral equations (\ref{eq:2.1}) for the separable Jost functions and applying Gronwall's inequality we get
\begin{subequations}\label{eq:2.3}
\begin{equation}\label{eq:2.3a}
\left\|e^{-i\zl x}\,\psi(x,\zl)\right\| \leq \exp\left(|\lambda|\,\int_{x}^{\infty} \mathrm{d}\xi\,\left\|\bm(\xi)\cdot\bzs-\sigma_{3}\right\|\right)
\end{equation}
uniformly in $(\zl, x)$ for $\zl\in\overline{\C}^{+}$ and $x\in[x_0,+\infty)$ for all $x_{0}\in\mathbb{R}$, and
\begin{equation}\label{eq:2.3b}
\left\|e^{i\zl x}\,\overline{\psi}(x,\zl)\right\| \leq \exp\left(|\lambda|\,\int_{x}^{\infty} \mathrm{d}\xi\,\left\|\bm(\xi)\cdot\bzs-\sigma_{3}\right\|\right)
\end{equation}
\end{subequations}
uniformly in $(\zl, x)$ for $\zl\in\overline{\C}^{-}$ and $x\in[x_0,+\infty)$ for all $x_{0}\in\mathbb{R}$, thus proving the continuity (in $\zl\in\overline{\C}^{+}$) and analyticity (in $\zl\in\C^{+}$) of $e^{-i\zl x}\,\psi(x,\zl)$, and similarly for $e^{i\zl x}\,\overline{\psi}(x,\zl)$. The proof for the other Faddeev functions is analogous.
\end{proof}
Taking the limit of the columns of (\ref{eq:2.1}) as $x\to-\infty$ we get
\begin{subequations}\label{eq:2.5}
\begin{align}
&\tau(\zl)=1-i\zl\int_{-\infty}^\infty \mathrm{d}\xi\,\Big[m_{0}(\xi)
\,e^{-i\zl\xi}\,\psi^\up(\xi,\zl)+m_{-}(\xi)\,
e^{-i\zl\xi}\,\psi^\dn(\xi,\zl)\Big],\label{eq:2.5a}\\
&\tau(\zl^*)^*=1+i\zl\int_{-\infty}^\infty \mathrm{d}\xi\,\left[m_{0}(\xi)\,
e^{i\zl\xi}\,\overline{\psi}^\dn(\xi,\zl)-m_{+}(\xi)\,
e^{i\zl\xi}\,\overline{\psi}^\up(\xi,\zl)\right].\label{eq:2.5b}
\end{align}
\end{subequations}
Thus $\tau(\zl)$ is continuous in $\zl\in\overline{\C^+}$, is analytic in $\zl\in\C^+$, and satisfies $\tau(0)=1$. In the same way we get
\begin{subequations}\label{eq:2.6}
\begin{align}
&\varrho(\zl)^*=i\zl\int_{-\infty}^\infty \mathrm{d}\xi\,\left[
m_{0}(\xi)\,e^{i\zl\xi}\,\psi^\dn(\xi,\zl)-m_{+}(\xi)\,
e^{i\zl\xi}\,\psi^\up(\xi,\zl)\right],\label{eq:2.6a}\\
&\varrho(\zl)=i\zl\int_{-\infty}^\infty \mathrm{d}\xi\,\left[
m_{0}(\xi)\,e^{-i\zl\xi}\,\overline{\psi}^\up(\xi,\zl)+m_{-}(\xi)\,
e^{-i\zl\xi}\,\overline{\psi}^\dn(\xi,\zl)\right].\label{eq:2.6b}
\end{align}
\end{subequations}
where $\varrho(\zl)$ is continuous for $\zl\in\R$, and $\varrho(\zl)/\zl$ vanishes as $\zl\to\pm\infty$. Thus $\varrho(0)=0$ and $\varrho_\zl(0)$ exists.

From \eqref{eq:2.5} and \eqref{eq:2.6} it is clear that no information is available on their asymptotics as $\zl\to\infty$. In order to get such information let us derive a different set of Volterra integral equations. To do so we need Assumptions \ref{HP1} and \ref{HP2}, namely that $\bm(x)\cdot\bzs$ has an almost everywhere existing derivative $\bm^\prime(x)\cdot\bzs$ with respect to $x$ which has its entries in $L^1(\R)$, and that $m_3(x)>-1$ for all $x\in\R$. Here and thereafter the prime indicates the total derivative with respect to the spatial variable $x$.

Under Assumption \ref{HP1}, we can apply integration by parts to \eqref{eq:2.1a} obtaining
\begin{footnotesize}
\begin{align*}
\Psi(x,\zl)&=e^{i\zl x\zs_3}+\left[e^{-i\zl(\xi-x)\zs_3}\zs_3
(\bm^0(\xi)\cdot\bzs)\Psi(\xi,\zl)\right]_{\xi=x}^\infty\\
&\quad-\int_x^\infty \mathrm{d}\xi\,e^{-i\zl(\xi-x)\zs_3}\zs_3\left[(\bm^\prime(\xi)\cdot\bzs)
\Psi(\xi,\zl)+(\bm^0(\xi)\cdot\bzs)\frac{\partial\Psi}{\partial\xi}(\xi,\zl)
\right]\\ &=e^{i\zl x\zs_3}-\zs_3(\bm^0(x)\cdot\bzs)\Psi(x,\zl)\\
&\quad-\int_x^\infty \mathrm{d}\xi\,e^{-i\zl(\xi-x)\zs_3}\zs_3\Big[(\bm^\prime(\xi)\cdot\bzs)
+i\zl(\bm^0(\xi)\cdot\bzs)(\bm(\xi)\cdot\bzs)\Big]\Psi(\xi,\zl),
\end{align*}
\end{footnotesize}
where we have used \eqref{eq:1.5a}. Observing that
$$
I_2+\zs_3(\bm^0(x)\cdot\bzs)=\zs_3(\bm(x)\cdot\bzs)\in SU(2)\,,
$$
we obtain
\begin{align}
\zs_3(\bm(x)\cdot\bzs)&\Psi(x,\zl)=e^{i\zl x\zs_3}-\int_x^\infty \mathrm{d}\xi\,
e^{-i\zl(\xi-x)\zs_3}\zs_3(\bm^\prime(\xi)\cdot\bzs)\Psi(\xi,\zl)\nonumber\\
&-i\zl\int_x^\infty \mathrm{d}\xi\,e^{-i\zl(\xi-x)\zs_3}\zs_3(\bm^0(\xi)\cdot\bzs)
(\bm(\xi)\cdot\bzs)\Psi(\xi,\zl).\label{eq:2.7}
\end{align}
It is easy to verify that
\begin{equation}\label{eq:2.8}
(\bm^0\cdot\bzs)(\bm\cdot\bzs)=(\bm\cdot\bzs)^2-\zs_3(\bm\cdot\bzs)=I_2
-\zs_3(\bm\cdot\bzs)=-\zs_3(\bm^0\cdot\bzs).
\end{equation}
We employ \eqref{eq:2.8} to write \eqref{eq:2.7} in the equivalent form
\begin{align*}
\zs_3(\bm(x)\cdot\bzs)\Psi(x,\zl)\,=&\,e^{i\zl x\zs_3}-\int_x^\infty \mathrm{d}\xi\,
e^{-i\zl(\xi-x)\zs_3}\zs_3(\bm^\prime(\xi)\cdot\bzs)\Psi(\xi,\zl)\nonumber\\
&+i\zl\int_x^\infty \mathrm{d}\xi\,e^{-i\zl(\xi-x)\zs_3}(\bm^0(\xi)\cdot\bzs)
\Psi(\xi,\zl).
\end{align*}
Taking half the sum of \eqref{eq:2.1a} and the latter equation we get
\begin{equation}\label{eq:2.9}
D(x)\Psi(x,\zl)=e^{i\zl x\zs_3}-\int_x^\infty \mathrm{d}\xi\,
e^{-i\zl(\xi-x)\zs_3}D^\prime(\xi)\Psi(\xi,\zl),
\end{equation}
where we define
\begin{equation}\label{eq:2.10}
D(x)=\tfrac{1}{2}\Big[
I_2+\zs_3(\bm(x)\cdot\bzs)\Big]=\tfrac{1}{2}
\begin{pmatrix}1+m_3(x)&m_-(x)\\-m_+(x)&1+m_3(x)\end{pmatrix},
\end{equation}
which is a matrix of determinant $\tfrac{1}{2}(1+m_3(x))$. Under Assumption \ref{HP2}, the matrix $D(x)$ is invertible and its inverse
$$
D(x)^{-1}=\tfrac{1}{1+m_3}\left(\begin{smallmatrix}1+m_3&-m_-\\m_+&1+m_3 \end{smallmatrix}\right)\,.
$$
has norm $(1+m_3)^{-1/2}$. Thus, $D(x)$ and $D(x)^{-1}$ are bounded in $x\in\R$. Note that $D(x)\to I_2$ as $x\to\pm\infty$. We may therefore apply Gronwall's inequality to \eqref{eq:2.9} and find that
\begin{equation}\label{eq:inequalityPsi}
\|\Psi(x,\zl)\|\le\frac{1}{\sqrt{1+m_3(x)}}\exp\left[\tfrac{1}{2\,\sqrt{1+m_{3}(x)}}
\int_x^\infty \mathrm{d}\xi\,\|(\bm^\prime(\xi)\cdot\bzs)\|\right]\,.
\end{equation}

\begin{remark} In the same way and under Assumptions \ref{HP1} and \ref{HP2}, adapting the procedure presented above to the Jost matrix $\Phi(x,\zl)$, we get
\begin{equation}\label{eq:2.11}
D(x)\Phi(x,\zl)=e^{i\zl x\zs_3}+\int_{-\infty}^x \mathrm{d}\xi\,e^{i\zl(x-\xi)\zs_3}
D^\prime(\xi)\Phi(\xi,\zl).
\end{equation}
We may therefore apply Gronwall's inequality to \eqref{eq:2.11} to obtain
\begin{equation}\label{eq:inequalityPhi}
\|\Phi(x,\zl)\|\le\frac{1}{\sqrt{1+m_3(x)}}\exp\left[\tfrac{1}{2\,\sqrt{1+m_{3}(x)}}
\int_{-\infty}^x \mathrm{d}\xi\,\|(\bm^\prime(\xi)\cdot\bzs)\|\right].
\end{equation}
\end{remark}
Observe that inequalities (\ref{eq:inequalityPsi}) and (\ref{eq:inequalityPhi}) improve inequalities (2.7) as they provide bounds for the Jost functions also in the limit as $\lambda\rightarrow\infty$.

Equations \eqref{eq:2.9} and \eqref{eq:2.11} allow us to prove that the analyticity and the continuity properties of the Jost solutions extend to the closed upper and lower half-planes. In other words, the Jost solutions and the coefficient $\tau(\zl)$ have a finite limit as $\zl\to\infty$ from within the closure of its half-plane of analyticity (whereas, focussing on reflectionless solutions, we will eventually consider $\varrho(\zl)=0$, see Section \ref{sub:reconstruction}). In order to prove these results we need to find a ``suitable'' triangular representation for the Jost solutions. We have the following:
\begin{proposition}\label{P2}
There exists an {\it auxiliary matrix function} $\bK^{up}(x,y)$ such that
\begin{equation}\label{eq:3.10}
\Psi(x,\zl)=\bH^{up}(x)e^{i\zl x\zs_3}+\int_x^\infty \mathrm{d}\xi\,\bK^{up}(x,\xi)
e^{i\zl\xi\zs_3},
\end{equation}
where $\bH^{up}(x)$ is a matrix function satisfying $\bH^{up}(x)=\zs_2\,\bH^{up}(x)^*\,\zs_2$ and
$\bH^{up}(x)\to I_2$ as $x\to+\infty$, and $\int_x^\infty \mathrm{d}\xi\,\|\bK^{up}(x,\xi)\|$
converges uniformly in $x\in\R$.
\end{proposition}
Before giving the proof let us remark that equations \eqref{eq:3.10} play here the role that in \cite{T,ZT} is attributed to equation (13) in \cite{ZT}.
\begin{proof}
First of all, we em\-ploy the sym\-metry re\-la\-tion
$$
\Psi(x,\zl)^*=\zs_2\,\Psi(x,\zl)\,\zs_2
$$
to derive the structure of the auxiliary matrix function
\begin{subequations}\label{eq:3.11}
\begin{equation}\label{eq:3.11a}
\bK^{up}(x,y)=\begin{pmatrix}K^{up}_1(x,y)&-K^{up}_2(x,y)^*\\ K^{up}_2(x,y)
&K^{up}_1(x,y)^*\end{pmatrix},
\end{equation}
where $K^{up}_1(x,y)$ and $K^{up}_2(x,y)$ are scalar functions. Because of the same symmetry, we also have
\begin{equation}\label{eq:3.11b}
\bH^{up}(x)=\begin{pmatrix}H^{up}_1(x)&-H^{up}_2(x)^*\\ H^{up}_2(x)
&H^{up}_1(x)^*\end{pmatrix},
\end{equation}
\end{subequations}
where $H^{up}_1(x)$ and $H^{up}_2(x)$ are scalar functions. Substituting \eqref{eq:3.10} into \eqref{eq:2.9} we get
\begin{align*}
&D(x)\left\{\bH^{up}(x)e^{i\zl x\zs_3}+\int_x^\infty \mathrm{d}\xi\,\bK^{up}(x,\xi)
e^{i\zl\xi\zs_3}\right\}\\
&\quad=e^{i\zl x\zs_3}-\int_x^\infty \mathrm{d}\xi\,
e^{-i\zl(\xi-x)\zs_3}D^\prime(\xi)\bH^{up}(\xi)e^{i\zl\xi\zs_3}\\
&\qquad-\int_x^\infty \mathrm{d}\xi\,e^{-i\zl(\xi-\zeta)\zs_3}D^\prime(\xi)\int_\xi^\infty \mathrm{d}\zeta\,
\bK^{up}(\xi,\zeta)e^{i\zl\zeta\zs_3}.
\end{align*}
Letting
$$
\bH_e(x)=\tfrac{1}{2}\,\Big\{D^\prime(x)\,\bH^{up}(x)+\zs_3\,D^\prime(x)\,\bH^{up}(x)\,\zs_3\Big\}\,,
$$
and splitting the integrand of the last integral term into diagonal and off-diagonal parts, we obtain
\begin{footnotesize}
\begin{align*}
&D(x)\int_x^\infty \mathrm{d}\xi\,\bK^{up}(x,\xi)\,e^{i\zl\xi\zs_3}
=\Big\{I_2-\bH_e(x)-D(x)\bH^{up}(x)\Big\}\,e^{i\zl x\zs_3}\nonumber\\
&-\tfrac{1}{4}\,\int_x^\infty \mathrm{d}\xi\,\left[D^\prime\left(\tfrac{\xi+x}{2}\right)
\bH^{up}\left(\tfrac{\xi+x}{2}\right)-\zs_3 D^\prime\left(\tfrac{\xi+x}{2}\right)
\bH^{up}\left(\tfrac{\xi+x}{2}\right)\zs_3\right]e^{i\zl\xi\zs_3}\nonumber\\
&-\tfrac{1}{2}\, \int_x^\infty \mathrm{d}\zeta\, \int_x^\infty \mathrm{d}\xi\, \Big[D^\prime(\xi)\,
\bK^{up}(\xi,\zeta+\xi-x)+\zs_3\,D^\prime(\xi)\,\bK^{up}(\xi,\zeta+\xi-x)\zs_3\Big]\,e^{i\zl\zeta\zs_3}
\nonumber\\
&-\tfrac{1}{2}\, \int_x^\infty \mathrm{d}\zeta\, \int_x^{\tfrac{\zeta+x}{2}}\mathrm{d}\xi\, \Big[
D^\prime(\xi)\,\bK^{up}(\xi,\zeta+x-\xi)-\zs_3\,D^\prime(\xi)\,\bK^{up}(\xi,\zeta+x-\xi)\zs_3\Big]\,
e^{i\zl\zeta\zs_3}.
\end{align*}
\end{footnotesize}
Choosing $\bH^{up}(x)$ such that the nonintegral terms in the right-hand side cancel each other and stripping off the Fourier transform, we get for $y\ge x$ the integral equation
\begin{footnotesize}
\begin{align}
D(x)\bK^{up}(x,y)
=&-\tfrac{1}{4}\Big[D^\prime\left(\tfrac{x+y}{2}\right)\,\bH^{up}\left(\tfrac{x+y}{2}\right)
-\zs_3\, D^\prime\left(\tfrac{x+y}{2}\right)\,\bH^{up}\left(\tfrac{x+y}{2}\right)\,\zs_3\Big]\nonumber\\
&-\tfrac{1}{2}\, \int_x^\infty \mathrm{d}\xi\, \Big[D^\prime(\xi)\,\bK^{up}(\xi,y+\xi-x)
+\zs_3\,D^\prime(\xi)\,\bK^{up}(\xi,y+\xi-x)\zs_3\Big]\nonumber\\
&-\tfrac{1}{2}\, \int_x^{\tfrac{x+y}{2}}\mathrm{d}\xi\, \left[D^\prime(\xi)\,
\bK^{up}(\xi,y+x-\xi)-\zs_3\,D^\prime(\xi)\,\bK^{up}(\xi,y+x-\xi)\,\zs_3\right].\label{eq:3.14}
\end{align}
\end{footnotesize}
Using Gronwall's inequality it can be proved in a standard way \cite{APT, CORBOOK, DBVV} that \eqref{eq:3.14} has a unique solution $\bK^{up}(x,y)$ satisfying
$$
\int_x^\infty \mathrm{d}\xi\,\|\bK^{up}(x,\xi)\|\le\left(\frac{1}{2\sqrt{1+m_3(x)}}
\int_x^\infty \mathrm{d}\xi\,\|D^\prime(\xi)\bH^{up}(\xi)\|\right)e^{2\int_z^\infty \mathrm{d}\xi\,
\|\bm^\prime(\xi)\|}\,.
$$
\end{proof}
Analogously we have the following
\begin{proposition}\label{P3}
There exists an {\it auxiliary matrix function} $\bK^{dn}(x,y)$ such that
\begin{equation}\label{eq:3.12}
\Phi(x,\zl)=\bH^{dn}(x)e^{i\zl x\zs_3}+\int_{-\infty}^x \mathrm{d}\xi\,
\bK^{dn}(x,\xi)e^{i\zl\xi\zs_3},
\end{equation}
where ${\bH^{dn}}(x)$ is a matrix function satisfying
${\bH^{dn}}(x)=\zs_2\,{\bH^{dn}}(x)^*\,\zs_2$ and ${\bH^{dn}}(x)\to I_2$ as
$x\to-\infty$, and $\int_{-\infty}^x \mathrm{d}\xi\,\|\bK^{dn}(x,\xi)\|$ converges uniformly in
$x\in\R$.
\end{proposition}
We omit the details of the proof because it is analogous to the proof of Proposition \ref{P2}. We only remark that, because of the symmetry relation $\Phi(x,\zl)^*=\zs_2\,\Phi(x,\zl)\,\zs_2$, the auxiliary matrix has the following structure
\begin{subequations}\label{eq:3.13}
\begin{equation}\label{eq:3.13a}
\bK^{dn}(x,y)=\begin{pmatrix}K^{dn}_1(x,y)^*&K^{dn}_2(x,y)\\
-K^{dn}_2(x,y)^*&K^{dn}_1(x,y)\end{pmatrix},
\end{equation}
where $K^{dn}_1(x,y)$ and $K^{dn}_2(x,y)$ are scalar functions. Because of the same symmetry, we also have
\begin{equation}\label{eq:3.13b}
\bH^{dn}(x)=\begin{pmatrix}H^{dn}_1(x)^*&H^{dn}_2(x)\\
-H^{dn}_2(x)^*&H^{dn}_1(x)\end{pmatrix},
\end{equation}
\end{subequations}
where $H^{dn}_1(x)$ and $H^{dn}_2(x)$ are scalar functions.
Finally we can prove the following:
\begin{theorem}\label{theo1}
Suppose that
\begin{itemize}\vspace{-0.3\baselineskip}
\item[1)] $\bm^0(x)=\bm(x)-\be_3$ has its entries in $L^1(\R)$;\vspace{-0.5\baselineskip}
\item[2)] $\bm(x)\cdot\bzs$ has an almost everywhere existing derivative $\bm^\prime(x)\cdot\bzs$ with respect to $x$ which has its entries in $L^1(\R)$ (Assumption \ref{HP1});\vspace{-0.5\baselineskip}
\item[3)] $m_3(x)>-1$ for all $x\in\R$ (Assumption \ref{HP2}).
\end{itemize}\vspace{-0.1\baselineskip}
Then the functions
$e^{-i\zl x}\,\psi(x,\zl)$ and $e^{i\zl x}\,\phi(x,\zl)$,
which are analytic in $\zl\in\C^+$ and continuous in $\zl\in\C^+\cup\R$ (see Proposition \ref{P1}), have a finite limit as $\zl\to\infty$ from within the closure of $\C^+$. Analogously, the functions
$e^{i\zl x}\,\overline{\psi}(x,\zl)$ and $e^{-i\zl x}\,\overline{\phi}(x,\zl)$,
which are analytic in $\zl\in\C^-$ and continuous in $\zl\in\C^-\cup\R$ (see Proposition \ref{P1}), admit a finite limit as $\zl\to\infty$ from within the closure of $\C^-$. Moreover, the coefficient $\tau(\zl)$ has a finite limit when $\zl\to \infty$ from within $\overline{\C}^+$ while $\varrho(\zl)$ may not admit analytical continuation outside the real line and $\varrho(\zl)\to 0$ when $\zl\to\pm\infty$.
\end{theorem}
\begin{proof}
We give the proof only for the Faddeev functions $e^{-i\zl x}\psi^\up(x,\zl)$, $e^{-i\zl x}\psi^\dn(x,\zl)$, $e^{i\zl x}\phi^\up(x,\zl)$ and $e^{i\zl x}\phi^\dn(x,\zl)$, because the proof for the other Faddeev functions proceeds in the same way. Under the hypothesis $2)$, the Jost matrices satisfy equations \eqref{eq:2.9} and \eqref{eq:2.11}. Separating the two columns of $\Psi(x,\zl)$ and $\Phi(x,\zl)$ in \eqref{eq:2.9} and \eqref{eq:2.11}, we have
\begin{subequations}\label{eq:2.12}
\begin{align}
\left[D(x)e^{-i\zl x}\psi(x,\zl)\right]^\up&=1-\int_x^\infty \mathrm{d}\xi\,
\left[D^\prime(\xi)e^{-i\zl\xi}\psi(\xi,\zl)\right]^\up,\label{eq:2.12a}\\
\left[D(x)e^{-i\zl x}\psi(x,\zl)\right]^\dn&=-\int_x^\infty \mathrm{d}\xi\,
e^{2i\zl(\xi-x)}
\left[D^\prime(\xi)e^{-i\zl\xi}\psi(\xi,\zl)\right]^\dn,\label{eq:2.12b}
%\left[D(x)e^{i\zl x\overline{\psi}(x,\zl)\right]^\up
%&=-\int_x^\infty \mathrm{d}\xi\,e^{-2i\zl(\xi-x)}\left[D^\prime(\xi)
%e^{i\zl\xi}\overline{\psi}(\xi,\zl)\right]^\up,\label{eq:2.12c}\\
%\left[D(x)e^{i\zl x}\overline{\psi}(x,\zl)\right]^\dn
%&=1-\int_x^\infty \mathrm{d}\xi\,\left[D^\prime(\xi)
%e^{i\zl\xi}\overline{\psi}(\xi,\zl)\right]^\dn,\label{eq:2.12d}
\end{align}
\end{subequations}
as well as
\begin{subequations}\label{eq:2.13}
\begin{align}
%\left[D(x)e^{-i\zl x}\overline{\phi}(x,\zl)\right]^\up
%&=1+\int_{-\infty}^x \mathrm{d}\xi\,\left[D^\prime(\xi)
%e^{-i\zl\xi}\overline{\phi}(\xi,\zl)\right]^\up,\label{eq:2.13a}\\
%\left[D(x)e^{-i\zl x}\overline{\phi}(x,\zl)\right]^\dn
%&=\int_{-\infty}^x \mathrm{d}\xi\,e^{-2i\zl(x-\xi)}\left[D^\prime(\xi)
%e^{-i\zl\xi}\overline{\phi}(\xi,\zl)\right]^\dn,\label{eq:2.13b}\\
\left[D(x)e^{i\zl x}\phi(x,\zl)\right]^\up&=\int_{-\infty}^x \mathrm{d}\xi\,
e^{2i\zl(x-\xi)}
\left[D^\prime(\xi)e^{i\zl\xi}\phi(\xi,\zl)\right]^\up,\label{eq:2.13c}\\
\left[D(x)e^{i\zl x}\phi(x,\zl)\right]^\dn&=1+\int_{-\infty}^x \mathrm{d}\xi\,
\left[D^\prime(\xi)e^{i\zl\xi}\phi(\xi,\zl)\right]^\dn.\label{eq:2.13d}
\end{align}
\end{subequations}
Taking the limit as $x\to-\infty$ and using \eqref{eq:1.7} we get
\begin{subequations}\label{eq:2.14}
\begin{align}
\tau(\zl)&=1-\int_{-\infty}^\infty \mathrm{d}\xi\,
\left[D^\prime(\xi)e^{-i\zl\xi}\psi(\xi,\zl)\right]^\up,\label{eq:2.14a}\\
\varrho(\zl)^*&=-\int_{-\infty}^\infty \mathrm{d}\xi\,e^{2i\zl\xi}
\left[D^\prime(\xi)e^{-i\zl\xi}\psi(\xi,\zl)\right]^\dn,\label{eq:2.14b}
%\tcm{\varrho(\zl)}&=\int_z^\infty \mathrm{d}\xi\,\tcm{e^{-2i\zl\xi}}\left[D^\prime(\xi)
%\tcb{e^{i\zl\xi}\overline{\psi}(\xi,\zl)}\right]^\up,\label{eq:2.14c}\\
%\tcb{\tau(\zl^*)^*}&=1-\int_{-\infty}^\infty \mathrm{d}\xi\,\left[D^\prime(\xi)
%\tcb{e^{i\zl\xi}\overline{\psi}(\xi,\zl)}\right]^\dn.\label{eq:2.14d}
\end{align}
\end{subequations}
where we have used that $D(x)\to I_2$ as $x\to-\infty$. In the same way we derive
\begin{subequations}\label{eq:2.15}
\begin{align}
%\tau(\zl^*)^*&=1+\int_{-\infty}^\infty \mathrm{d}\xi\,\left[D^\prime(\xi)
%e^{-i\zl\xi}\overline{\phi}(\xi,\zl)\right]^\up,\label{eq:2.15a}\\
%\varrho(\zl)^*&=-\int_z^\infty \mathrm{d}\xi\,e^{2i\zl\xi}\left[D^\prime(\xi)
%e^{-i\zl\xi}\overline{\phi}(\xi,\zl)\right]^\dn,\label{eq:2.15b}\\
\varrho(\zl)&=\int_{-\infty}^\infty \mathrm{d}\xi\,e^{-2i\zl\xi}\,
\left[D^\prime(\xi)\,e^{i\zl\xi}\,\phi(\xi,\zl)\right]^\up,\label{eq:2.15c}\\
\tau(\zl)^*&=1+\int_{-\infty}^\infty \mathrm{d}\xi\,
\left[D^\prime(\xi)\,e^{i\zl\xi}\,\phi(\xi,\zl)\right]^\dn.\label{eq:2.15d}
\end{align}
\end{subequations}
From \eqref{eq:1.4a} and \eqref{eq:1.4c}, via \eqref{eq:3.10}, it is immediate to see that
$$
\tau(\zl)\to1-\int_{-\infty}^\infty \mathrm{d}\xi\,\begin{pmatrix}1&0\end{pmatrix}\,
D^\prime(\xi)\,\bH^{up}(\xi)\,\begin{pmatrix}1\\0\end{pmatrix}\,,\,\,\,
\mbox{as}\,\zl\to\infty\ \text{from}\ \text{within}\ \overline{\C^+}\,,
$$
and $\varrho(\zl)\to0$ as $\zl\to\pm\infty$. Note that the limit of $\tau(\zl)$ as $|\zl|\to\infty$ is a complex number of modulus 1.
\end{proof}

\subsection{Scattering data}\label{sub:b}
In this subsection, for the sake of completeness, we introduce the scattering matrix and the scattering coefficients.

From now on, we assume that the coefficient $\tau(\zl)$ introduced in the preceding section is such that $\tau(\zl)\neq0$ for all $\zl\in\R$, \textit{i.e.} there are no spectral singularities. We can write the identity \eqref{eq:1.7} as the following Riemann-Hilbert problems:
\begin{subequations}\label{eq:4.1}
\begin{align}
\begin{pmatrix}\overline{\phi}(x,\zl)&\overline{\psi}(x,\zl)\end{pmatrix}
&=\begin{pmatrix}\psi(x,\zl)&\phi(x,\zl)\end{pmatrix}
\begin{pmatrix}\tfrac{1}{\tau(\zl)}&-\tfrac{\varrho(\zl)}{\tau(\zl)}\\&\\
-\tfrac{\varrho(\zl)^*}{\tau(\zl)}&\tfrac{1}{\tau(\zl)}\end{pmatrix},\label{4.1a}\\
&\nonumber\\
\begin{pmatrix}\psi(x,\zl)&\phi(x,\zl)\end{pmatrix}
&=\begin{pmatrix}\overline{\phi}(x,\zl)&\overline{\psi}(x,\zl)\end{pmatrix}
\begin{pmatrix}\tfrac{1}{\tau(\zl)^*}&\tfrac{\varrho(\zl)}{\tau(\zl)^*}\\&\\
\tfrac{\varrho(\zl)^*}{\tau(\zl)^*}&\tfrac{1}{\tau(\zl)^*}\end{pmatrix}.\label{4.1b}
\end{align}
\end{subequations}
Putting $\bF_-(x,\zl)=\begin{pmatrix}\overline{\phi}(x,\zl) &\overline{\psi}(x,\zl)\end{pmatrix}$ and $\bF_+(x,\zl)=\begin{pmatrix} \psi(x,\zl)&\phi(x,\zl)\end{pmatrix}$, we obtain the {\it Riemann-Hilbert} problem
\begin{equation}\label{eq:4.2}
\bF_-(x,\zl)=\bF_+(x,\zl)\,\zs_3\,S(\zl)\,\zs_3\,,
\end{equation}
where the {\it scattering matrix} $S(\zl)$ is
$$
S(\zl)=\begin{pmatrix}T(\zl)&R(\zl)\\L(\zl)&T(\zl)\end{pmatrix}\,.
$$
In other words,
\begin{equation}\label{eq:4.3}
T(\zl)=\frac{1}{\tau(\zl)}\,,\qquad R(\zl)=\frac{\varrho(\zl)}{\tau(\zl)}\,,\qquad
L(\zl)=\frac{\varrho(\zl)^*}{\tau(\zl)}\,.
\end{equation}
We call $T$, $R$, and $L$ the transmission coefficient, the reflection coefficient from the right, and the reflection coefficient from the left, respectively. Equations \eqref{eq:4.1} then imply that
$$
S(\zl)^\dagger=\zs_3\,S(\zl)^{-1}\,\zs_3\,,\qquad\zl\in\R\,.
$$
Thus $S(\zl)$ is $\zs_3$-unitary and (recalling that $|\tau|^{2}+|\varrho|^{2}=1$) has determinant $\tau(\zl)^*/\tau(\zl)$. Also, $S(\zl)\to e^{-i\za}\,I_2$ as $\zl\to\pm\infty$ for a suitable complex number $e^{-i\za}$ of modulus $1$. We easily derive the Fourier representations
\begin{subequations}\label{eq:4.4}
\begin{align}
\bF_+(x,\zl)e^{-i\zl x\zs_3}&=\begin{pmatrix}H^{up}_1(x)&H^{dn}_{2}(x)\\H^{up}_2(x)&H^{dn}_{1}(x)\end{pmatrix}\nonumber\\
&\quad+\int_0^\infty \mathrm{d}\xi\,e^{i\zl\,\xi}\, \begin{pmatrix}K^{up}_1(x,x+\xi)&K^{dn}_2(x,x-\xi)\\K^{up}_2(x,x+\xi)&K^{dn}_1(x,x-\xi)\end{pmatrix}\,,\label{eq:4.4a}\\
&\nonumber\\
\bF_-(x,\zl)e^{-i\zl x\zs_3}&=\begin{pmatrix}H^{dn}_{1}(x)^*&-H^{up}_2(x)^*\\-H^{dn}_{2}(x)^*&H^{up}_1(x)^*\end{pmatrix}\nonumber\\
&\quad+\int_0^\infty \mathrm{d}\xi\,e^{-i\zl\,\xi}\, \begin{pmatrix}K^{dn}_1(x,x-\xi)^*&-K^{up}_2(x,x+\xi)^*\\-K^{dn}_2(x,x-\xi)^*&K^{up}_1(x,x+\xi)^*\end{pmatrix}\,,\label{eq:4.4b}
\end{align}
\end{subequations}
where
$$\int_0^\infty \mathrm{d}\xi\Big[|K^{up}_1(x,x+\xi)|+|K^{up}_2(x,x+\xi)|+|K^{dn}_1(x,x-\xi)|+|K^{dn}_2(x,x-\xi)|
\Big]$$
converges uniformly in $x\in\R$.
\bigskip

The scattering data associated with the first of equation \eqref{eq:pair} are:
\begin{itemize}\vspace{-0.3\baselineskip}
\item[1.] one of the reflection coefficients;\vspace{-0.5\baselineskip}
\item[2.] the poles of the transmission coefficient $T(\zl)$ (or of $T(\zl^*)^*$); we call such poles {\it the discrete eigenvalues} in the upper half-plane $\C^+$ (or in the lower half-plane $\C^-$) and denote them by $ia_j$ (or by $-ia_j^*$) for $j=1,\ldots,n$, with $\mathrm{Re}(a_{j})>0$;\vspace{-0.5\baselineskip}
\item[3.] a set of constants $c_j$ ($\overline{c}_j$) for $j=1,\ldots,n$ associated to the discrete eigenvalues $ia_j$ ($-ia_j^*$) $j=1,\ldots, n$ in the upper half-plane (lower half-plane); these constants are called the {\it norming constants}.\vspace{-0.1\baselineskip}
\end{itemize}
It is well-known that if there are no spectral singularities, then the number of discrete eigenvalues is finite \cite{FT}. It is crucial to observe that, in general, the poles of the transmission coefficient $T(\zl)$ are not necessarily simple and may have multiplicity larger than one. However, for the sake of simplicity, unless explicitly indicated differently, here and thereafter in Section \ref{sec:2} we assume that each pole of the transmission coefficient has multiplicity equal to one, as this is not restrictive when proving the symmetry of the norming constants (see Proposition \ref{P4}). The same relations can be established when the multiplicity is greater than one by following the procedure illustrated in \cite{D}. The way to construct the norming constants is standard (see \cite{AblSeg, CdG, FT}).

Let us assume that there are finitely many poles $ia_1,\ldots,ia_n$ of the transmission coefficient $T(\zl)$ in the upper half-plane $\C^+$, all of which are assumed to be simple. Following \cite{AblSeg, CdG, FT}, we let $\theta_j$ stand for the residue of
$T(\zl)$ at $\zl=ia_j$, \textit{i.e.}
\begin{align}\label{eq:res}
\theta_j=\underset{\zl=ia_j}{\mathrm{Res}}\left(T(\zl)\right)&=\lim_{\zl\to ia_j}\,(\zl-ia_j)\,T(\zl)\nonumber\\
&=\lim_{\zl\to ia_j}\,
\frac{\zl-ia_j}{\tau(\zl)-\tau(ia_j)}=\left(\left.\frac{\mathrm{d}\tau}{\mathrm{d}\zl}\right|_{\zl=ia_j}\right)^{-1}\,.
\end{align}
We then introduce the {\it norming constants} $c_j$ such that
\begin{subequations}\label{eq:4.5}
\begin{equation}\label{eq:4.5a}
\theta_j\,\phi(x,ia_j)=i\,c_j\,\psi(x,ia_j),\qquad j=1,2,\ldots,n.
\end{equation}
By the same token, $T(\zl^*)^*$ has the simple poles $-ia_1^*,\ldots,-ia_n^*$ in $\C^-$, all of them simple. The corresponding norming constants
$\overline{c}_j$ are defined by
\begin{equation}\label{eq:4.5b}
\theta_j^*\,\overline{\phi}(x,-ia_j^*)=-i\,\overline{c}_j\,\overline{\psi}(x,-ia_j^*),
\qquad j=1,2,\ldots,n.
\end{equation}
\end{subequations}
The next proposition shows how the norming constants introduced in the upper half-plane are related to those defined in the lower half-plane.
\begin{proposition}\label{P4}
The norming constants satisfy the following relations:
$$
\overline{c}_j=-(c_j)^*\,.
$$
\end{proposition}
\begin{proof}
By applying the triangular representations to \eqref{eq:4.5a} and \eqref{eq:4.5b} we get the pair of equalities
\begin{small}
\begin{align*}
&\theta_j\left\{\begin{pmatrix}H^{dn}_{2}(x)\\ H^{dn}_{1}(x)\end{pmatrix}
+\int_0^\infty \mathrm{d}\xi\,e^{-a_j\xi}\,
\begin{pmatrix}K^{dn}_2(x,x-\xi)\\K^{dn}_1(x,x-\xi)\end{pmatrix}\right\}\\
&\hspace{2.5cm}=\,ic_j\left\{\begin{pmatrix}H^{up}_1(x)\\H^{up}_2(x)\end{pmatrix}
+\int_0^\infty \mathrm{d}\xi\,
\begin{pmatrix}K^{up}_1(x,x+\xi)\\K^{up}_2(x,x+\xi)\end{pmatrix}\right\}\,,\\
&\\
&\theta_j^*\left\{\begin{pmatrix}H^{dn}_{1}(x)^*\\-H^{dn}_{2}(x)^*\end{pmatrix}
+\int_0^\infty \mathrm{d}\xi\,e^{-a_j^*\xi}\,
\begin{pmatrix}K^{dn}_1(x,x-\xi)^*\\-K^{dn}_2(x,x-\xi)^*\end{pmatrix}\right\}\\
&\hspace{2.5cm}=\,-i\overline{c}_j\left\{\begin{pmatrix}
-H^{up}_2(x)^*\\H^{up}_1(x)^*\end{pmatrix}
+\int_0^\infty \mathrm{d}\xi\,e^{-a_j^*\xi}\,
\begin{pmatrix}-K^{up}_2(x,x+\xi)^*\\ K^{up}_1(x,x+\xi)^*\end{pmatrix}\right\}\,.
\end{align*}
\end{small}
Taking the complex conjugate of the first equation and premultiplying the result by $\left(\begin{smallmatrix}0&1\\-1&0\end{smallmatrix}\right)$, we obtain the second equation, provided $\overline{c}_j=-(c_j)^*$.
\end{proof}

\subsection{Marchenko equations}\label{sub:c}
In this subsection we formulate the Marchenko integral equations and establish the connection between the solutions of these equations and the solution of the initial value problem \eqref{eq:initial}. We refer the reader to the Appendix \ref{sec:A} for the details on the derivation of \eqref{eq:4.11}.

In order to derive the Marchenko equations we need to the following
\begin{proposition}\label{Preflection}
Suppose that Assumptions \ref{HP1} and \ref{HP2} hold, and suppose that there are no spectral singularities. Then there exist functions $\hat{R}$ and $\hat{L}$ in $L^1(\R)$ such that
\begin{equation}\label{eq:reflections}
R(\zl)=\int_{-\infty}^{\infty}\mathrm{d}\xi\,e^{-i\zl\,\xi}\,\hat{R}(\xi)\,,\quad
L(\zl)=\int_{-\infty}^{\infty}\mathrm{d}\xi\,e^{i\zl\,\xi}\,\hat{L}(\xi)\,.
\end{equation}
\end{proposition}
\begin{proof}
Let $\mathfrak{B} =\{F(\zl)=c+\int_{-\infty}^{\infty}\mathrm{d}\xi\,e^{i\zl\xi}\,f(\xi)\,:c\in\C,\,\,f\in L^1(\R)\}$ denote the complex Banach algebra with norm $\|F\|=|c|+{\|f\|}_{1}\,$. By the Gelfand theory (see Chapter $11$ in \cite{RudinG}), the invertible elements of $\mathfrak{B}$ are those $F(\zl)$ for which $c\neq 0$ and $c+\int_{-\infty}^{\infty}\mathrm{d}\xi\,e^{i\zl\xi}\,f(\xi)\neq 0$ for all $\zl\in\R$. In fact, in this case
$$F^{-1}(\zl)=\dfrac{1}{c}+\int_{-\infty}^{\infty}\mathrm{d}\xi\,e^{i\zl\xi}\,g(\xi)$$
for a suitable $g\in L^1(\R)$. Using \eqref{eq:2.15} and \eqref{eq:3.12} we have
\begin{footnotesize}
\begin{align*}
\varrho(\zl)&=\int_{-\infty}^\infty \mathrm{d}\xi\,e^{-2i\zl\xi}
\left[D^\prime(\xi)\,\left\{\bH^{dn}(\xi)e^{-i\zl \xi \sigma_3}+\int_{-\infty}^{\xi}\mathrm{d}\zeta\bK^{dn}(\xi,\zeta)e^{-i\zl \zeta\sigma_3}\right\}\begin{pmatrix}1\\0\end{pmatrix}\right]^{up}\,,\\
\tau(\zl)&=1+\int_{-\infty}^\infty \mathrm{d}\xi\,
\left[D^\prime(\xi)\,\left\{\bH^{dn}(\xi)e^{-i\zl \xi \sigma_3}+\int_{-\infty}^{\xi}\mathrm{d}\zeta\bK^{dn}(\xi,\zeta)e^{-i\zl \zeta\sigma_3}\right\}\begin{pmatrix}1\\0\end{pmatrix}\right]^{dn}\,,
\end{align*}
\end{footnotesize}
\noindent where $\int_{-\infty}^{x}\mathrm{d}\xi \|\bK^{dn}(x,\xi)\|$ converges uniformly in $x\in\R$, $\bH^{dn}(x)\in SU(2)$, and the entries of $D^\prime(x)$ belong to $L^1(\R)$. Hence, $\varrho(\zl)$ and $\tau(\zl)$ are in $\mathfrak{B}$. Since there are no spectral singularities, the transition coefficient $\tau(\zl)$ is an invertible element of $\mathfrak{B}$. The algebra properties of $\mathfrak{B}$ then imply that $R(\zl)=\tfrac{\varrho(\zl)}{\tau(\zl)}$ belongs to $\mathfrak{B}$. Analogously, we can prove that $\varrho(\zl)^*$ is in $\mathfrak{B}$ and then also $L(\zl)=\tfrac{\varrho(\zl)^*}{\tau(\zl)}$ belongs to $\mathfrak{B}$.
\end{proof}

We have the following
\begin{theorem}\label{theo:Marchenko}
The auxiliary function $\bK^{up}(x,y)$ which appears in \eqref{eq:3.10} has to satisfy the following integral Marchenko equations.
\begin{equation}\label{eq:4.10}
\bK^{up}(x,y)+\bH^{up}(x)\,\bO(x+y)+\int_x^\infty \mathrm{d}\xi\,\bK^{up}(x,\xi)\,\bO(\xi+y)=0_{2\times2},
\end{equation}
where
\begin{align}\label{eq:kernel}
\bO(x)=\begin{pmatrix}0&\Omega(x)\\-\Omega(x)^*&0\end{pmatrix},
\quad\text{with}\quad
\Omega(x)=\hat{R}(x)+\sum_{j=1}^{n}\,c_j\,e^{-a_jx}\,,
\end{align}
and $\hat{R}(x)$ is the Fourier transform of the reflection coefficient (see \eqref{eq:reflections}).
\end{theorem}
We give the proof in Appendix \ref{sec:A}.
The analogous Marchenko equations satisfied by the auxiliary function $\bK^{dn}(x,y)$, which appear in \eqref{eq:3.12}, are given in Appendix \ref{sec:A} (see equation \eqref{eq:4.10dn}). The generalization of formula \eqref{eq:kernel} to the case of poles with multiplicity larger than one is given in Section \ref{sec:3} (see formula \eqref{eq:kernel1}).

Recall that $\bH^{up}(x) \in SU(2)$. By setting
\begin{equation}\label{eq:L}
\bK^{up}(x,y)=\bH^{up}(x)\,\bL(x,y)\,,
\end{equation}
we can convert \eqref{eq:4.10} into the (``usual'') Marchenko integral equation:
\begin{equation}\label{eq:4.11}
\bL(x,y)+\bO(x+y)+\int_x^\infty \mathrm{d}\xi\,\bL(x,\xi)\,\bO(\xi+y)=0_{2\times2}.
\end{equation}
By following the same proof as in the focusing AKNS case \cite{D, CORBOOK}, we find that equation \eqref{eq:4.11} is uniquely solvable on the space $L^1(x,+\infty)^{2\times2}$. Before proceeding further, we observe here that, by setting $\bK^{dn}(x,y)=\bH^{dn}(x)\,\overline{\bL}(x,y)$ (see \eqref{eq:Ldn}), an analogous Marchenko integral equation can be obtained for $\overline{\bL}(x,y)$ (see \eqref{eq:4.11dn}), as illustrated at the end of Appendix \ref{sec:A}.

For later convenience let us introduce the following notations
\begin{equation}\label{eq:notazione}
\tilde{\bK}(x)=\int_x^\infty \mathrm{d}\xi\,\bK^{up}(x,\xi)\,,\quad \tilde{\bL}(x)=\int_x^\infty \mathrm{d}\xi\,\bL(x,\xi)\,,
\end{equation}
where $\bK^{up}(x,y)$ and $\bL(x,y)$ satisfy the Marchenko integral equations \eqref{eq:4.10} and \eqref{eq:4.11}, respectively. Using \eqref{eq:2.1a} and the asymptotic relation \eqref{eq:1.4a}, we get from the triangular representation \eqref{eq:3.10}
\begin{subequations}\label{eq:3.17}
\begin{equation}\label{eq:3.17a}
I_2=\Psi(x,0)=\bH^{up}(x)+\tilde{\bK}(x)=\bH^{up}(x)\,\left[I_2+\tilde{\bL}(x)\right]\,,
\end{equation}
where
\begin{equation}\label{eq:3.17b}
\tilde{\bL}(x)={\bH^{up}(x)}^{-1}\,\,\tilde{\bK}(x)\,.
\end{equation}
\end{subequations}
Moreover, using (\ref{eq:3.11}) and (\ref{eq:3.17}), we observe that
\begin{equation}\label{eq:Linverse}
\left[I_2+\tilde{\bL}(x)\right]^{-1}=\left[I_2+\tilde{\bL}(x)\right]^{\dagger}=\left[I_2+{\tilde{\bL}^{\dagger}(x)}\right]\,,
\end{equation}
thus the structure of $\tilde{\bL}$ is
\begin{equation}\label{eq:Lstructure}
\tilde{\bL}(x)=\begin{pmatrix}{{\tilde{L}}_1(x)}&{{-\tilde{L}_2}(x)}^*\\{{\tilde{L}}_2(x)}&{{\tilde{L}_1}(x)}^*\end{pmatrix}\,.
\end{equation}

The relation between the Marchenko integral equation and the solution of equation \eqref{eq:HF1} is immediately clarified by the following
\begin{proposition}\label{P5}
The solutions of the initial value problem \eqref{eq:initial} are expressed in terms of the solutions of the Marchenko equations as:
\begin{equation}\label{eq:3.19}
\bm(x)\cdot\bzs =\bH^{up}(x)\,\,\zs_3\,\,{\bH^{up}(x)}^{-1}
=\left[I_2+{\tilde{\bL}(x)}^{\dagger}\right]\, \zs_3\, \left[I_2+\tilde{\bL}(x)\right]\,.
\end{equation}
\end{proposition}
\begin{proof}
Suppose first that in addition to the Assumptions \ref{HP1} and \ref{HP2}, we also have that $\bm^{\prime\prime}(x)$ exists almost everywhere and has its entries in $L^1(\R)$. Then, it is easy to verify that
\begin{equation}\label{eq:3.18}
\int_x^\infty \mathrm{d}\xi\left\{\left\|{\frac{\partial\bK}{\partial x}}^{up}(x,\xi)\right\|
+\left\|{\frac{\partial\bK}{\partial\xi}}^{up}(x,\xi)\right\|\right\}<+\infty,
\end{equation}
uniformly in $x\ge x_0$ for all $x_0\in\R$. Applying integration by parts to \eqref{eq:3.10} to remove factors $\zl$ in front of Fourier integral terms and using \eqref{eq:1.5a}, we get
\begin{align*}
0_{2\times2}&=\frac{\partial\Psi}{\partial x}(x,\zl)-i\zl(\bm(x)\cdot\bzs)
\Psi(x,\zl)\nonumber\\
&=i\zl\Big\{\bH^{up}(x)\zs_3-(\bm(x)\cdot\bzs)\bH^{up}(x)\Big\}\,e^{i\zl x\zs_3}
\nonumber\\
&\quad+\left\{{\frac{\mathrm{d}\bH}{\mathrm{d}x}}^{up}(x)-\bK^{up}(x,x)+(\bm(x)\cdot\bzs)\bK^{up}(x,x)\zs_3\right\}
\,e^{i\zl x\zs_3}\nonumber\\
&\quad+\int_x^\infty \mathrm{d}\xi\left\{{\frac{\partial\bK}{\partial x}}^{up}(x,\xi)
+(\bm(x)\cdot\bzs){\frac{\partial\bK}{\partial\xi}}^{up}(x,\xi)\zs_3\right\}
\,e^{i\zl x\zs_3}\,,
\end{align*}
where the integral vanishes as $\zl\to\pm\infty$. Dividing the above expression by $i\zl e^{i\zl x\sigma_3}$, taking the limit as $\zl\to\pm\infty$, and using \eqref{eq:Linverse}, we arrive at equation \eqref{eq:3.19}. In order to extend this proposition to $\bm(x)$ satisfying solely Assumptions \ref{HP1} and \ref{HP2}, it suffices to construct a sequence $\left\{\bm_n(x)\right\}$, whose elements $\bm_n(x)$ satisfy the additional assumption, namely $\bm_n^{\prime\prime}(x)$ exists almost everywhere and has its entries in $L^1(\R)$ for all $n$. For all $\bm_n(x)$, we also define the corresponding matrix function $\bH_n^{up}(x)$ by means of the same construction as in Proposition \ref{P2}. By requiring that $\|\bm(x)-\bm_n(x)\|_1+\|\bm^{\prime}(x)-\bm^{\prime}_n(x)\|_1$ vanishes as $n$ approaches $\infty$, we obtain that $\bm_n(x)$ and $\bH_n^{up}(x)$ converge pointwise to $\bm(x)$ and $\bH^{up}(x)$, respectively. This concludes the proof.
\end{proof}
We conclude this section by observing that the equation analogous to \eqref{eq:3.19} (see equation \eqref{eq:3.19dn}), relating the solution of the initial value problem \eqref{eq:HF1} to the solution of the Marchenko integral equation for $\overline{\bL}(x,y)$ (see \eqref{eq:4.11dn}), is given in Appendix \ref{sec:A}.

\subsection{Time evolution of the scattering data}\label{sub:d}
In this subsection we derive the time evolution of the scattering data. In doing so we correct typos resulting in sign errors in the time factors $e^{\pm 4 i\lambda^2 t}$ in \cite{T, ZT}. We shall arrive at the same time evolution as for the NLS equation.

Recall the Lax pair $(\bA, \bB)$ is given by \eqref{eq:pair}. Suppose that $V(x,t;\zl)$ is a nonsingular $2\times2$ matrix function satisfying
$$
V_x=\bA\,V,\qquad V_t=\bB\,V\,,
$$
where $V$ need not be one of the Jost matrices. Then there exist two invertible matrices $U_\Psi$ and $U_\Phi$, depending on $(t,\zl)$ but not on $x$, such that $\Psi=V\,U_\Psi^{-1}$ and $\Phi=V\,U_\Phi^{-1}$. Then
\begin{align*}
\Psi_t=V_t\,U_\Psi^{-1}-V\,U_\Psi^{-1}\,{[U_\Psi]}_t\,U_\Psi^{-1}&=\bB\,V\,U_\Psi^{-1}
-V\,U_\Psi^{-1}\,{[U_\Psi]}_t\,U_\Psi^{-1}\nonumber\\
&=\bB\,\Psi-\Psi\,{[U_\Psi]}_t\,U_\Psi^{-1},
\end{align*}
implying
\begin{subequations}\label{eq:6.1}
\begin{equation}\label{eq:6.1a}
{[U_\Psi]}_t\,U_\Psi^{-1}=\Psi^{-1}\,\bB\,\Psi-\Psi^{-1}\,\Psi_t.
\end{equation}
Analogously, for the other Jost matrix $\Phi(x,\zl)$ we get
\begin{equation}\label{eq:6.1b}
{[U_\Phi]}_t\,U_\Phi^{-1}=\Phi^{-1}\,\bB\,\Phi-\Phi^{-1}\,\Phi_t.
\end{equation}
\end{subequations}
Here the left-hand side does not depend on $x$, whereas the right-hand side only seemingly depends on $x$. We may therefore allow $x$ to tend to $+\infty$ without losing the validity of \eqref{eq:6.1a}, as well as to $-\infty$ without losing the validity of \eqref{eq:6.1b}. Since $\bB\simeq-2i\,\zl^2\,\zs_3$ and
$\Psi\simeq e^{i\zl x\zs_3}$ as $x\to+\infty$, from (\ref{eq:6.1a}) we obtain
\begin{subequations}\label{eq:6.2}
\begin{equation}\label{eq:6.2a}
{[U_\Psi]}_t\,U_\Psi^{-1}=-2i\,\zl^2\,\zs_3.
\end{equation}
Similarly, for the other Jost matrix $\Phi(x,\zl)$ we get
\begin{equation}\label{eq:6.2b}
{[U_\Phi]}_t\,U_\Phi^{-1}=-2i\,\zl^2\,\zs_3.
\end{equation}
\end{subequations}
From \eqref{eq:1.7}, for the transmission coefficient we get
\begin{align*}
\bT_t&={\left(\Phi^{-1}\,\Psi\right)}_t=\Phi^{-1}\,\Psi_t-\Phi^{-1}\,\Phi_t\,\Phi^{-1}\,\Psi\nonumber\\
&=\Phi^{-1}\,\Big(\bB\,\Psi-\Psi{[U_\Psi]}_t\,U_\Psi^{-1}\Big)-\Phi^{-1}\,\Big(\bB\,\Phi-\Phi\,
{[U_\Phi]}_t\,U_\Phi^{-1}\Big)\,\Phi^{-1}\,\Psi\nonumber\\
&=\Phi^{-1}\,\bB\,\Psi-\bT\,{[U_\Psi]}_t\,U_\Psi^{-1}-\Phi^{-1}\,\bB\,\Psi
+{[U_\Phi]}_t\,U_\Phi^{-1}\,\bT\nonumber\\
&=2i\,\zl^2\,\big(\bT\zs_3-\zs_3\,\bT\big),
\end{align*}
so that
\begin{equation}\label{eq:6.3a}
\bT(\zl,t)=e^{-2i\,\zl^2\,t\zs_3}\,\bT(\zl,0)\,e^{2i\,\zl^2\,t\zs_3}\,.
\end{equation}
Consequently, $\tau(\zl)$ and $T(\zl)$ do not depend on $t$, whereas
\begin{equation}\label{eq:6.3b}
R(\zl,t)=e^{-4i\zl^2t}\,R(\zl,0),\qquad L(\zl,t)=e^{4i\zl^2t}\,L(\zl,0)\,.
\end{equation}
Differentiating \eqref{eq:4.5a} with respect to $t$ we obtain
$$
\theta_j\,\phi_t(x,ia_j)=ic_j\,\psi_t(x,ia_j)+i{[c_j]}_t\,\psi(x,ia_j)\,.
$$
Using \eqref{eq:6.1} and \eqref{eq:6.2}, we get
\begin{align*}
\theta_j\,\Big\{\bB(ia_j)\,\phi(x,ia_j)-2ia_j^2\,\phi(x,ia_j)\Big\}
=&i\,c_j\,\Big\{\bB(ia_j)\,\psi(x,ia_j)+2ia_j^2\,\psi(x,ia_j)\Big\}\\
&+i{[c_j]}_t\,\psi(x,ia_j)\,.
\end{align*}
Using \eqref{eq:4.5a} again we obtain
$$
{[c_j]}_t=-4\,i\,a_j^2\,c_j\,.
$$
Remembering that $\overline{c}_j=-c_j^*$ (see Proposition \ref{P4}), finally we obtain the time evolution of the norming constants
\begin{equation}\label{eq:6.3c}
c_j(t)=e^{-4ia_j^2t}\,c_j(0)\,,\qquad\overline{c}_j(t)=e^{4i{a_j^*}^2t}\,\overline{c}_{j}(0).
\end{equation}

\subsection{Inverse scattering transform}\label{sub:e}
Having presented the {\it direct scattering problem} (consisting in the construction of the scattering data when $\bm(x, 0)$ is known), the {\it inverse scattering problem} (amounting to the construction of $\bm(x)$ when the scattering data are given) and the {\it time evolution of the scattering data} associated to the first of equation \eqref{eq:pair}, we can discuss how the IST allows us to obtain the solution to the initial value problem for \eqref{eq:HF1}.

Using the initial condition $\bm(x,0)$ as a potential in the system \eqref{eq:initial}, we develop the direct scattering theory as shown above and build the scattering data. Successively, let the initial scattering data evolve in time in agreement with equation \eqref{eq:6.3a}-\eqref{eq:6.3c}. The solution of the Heisenberg equation is then obtained by solving the Marchenko equation \eqref{eq:4.11} where the kernel $\bO(x)$ is replaced by $\bO(x;t)$ (\textit{i.e.} taking into account \eqref{eq:6.3a}, \eqref{eq:6.3b}, and \eqref{eq:6.3c}), and then using relation \eqref{eq:3.19}.

\subsection{Gauge transformation.}\label{sub:f}
In \cite{ZT} the authors proved the existence of a gauge transformation which allows one to pass from the solutions of the Heisenberg ferromagnet equation to those of the NLS equations. Here we show that this transformation is determined by the matrix $\bH^{up}(x)$ (${\bH^{dn}}(x)$) introduced in the triangular representation \eqref{eq:3.10} (\eqref{eq:3.12}). For the sake of brevity and simplicity, in this subsection we omit the time dependence.

It is well known \cite{ZS, APT, CORBOOK} that the triangular representation for the Jost solutions of the Zakharov-Shabat system, namely the scattering problem associated to the NLS equation, takes the form
\begin{subequations}\label{eq:ZS}
\begin{align}
\Psi_{ZS}(x,\zl)=e^{i\zl x\zs_3}+\int_x^\infty \mathrm{d}\xi\,\bL_{ZS}^{up}(x,\xi)\,e^{i\zl \xi\zs_3}\,,\label{eq:ZSa}\\
\Phi_{ZS}(x,\zl)=e^{i\zl x\zs_3}+\int_{-\infty}^x \mathrm{d}\xi\,\bL_{ZS}^{dn}(x,\xi)\,e^{i\zl \xi\zs_3}\,.\label{eq:ZSb}
\end{align}
\end{subequations}
for certain kernels $\bL_{ZS}^{up}$, $\bL_{ZS}^{dn}$. Comparing \eqref{eq:ZSa} to \eqref{eq:3.10}, and \eqref{eq:ZSb} to \eqref{eq:3.12}, and using \eqref{eq:notazione} and \eqref{eq:3.17}, as well as \eqref{eq:3.13} and \eqref{eq:L}, we get the following relation
\begin{subequations}\label{eq:connection}
\begin{equation}\label{eq:connectiona}
\Psi_{ZS}(x,\zl)={\bH^{up}(x)}^{-1}\,\Psi(x,\zl)\,, \quad \Phi_{ZS}(x,\zl)={\bH^{dn}(x)}^{-1}\,\Phi(x,\zl)\,,
\end{equation}
where $\Psi_{ZS}(x,\zl)$, $\Phi_{ZS}(x,\zl)$ are the Jost matrices of the Zakharov-Sahabat system, $\Psi(x,\zl)$, $\Phi(x,\zl)$ are the Jost matrices of the scattering problem \eqref{eq:pair} associated to the Heisenberg ferromagnet equation \eqref{eq:HF1}, and
\begin{footnotesize}
\begin{align}
\bH^{up}(x)&=\left[I_2+\int_x^\infty \mathrm{d}\xi\,\bL_{ZS}^{up}(x,\xi)\right]^{-1}\,,\,\label{eq:connectionb}\\
\bH^{dn}(x)&=\left[I_2+\int_{-\infty}^{x} \mathrm{d}\xi\,\bL_{ZS}^{dn}(x,\xi)\right]^{-1}\,.\label{eq:connectionc}
\end{align}
\end{footnotesize}
\end{subequations}

We have the following:
\begin{theorem}\label{eq:GaugeZS}
The solutions of the initial value problem \eqref{eq:initial} are expressed in terms of the Jost solutions of the Zakharov-Shabat system as:
\begin{equation}\label{eq:GaugeH}
\bm(x)\cdot\bzs=\Psi_{ZS}^{-1}(x,0)\,\sigma_3\,\Psi_{ZS}(x,0)=\Phi_{ZS}^{-1}(x,0)\,\sigma_3\,\Phi_{ZS}(x,0)\,.
\end{equation}
\end{theorem}
\begin{proof}
We give the proof only for the first of \eqref{eq:GaugeH} because the proof of the second one is very similar. By using \eqref{eq:connection} and \eqref{eq:3.19}, we immediately get:
\begin{align*}
\Psi_{ZS}^{-1}(x,0)\,\sigma_3\,\Psi_{ZS}(x,0)&=\left[{\bH^{up}(x)}^{-1}\,\Psi(x,0)\right]^{-1}\,\sigma_3\,\left[{\bH^{up}(x)}^{-1}\,\Psi(x,0)\right]\\
&=\Psi(x,0)^{-1}\bH^{up}(x)\,\sigma_3\,{\bH^{up}(x)}^{-1}\,\Psi(x,0)\\
&=\Psi(x,0)^{-1}\,\bm(x)\cdot\bzs\,\Psi(x,0)=\bm(x)\cdot\bzs\,,
\end{align*}
because $\Psi(x,0)=I_2$.
\end{proof}

Finally, we observe that, in principle, equation (\ref{eq:GaugeH}) can be used to generate solutions of (\ref{eq:HF1}). However, for the reasons discussed in the Introduction, we prefer to follow the approach based on the matrix triplet which allows to get an explicit and general multi-soliton solution formula.

\section{Matrix triplet method}\label{sec:3}
In this section we construct an explicit soliton solution formula for equation \eqref{eq:HF1}. To this aim, we apply the \emph{matrix triplet} technique, successfully used in \cite{DM0, DM1, DM2, DM3, DM4}. %The basic idea behind this method is to represent the kernel appearing in the Marchenko equation in a separated form. This leads to explicitly solvable Marchenko equations and then, by using equation \eqref{eq:3.19}, we can derive an explicit solution formula for equation \eqref{eq:HF1}.
Furthermore, we use this method to get explicit expressions for the Jost solutions in the reflectionless case when the corresponding scattering data are specified.

\subsection{Explicit soliton solutions for equation \eqref{eq:HF1}\label{sub:solitonsolutions}}
We want to restrict ourselves to the case $R(\zl)=0$. In this case the expression for $\bO(x;0)$ is given by \eqref{eq:kernel} putting in it $\hat{R}=0$. In particular, we can treat the situation where the discrete eigenvalues are not necessarily simple \cite{D} by generalizing formula \eqref{eq:kernel} as follows
\begin{equation}\label{eq:kernel1}
\Omega(x;t)=\sum_{j=1}^n\sum_{k=0}^{n_j-1}c_{jk}(t)\frac{x^k}{k!}\,e^{-a_jx}\,.
\end{equation}
In \eqref{eq:kernel1}, $n$ is the number of discrete eigenvalues $\{ia_j\}_{j=1}^{n}$, namely the poles of the transmission coefficient $T(\zl)$ in $\C^+$ (thus, satisfying $\mathrm{Re}(a_{j})>0$); the quantities $a_{j}$ are obtained by multiplying the discrete eigenvalues by $-i$; $n_j$ is the algebraic multiplicity of $ia_j$; and $\left\{c_{jk}(t)\right\}_{k=0}^{n_{j}-1}$, for all $j=1,2,...,n$, are the (time-dependent) norming constants corresponding to $ia_{j}$, evolving in time according to \eqref{eq:6.3c}.

To recover the solution of \eqref{eq:initial} we follow the three steps indicated below.
\begin{itemize}
\item[a.] Suppose that the scattering data, namely the discrete eigenvalues and the corresponding norming constants,
    $$
    \{i a_j\}_{j=1}^{n}\,\quad\mbox{ and }\quad
    \left\{\{c_{jk}(t)\}_{k=0}^{n_j-1}\right\}_{j=1}^{n},
    $$
    are given. Then, we construct $\bO(x)$ as in \eqref{eq:kernel} and we let it evolve in time using \eqref{eq:kernel1}:
    \begin{equation}\label{eq:kernel2}
    \bO(x;t)=\begin{pmatrix}0&\Omega(x;t)\\
    -\Omega(x;t)^*&0\end{pmatrix}\,.
    \end{equation}
\item[b.] We solve the Marchenko integral equation \eqref{eq:4.11}:
    \begin{equation*}
    \bL(x,y;t)+\bO(x+y;t)+\int_x^\infty \mathrm{d}\xi\,\bL(x,\xi;t)\,\bO(\xi+y;t)=0_{2\times2}\,.
    \end{equation*}
    where $\xi>x$ and the kernel $\bO(x,y)$ is given in \eqref{eq:kernel2}.
\item[c.] We construct the potential $\bm(x;t)$ by using formula\eqref{eq:3.19}:
    \begin{equation}\label{eq:solvingformula}
    \bm(x;t)\cdot\bzs=
    \left[I_2+{\tilde{\bL}^{\dagger}(x;t)}\right]\,\zs_3\,\left[I_2+\tilde{\bL}(x;t)\right],
    \end{equation}
    where $\tilde{\bL}(x)=\int_x^\infty \mathrm{d}\xi\,\bL(x,\xi)$.
\end{itemize}
\noindent Let us follow the above procedure (an analogous procedure can be developed with the kernel $\overline{\bO}$, as per in \eqref{eq:4.10dnb}, and solving the Marchenko equation \eqref{eq:4.11dn} for $\overline{\bL}$). We start by disregarding the time dependence (\textit{e.g.} we construct $\bO(x)$ assuming no dependence on the time). We will subsequently show how to take the time dependence into account.

It is well known \cite{Dym, CORBOOK} that it is possible to factorize a matrix function which is in the form \eqref{eq:kernel2} with \eqref{eq:kernel1} by using a suitable triplet of matrices. More precisely, let $\bar{n}=\sum_{j=1}^{n}n_j$, and suppose $(\CA,\CB,\CC)$ is a matrix triplet such that all the eigenvalues of the $2\bar{n}\times2\bar{n}$ matrix $\CA$ have positive real parts, $\CB$ is $2\bar{n}\times2$, and $\CC$ is $2\times2\bar{n}$. We then set
\begin{subequations}\label{eq:5.1}
\begin{equation}\label{eq:5.1a}
\bO(x)=\begin{pmatrix}0&\Omega(x)\\-\Omega(x)^*&0\end{pmatrix}\deff\CC\,e^{-x\CA}\,\CB\,.
\end{equation}
Alternatively, equation \eqref{eq:5.1a} can be written by setting
\begin{equation}\label{eq:5.1b}
\Omega(x)=\sum_{j=1}^n\sum_{k=0}^{n_j-1}c_{jk}\frac{x^k}{k!}\,e^{-a_jx}=C\,e^{-xA}\,B\,,
\end{equation}
with
\begin{equation}\label{eq:5.4a}
\CA=\begin{pmatrix}A&0_{\bar{n}\times\bar{n}}\\0_{\bar{n}\times\bar{n}}&A^\dagger\end{pmatrix},
\quad\CB=\begin{pmatrix}0_{\bar{n}\times1}&B\\-C^\dagger&0_{\bar{n}\times1}\end{pmatrix},
\quad\CC=\begin{pmatrix}C&0_{1\times\bar{n}}\\0_{1\times\bar{n}}&B^\dagger\end{pmatrix}.
\end{equation}
\end{subequations}
Here $A$ is an $\bar{n}\times\bar{n}$ matrix whose $n$ eigenvalues $\left\{a_{j}\right\}_{j=1}^{n}$ are obtained from the poles $\left\{i a_{j}\right\}_{j=1}^{n}$ of the transmission coefficient $T(\zl)$ (namely the discrete eigenvalues) by multiplication by a factor $-i$ (we shall see a proof of this fact in Subsection ${3.2}$); $B$ is a $\bar{n}\times1$ matrix; and $C$ is a $1\times\bar{n}$ matrix. Furthermore, we assume that the triplet $(A,B,C)$ is a \textit{minimal} triplet in the sense that the matrix order of $A$ is minimal among all triplets representing the same Marchenko kernel by means of \eqref{eq:5.1} \cite{Dym, CORBOOK}. As the discrete eigenvalues $\left\{i a_{j}\right\}_{j=1}^{n}$ belong to the upper half-plane $\mathbb{C}^{+}$, we have $\mathrm{Re}(a_{j})>0$ for all $j$, namely all the eigenvalues of the matrix $A$ have positive real parts: this fact is necessary in order to assure the convergence of the integrals in \eqref{eq:5.4g}. Moreover, we recall that the minimality of the triplet $(A,B,C)$ entails that the geometric multiplicity of the eigenvalues of $A$ be one \cite{DM0}.

In particular, it is worth observing here that it is not restrictive (in fact, it is the typical choice) to set the triplet $(A, B, C)$ as follows \cite{CORBOOK}:
\begin{subequations}\label{eq:ABC}
\begin{equation}\label{eq:ABCa}
A_{\bar{n}\times\bar{n}}=\begin{pmatrix}
A_1&0&\cdots&0\\
0& A_2&\cdots&0\\
\vdots&\vdots &\ddots&\vdots\\
0&0&\cdots&A_n
\end{pmatrix}\,,\,\,
B_{\bar{n}\times1}=\begin{pmatrix} B_1\\ B_2\\ \vdots\\ B_n\end{pmatrix}\,,\,\,
C_{1\times\bar{n}}=\begin{pmatrix} C_1&C_2& \cdots & C_n\end{pmatrix}\,,
\end{equation}
where $A$ is in Jordan canonical form, with $A_j$ being the Jordan block of dimension $n_j\times n_j$ corresponding to the discrete eigenvalue $ia_j$,
\begin{equation}\label{eq:ABCb}
A_j = \left\{\begin{array}{cl}
a_j&\mbox{ if } n_{j}=1\\
&\\
\begin{psmallmatrix}a_{j}&1&0&0\\0&\cdot&\cdot&0\\0&0&\cdot&1\\0&0&0&a_j\end{psmallmatrix}
&\mbox{ if } n_{j}>1\,;\\
\end{array}\right.
\end{equation}
$B_j$ is a column vector of dimension $n_j$, typically chosen to be a vector of ones; and $C_j$ is a row vector of dimension $n_j$, typically chosen to be the vector of the norming constants corresponding to the discrete eigenvalue $ia_j$,
\begin{equation}\label{eq:ABCc}
C_{j} = \begin{pmatrix}c_{j,0}&c_{j,1}&\cdots&c_{j,n_j-1}\end{pmatrix}\,,
\end{equation}
\end{subequations}
so that the elements of $C$ are chosen to be the $\bar{n}$ norming constants $\left\{\{c_{jk}\}_{k=0}^{n_j-1}\right\}_{j=1}^{n}$. Note that, due to the minimality, if the triplet $(A, B, C)$ is set as in \eqref{eq:ABC}, then $A$ features no repeated blocks on the main diagonal.

Howbeit, in the present Section \ref{sec:3} we make no specific choice for $(A,B,C)$, and we assume the matrix triplet to be as generic as possible, save for the conditions dictated by the minimality and the positiveness of the real part of the eigenvalues. Special choices, corresponding to different classes of soliton solutions, will be discussed in Section \ref{sec:4}.

In view of the following, we also introduce the matrix $\CP$, which is the unique solution of the Sylvester equation
\begin{subequations}\label{eq:5.4}
\begin{equation}\label{eq:5.4b}
\CA\,\CP+\CP\,\CA=\CB\,\CC\,,
\end{equation}
\noindent namely
\begin{equation}\label{eq:5.4c}
\CP=\int_0^\infty \mathrm{d}\xi\,e^{-\xi\CA}\,\CB\,\CC\,e^{-\xi\CA}.
\end{equation}
\noindent Note that it is also possible to write $\CP$ as
\begin{equation}\label{eq:5.4d}
\CP=\begin{pmatrix}0_{\bar{n}\times\bar{n}}&N\\-Q&0_{\bar{n}\times\bar{n}}\end{pmatrix},
\end{equation}
where $N$ and $Q$ solve the Lyapunov matrix equations
\begin{align}
A^{\dagger}\,Q+Q\,A&=C^{\dagger}\,C\,,\label{eq:5.4e}\\
A\,N+N\,A^{\dagger}&=B\,B^{\dagger}\,,\label{eq:5.4f}
\end{align}
that is
\begin{equation}\label{eq:5.4g}
N=\int_0^\infty \mathrm{d}\xi\,e^{-\xi\,A}\,B\,B^\dagger\,e^{-\xi\,A^\dagger},\qquad
Q=\int_0^\infty \mathrm{d}\xi\,e^{-\xi\,A^\dagger}\,C^\dagger\,C\,e^{-\xi\,A}.
\end{equation}
\noindent By the minimality of the triplet $(A,B,C)$ \cite[Sec.4.1]{CORBOOK}, we see that $N$ and $Q$ are positive Hermitian matrices. Thus $\CP$ is invertible and
\begin{equation}\label{eq:5.4h}
\CP^{-1}=\begin{pmatrix}0_{\bar{n}\times\bar{n}}&-Q^{-1}\\N^{-1}&0_{\bar{n}\times\bar{n}}\end{pmatrix}.
\end{equation}
\end{subequations}

Now we are ready to express the solution $\bL(x,y)$ of the Marchenko integral equation \eqref{eq:4.11} in terms of the triplet $(\mathcal{A},\mathcal{B},\mathcal{C})$ and of the matrix $\mathcal{P}$. Indeed, by substituting the expression of the kernel \eqref{eq:5.1} into \eqref{eq:4.11}, we arrive at the following Marchenko equation
\begin{equation}\label{eq:separable}
\bL(x,y)+\CC\,e^{-(x+y)\,\CA}\,\CB+\int_x^\infty \mathrm{d}\xi\,\bL(x,\xi)\,\CC\,e^{-(\xi+y)\,\CA}\,\CB
=0_{2\times2}.
\end{equation}
Equation \eqref{eq:separable} can be solved explicitly via separation of variables. In fact, looking for a solution in the form
$$\bL(x,y)=-\bF(x)e^{-y\CA}\CB,$$
after some straightforward calculations, we find
\begin{equation}\label{eq:solutionMarch}
\bL(x,y)=-\CC e^{-x\CA}[I_{2\bar{n}}+e^{-x\CA}\CP e^{-x\CA}]^{-1}e^{-y\CA}\CB,
\end{equation}
\noindent provided the inverse matrix exists for all $x\in\R$.

Finally in order to reconstruct the solution of \eqref{eq:initial} we have to integrate \eqref{eq:solutionMarch} with respect to $y$, obtaining the explicit formula
\begin{equation}\label{eq:final}
\tilde{\bL}(x)=-\CC e^{-x\CA}[I_{2\bar{n}}+e^{-x\CA}\CP e^{-x\CA}]^{-1}e^{-x\CA}\,\CA^{-1}\,\CB.
\end{equation}
The right-hand side of \eqref{eq:solvingformula} is now explicit and we can use such formula to recover the components $m_j(x)$, $j=1,2,3$, of the vector $\bm(x)$.

Let us now introduce the dependence on the time $t$. In order to recover it, we have to take into account the time evolution of the scattering data expressed by \eqref{eq:6.3a}-\eqref{eq:6.3c}. Then the (reflectionless) Marchenko kernels become:
\begin{subequations}\label{eq:6.4}
\begin{alignat}{2}
\Omega(x;t)&=\sum_{j=1}^n\sum_{k=0}^{n_j-1}c_{jk}(t)\frac{x^k}{k!}\,e^{-a_jx}& &=C\,e^{-4itA^2}\,e^{-xA}\,B\,,\\
\Omega(x;t)^*&=\sum_{j=1}^n\sum_{k=0}^{n_j-1}c_{jk}^*(t)\frac{x^k}{k!}\,e^{-a_j^*x}& &=B^\dagger\,e^{-xA^\dagger}\,e^{4it{A^\dagger}^2t}\,C^\dagger\,.
\end{alignat}
\end{subequations}
In other words, we may replace the matrix triplet $(A,B,C)$ for the triplet $(A,B,Ce^{-4itA^2})$ in a such way that \eqref{eq:6.3a}, \eqref{eq:6.3b} and \eqref{eq:6.3c} are satisfied ($A$ contains the discrete eigenvalues which are time independent and $C$ the norming constants). Consequently, the explicit right-hand side of \eqref{eq:solvingformula} can be written as follows:
\begin{subequations}\label{eq:final1}
\begin{align}\label{eq:final1a}
\tilde{\bL}(x;t)&=-\CC(t)\,e^{-x\CA}\,\left[I_{2\bar{n}}+e^{-x\CA}\CP(t) e^{-x\CA}\right]^{-1}\,e^{-x\CA}\,\CA^{-1}\,\CB(t)=\nonumber\\
&=-\CC(t)\,\left[e^{2\,x\CA}+\CP(t)\right]^{-1}\,\CA^{-1}\,\CB(t)\,,
\end{align}
where
\begin{equation}\label{eq:final1b}
\CB(t)=\begin{pmatrix}0_{\bar{n}\times1}&B\\-\left(Ce^{-4itA^2}\right)^\dagger&0_{\bar{n}\times1}\end{pmatrix},
\quad\CC(t)=\begin{pmatrix}Ce^{-4itA^2}&0_{1\times\bar{n}}\\0_{1\times\bar{n}}&B^\dagger\end{pmatrix}\,,
\end{equation}
and
\begin{equation}\label{eq:final1c}
\CP(t)=\begin{pmatrix}0_{\bar{n}\times\bar{n}}&N\\-Q(t)&0_{\bar{n}\times\bar{n}}\end{pmatrix}\,,
\end{equation}
with
\begin{equation}\label{eq:final1d}
Q(t)=\int_0^\infty dx\,e^{-xA^\dagger}\left(Ce^{-4itA^2}\right)^\dagger Ce^{-4itA^2}e^{-xA},
\end{equation}
satisfying
\begin{equation}\label{eq:final1e}
A^{\dagger}Q(t)+Q(t)A=(Ce^{-4itA^2})^{\dagger}(Ce^{-4itA^2}).
\end{equation}
\noindent It is worth observing that, if $Q$ is the solution of the Lyapunov matrix equation \eqref{eq:5.4e}, then
\begin{equation}\label{eq:final1f}
Q(t)=(e^{-4itA^2})^{\dagger}\,Q\,e^{-4itA^2}\,.
\end{equation}
\end{subequations}

Finally, after some algebra (see Appendix \ref{sec:B}), using \eqref{eq:solvingformula} with \eqref{eq:final1} and \eqref{eq:Lstructure}, we have an explicit expression for $\bm(x;t)$ in terms of the elements of the matrix $\tilde{\bL}(x,t)$:
\begin{subequations}\label{eq:final2}
\begin{align}
m_1(x;t)&=-2\,\mathrm{Re}\Big((1+\tilde{L}_1(x;t))\,\tilde{L}_2(x;t)\Big)\,,\label{eq:final2a}\\
m_2(x;t)&=-2\,\mathrm{Im}\Big((1+\tilde{L}_1(x;t))\,\tilde{L}_2(x;t)\Big)\,,\label{eq:final2b}\\
m_3(x;t)&=2\,\left|1+\tilde{L}_1(x;t)\right|^2-1\,.\label{eq:final2c}
\end{align}
\end{subequations}
\noindent From the conservation of the norm of the magnetization $\|\bm(x;t)\|=1$, we immediately get $|1+\tilde{L}_{1}|^{2}+|\tilde{L}_{2}|^{2}=1$. In fact, being a point on $\mathbb{S}^{2}$, the whole solution $\bm(x,t)$ should and can be described by solely two quantities, \textit{i.e.}, the norm of $(1+\tilde{L}_{1})$ and the argument of $(1+\tilde{L}_{1})\,\tilde{L}_{2}$.

Further and more explicit expressions for the magnetization vector are provided in Appendix \ref{sec:B}.

\subsection{Reconstruction of the Jost solutions \label{sub:reconstruction}}
By using the same notations introduced in the subsection above, let us compute the Jost solution $\Psi(x,\zl)$ by substituting the solution of the Marchenko equation \eqref{eq:4.10} (see formula \eqref{eq:solutionMarch}) into
\eqref{eq:3.10}. We get
\begin{footnotesize}
\begin{align*}
\Psi(x,\zl)e^{-i\zl  x\zs_3}&=\bH^{up}(x)-\bH^{up}(x)\,\CC\,e^{-x\CA}\,
\left[I_{2\bar{n}}+e^{-x\CA}\,\CP\,e^{-x\CA}\right]^{-1}\,
\int_x^\infty \mathrm{d}\xi\,e^{-\xi\CA}\CB e^{i\zl(\xi-x)\zs_3}\\
&=\bH^{up}(x)-\bH^{up}(x)\,\CC\,\left[\CP+e^{2x\CA}\right]^{-1}
\int_0^\infty \mathrm{d}\xi\,e^{-\xi\CA}\CB e^{i\zl \xi\zs_3}\\
&=\bH^{up}(x)-\bH^{up}(x)\,\CC\,\left[\CP+e^{2x\CA}\right]^{-1}\,\CD(\zl),
\end{align*}
\end{footnotesize}
where
\begin{equation*}
-\CA\,\CD(\zl)+i\zl\CD(\zl)\,\zs_3=-\CB
\end{equation*}
has a unique solution (note that, if $\Lambda(\CA)$ is the spectrum of $\CA$, then $\Lambda(\CA)\cap\{i\zl,-i\zl\}=\emptyset$, for $\zl\in\R$). Observing that (see \eqref{eq:5.4})
\begin{align*}\label{eq:ABsigma3a}
\Sigma_{3}\,\CB+\CB\,\sigma_{3}=0_{2\bar{n}\times2}\,,\quad
\Sigma_{3}\,\CA-\CA\,\Sigma_{3}=0_{2\bar{n}\times2\bar{n}}\,,
\end{align*}
where
\begin{equation*}\label{eq:ABsigma3b}
\Sigma_{3}=\sigma_{3}\otimes I_{\bar{n}}=\begin{pmatrix}I_{\bar{n}}&0_{\bar{n}\times\bar{n}}\\0_{\bar{n}\times\bar{n}}&-I_{\bar{n}}\end{pmatrix}\,,
\end{equation*}
we have
\begin{align*}
\CD(\zl)&=\int_0^\infty \mathrm{d}\xi\,e^{-\xi\CA}\,\CB\,e^{i\zl \xi\zs_3}
=\int_0^\infty \mathrm{d}\xi\,e^{-\xi\CA}e^{-i\zl \xi\zS_3}\,\CB
=\int_0^\infty \mathrm{d}\xi\,e^{-\xi\CA-i\zl \xi\zS_3}\,\CB\\
&=\left(i\zl\zS_3+\CA\right)^{-1}\,\CB=-i\,\left(\zl I_{2\bar{n}}-i\zS_3\CA\right)^{-1}\,\zS_3\,\CB
=i\left(\zl I_{2\bar{n}}-i\zS_3\,\CA\right)^{-1}\,\CB\,\zs_3.
\end{align*}
Consequently,
\begin{align*}
\Psi(x,\zl)e^{-i\zl x\zs_3}&=\bH^{up}(x)\left[I_{2}-i\,\CC\left[\CP+e^{2x\CA}\right]^{-1}
\left(\zl I_{2\bar{n}}-i\zS_3\,\CA\right)^{-1}\CB\zs_3\right]\,.
\end{align*}
Taking the limit as $x\to-\infty$, we obtain
\begin{footnotesize}
\begin{equation}\label{eq:limtaurho}
\begin{pmatrix}\tau(\zl)&-\ds\lim_{x\to-\infty}\,e^{2i\zl x}\varrho(\zl)\\
\ds\lim_{x\to-\infty}\,e^{-2i\zl x}\varrho(\zl)^*&\tau(\zl)^*\end{pmatrix}
=I_2-i\,\CC\CP^{-1}\left(\zl I_{2\bar{n}}-i\zS_3\,\CA\right)^{-1}\CB\,\zs_3\,,
\end{equation}
\end{footnotesize}
provided $\CP$ is invertible. Equation \eqref{eq:limtaurho} with \eqref{eq:5.4} imply that $\varrho(\zl)=0$ (reflectionless case). Moreover, from \eqref{eq:limtaurho}, after some algebraic manipulation, we also find the following:
\begin{subequations}\label{eq:5.3}
\begin{align}
\tau(\zl)&=1-iCQ^{-1}(\zl I_{\bar{n}}+iA^\dagger)^{-1}C^\dagger,\label{eq:5.3a}\\
\tau(\zl)^*&=1+iB^\dagger N^{-1}(\zl I_{\bar{n}}-iA)^{-1}B.\label{eq:5.3b}
\end{align}
Taking complex conjugate transposes we get the alternative expressions
\begin{align}
\tau(\zl)&=1-iB^\dagger(\zl I_{\bar{n}}+iA^\dagger)^{-1}N^{-1}B,\label{eq:5.3c}\\
\tau(\zl)^*&=1+iC(\zl I_{\bar{n}}-iA)^{-1}Q^{-1}C^\dagger.\label{eq:5.3d}
\end{align}
\end{subequations}
Thus $\tau(\infty)=1$. Formulae \eqref{eq:5.3} can be further simplified by means of the {\it matrix determinant lemma} (obtainable from the Sherman-Morrison-Woodbury formula, see Sec. 2.1.3 of \cite{GvL1983}), which states that, for a generic invertible ($\bar{n}\times\bar{n}$) square matrix $X$ and for generic ($\bar{n}\times1$) column vectors $U$, $V$, one has
\begin{equation}\label{eq:detlemma}
\det{\left(X+U\,V^{\dagger}\right)} = \left(1+V^{\dagger}\,X^{-1}\,U\right)\,\det{\left(X\right)}\,.
\end{equation}
Applying \eqref{eq:detlemma}, via \eqref{eq:5.4e} and \eqref{eq:5.4f}, to any of \eqref{eq:5.3}, we get
\begin{equation}\label{eq:taulambda}
\tau(\zl)=\frac{\det(\zl I_{\bar{n}}-iA)}{\det(\zl I_{\bar{n}}+iA^\dagger)}\,.
\end{equation}
This also proves that the discrete eigenvalues $\left\{ia_{j}\right\}_{j=1}^{n}$ in $\C^+$ are exactly the eigenvalues of $iA$. Finally, it is worth remarking that, in contrast to the triangular representations in \cite{T} and \cite{ZT} (see formulae (13) and (17) in \cite{ZT}), definitions \eqref{eq:3.10} and \eqref{eq:3.12} allowed a direct and straightforward computation of the asymptotic behaviour for large $\lambda$ of the Jost solutions, as well as an explicit expression for the coefficient $\tau(\zl)$.

\section{Classes of soliton solutions}\label{sec:4}
In the present section we discuss classes of soliton solutions of (\ref{eq:HF1}), as resulting from the explicit formula (\ref{eq:final2}) with \eqref{eq:final1}. Moreover, we provide several numerical examples, obtained by computing (on MATLAB R2017a) the terms $\tilde{L}_{1}$ and $\tilde{L}_{2}$ in \eqref{eq:final2} by means of formulae \eqref{eq:Lexplicit1a} and \eqref{eq:Lexplicit1d} when $x$ is large and negative, and formulae \eqref{eq:Lexplicit1b} and \eqref{eq:Lexplicit1e} when $x$ is large and positive.

An immediate classification of the soliton solutions of (\ref{eq:HF1}) can be obtained by considering the algebraic multiplicity of the eigenvalues of the matrix $A$ in the matrix triplet $(A,B,C)$ in \eqref{eq:5.1b}. Propagating and stationary soliton solutions (the so-called \textit{magnetic-droplet solitons}, see \cite{KosevichIvanovKovalev1990}) are associated to algebraically simple eigenvalues of $A$. Multiple-pole (or, more simply, \textit{multipole}) soliton solutions are instead associated to eigenvalues of $A$ having algebraic multiplicity larger than one (\textit{i.e.}, \textit{degenerate} eigenvalues). In the following, we choose $A$ to be in Jordan canonical form as in \eqref{eq:ABC}: single eigenvalues on the main diagonal are associated to individual (stationary or propagating) solitons, whereas Jordan blocks of algebraic multiplicity $n_{j}>1$ are associated to multipole solutions. No blocks are repeated, as the geometric multiplicity of each eigenvalue is one due to the minimality of the triplet \cite{DM0, CORBOOK}. As shown in \ref{sec:4a}, solitons propagate with a constant velocity that is directly proportional to the imaginary part of the associated eigenvalue of $A$, so that stationary solitons are associated to real eigenvalues. On the contrary, as illustrated in Section \ref{sec:4c}, multipole solitons propagate with a velocity that changes logarithmically in time \cite{Olmedilla1986, Schiebold2014}. Finally, as shown in Section \ref{sec:4b}, oscillating (breather-like) solutions can be considered as a subset of the multi-soliton solutions, and can be constructed by putting two (or more) stationary or propagating solitons very near to each other.

\subsection{One-soliton solution}\label{sec:4a}
In the context of the matrix triplet method, the one-soliton solution (already found in \cite{NakSas1974,Tjon,Lak,T}) corresponds to the choice $n=1$, $n_{1}=1$ in \eqref{eq:5.1} -- so that $\bar{n}=1$ -- and it can be obtained as follows. If we set the matrix triplet $(A,B,C)$ as
$$
A=\left(a\right)\,,\quad B=\left(1\right)\,, \mbox{ and } C=\left(c\right)\,,
$$
we get
$$
N = \left(\frac{1}{2\,\mathrm{Re}(a)}\right)\,,\quad
Q = \left(\frac{|c|^{2}}{2\,\mathrm{Re}(a)}\right)\,,\quad
Q(t) = \left(\frac{|c|^{2}\,e^{-4\,i\,\mathrm{Re}(a^{2})\,t}}{2\,\mathrm{Re}(a)}\right)\,,
$$
and from \eqref{eq:Lstructure} and \eqref{eq:Lexplicit1} we have
\begin{align*}
\tilde{L}_{1} &= -\frac{2\,|c|^{2}\,\mathrm{Re}(a)}{a^*\,\Big(|c|^{2}+4\,\mathrm{Re}(a)^{2}\,e^{4\,\mathrm{Re}(a)\,(x-4\,\mathrm{Im}(a)\,t)}\Big)}\,\,\,,\\
\tilde{L}_{2} &= \frac{2\,c^*\,\mathrm{Re}(a)\,e^{2\,a\,(x+2\,i\,a\,t)}}{a^*\,\Big(|c|^{2}+4\,\mathrm{Re}(a)^{2}\,e^{4\,\mathrm{Re}(a)\,(x-4\,\mathrm{Im}(a)\,t)}\Big)}\,\,\,.
\end{align*}
We observe that solely the modulus of $c$ appears in $\tilde{L}_{1}$. Then we set, without any loss of generality,
$$
a = p+i\,q\,,\quad p>0\,,
\quad\mbox{ and }\quad
|c| = 2\,p\,e^{2\,p\,x_{0}}\,
$$
for some $x_{0}\in\R$. From \eqref{eq:final2}, after some simple algebra, we obtain the celebrated one-soliton solution, see \cite{NakSas1974,Tjon,Lak,T}:
\begin{subequations}\label{eq:onesoliton}
\begin{align}
m_{+}(x,t)&=2\,
\frac{(p+i\,q)-(p-i\,q)\,e^{4\,p\,(x-4\,q\,t-x_{0})}}{(p-i\,q)^{2}\,\left[1+e^{4\,p\,(x-4\,q\,t-x_{0})}\right]^{2}}\,
e^{2\,(p+i\,q)\,(x-4\,q\,t)-2\,p\,x_0-i\,\mathrm{arg}(c)}\,e^{4\,i\,(p^{2}+q^{2})\,t}\,,\label{eq:onesolitona}\\
m_{3}(x,t) &= 1-\frac{2\,p^{2}\,\mathrm{sech}^{2}\Big(2\,p\,(x-4\,q\,t-x_{0})\Big)}{p^{2}+q^{2}}\,,\label{eq:onesolitonb}
\end{align}
\end{subequations}
where $m_{+}(x,t) = m_{1}(x,t)+i\,m_{2}(x,t)$. This solution describes a localized, coherent configuration featuring a locally inverted magnetization whose point of minimum on $\be_{3}$ travels at the constant speed
\begin{subequations}\label{eq:onesoliton-parameters}
\begin{equation}\label{eq:onesoliton-speed}
v = 4\,\mathrm{Im}(a) = 4\,q\,.
\end{equation}
The minimum of $m_3$ is $\frac{q^{2}-p^{2}}{q^{2}+p^{2}}$ and is attained at $x=4\,q\,t+x_0$. Furthermore, we observe that, in the right-hand side of (\ref{eq:onesolitona}), the exponent of the last exponential term is a phase factor depending only on the time $t$. Consequently, the magnetic configuration described by \eqref{eq:onesoliton} features also a precession of constant frequency
\begin{equation}\label{eq:onesoliton-frequency}
\omega = 4\,|a|^{2} = 4\,(p^{2}+q^{2})\,
\end{equation}
on the $(\be_1,\be_2)$ plane. Indeed, the constant speed $v$ and the precession frequency $\omega$, depending only on the real and imaginary parts of the eigenvalue $a$, characterize entirely the solution from the physical point of view. Moreover, by inverting \eqref{eq:onesoliton-speed} and \eqref{eq:onesoliton-frequency},
\begin{equation}\label{eq:onesoliton-pq}
p = \frac{1}{2}\,\sqrt{\omega-\frac{v^{2}}{4}}\,,\quad
q = \frac{v}{4}\,,
\end{equation}
we immediately obtain the condition for localization (see \cite{KosevichIvanovKovalev1990}),
\begin{equation}\label{eq:onesoliton-localization}
\omega\geq0\,,\quad |v|\leq2\,\sqrt{\omega}\,.
\end{equation}
\end{subequations}
It is worth observing that $m_{3}\neq-1$ for all propagating ($q\neq0$) soliton solutions. Stationary ($q=0$) soliton solutions, for which $m_{3}$ attains $-1$ at the minimum, can be safely constructed by considering the limit of a sequence of propagating soliton solutions associated to a sequence of eigenvalues of the form $a_{j} = p+i\,q_{j}$, where $q_j\neq0$ for all $j$, and $q_{j}\to0$ as $j\to\infty$. Therefore, a real eigenvalue ($q=0$) corresponds to a stationary ($v=0$) soliton.

On the other hand, the norming constant $c$ can be used to give the initial ($t=0$) position $x_{0}$ of the minimum of $m_{3}$ and the initial phase $\varphi_{0}$ on the $(\be_1,\be_2)$ plane (namely, the direction to which the magnetization points on the $(\be_1,\be_2)$ plane at $t=0$). It is not difficult to see that, if one sets, without any loss of generality,
\begin{equation}\label{eq:normingc}
c\equiv c(p,q,x_{0},\varphi_{0})=\left\{
\begin{array}{ll}
2\,i\,p\,\mathrm{sign}(q)\,\left(\frac{p+i\,q}{p-i\,q}\right)\,e^{2\,(p+i\,q)\,x_{0}-i\,\varphi_{0}}&\mbox{ if } q\neq0\\
&\\
2\,p\,e^{2\,p\,x_{0}-i\,\varphi_{0}}&\mbox{ if } q=0\,,
\end{array}
\right.
\end{equation}
then the one soliton solution \eqref{eq:onesoliton} can be elegantly written in the following neat form
\begin{subequations}\label{eq:onesoliton-omegav}
\begin{align}
\left(\begin{array}{c}m_1(x,t)\\m_2(x,t)\end{array}\right) &=
\frac{1-m_3(x,t)}{p}\,
\left(\begin{array}{cc}\cos\beta(x,t)&-\sin\beta(x,t)\\ \sin\beta(x,t)&\cos\beta(x,t)\end{array}\right)\,
\left(\begin{array}{c}
q\,\cosh\kappa(x,t)\\
p\,\sinh\kappa(x,t)
\end{array}\right)\,,\label{eq:onesoliton-omegav-m12}\\
m_3(x,t) & = 1-\frac{2\,p^{2}}{p^{2}+q^{2}}\,\mathrm{sech}^{2}\kappa(x,t)\,,\label{eq:onesoliton-omegav-m3}
\end{align}
where $p$ and $q$ are given in \eqref{eq:onesoliton-pq} and
\begin{align}\label{eq:onesoliton-omegav-angle}
\kappa(x,t)&=2\,p\,(x-v\,t-x_{0})=\sqrt{\omega-\frac{v^{2}}{4}\,\,}\,\,(x-v\,t-x_{0})\,\,,\\
\beta(x,t)&=\omega\,t+\frac{v}{2}\,(x-v\,t-x_{0})+\varphi_{0}\,.
\end{align}
\end{subequations}
In Figure \ref{fig:soliton-propagating} we illustrate a propagating, one-soliton solution for the choice
$$
v=1\,,\,\omega=2\,,\,x_0=-4\,,\,\varphi_0=0\,,
\,\mbox{ entailing }\,
p = \tfrac{\sqrt{7}}{4}\,,\, q = \tfrac{1}{4}\,,\, c=\tfrac{-7+i\,3\,\sqrt{7}}{8}\,e^{-2\,(i+\sqrt{7})}\,.
$$
\begin{center}
\begin{figure}[!ht]
\hspace{-0.0cm}\subfigure[{$m_{1}(x,t)$}\label{fig:soliton-propagating-M1}]{\includegraphics[scale=1]{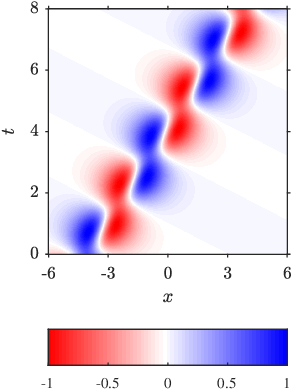}}
\hspace{+0.2cm}\subfigure[{$m_{2}(x,t)$}\label{fig:soliton-propagating-M2}]{\includegraphics[scale=1]{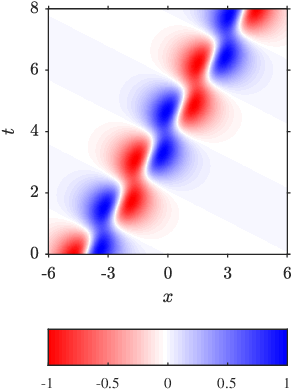}}
\hspace{+0.2cm}\subfigure[{$m_{3}(x,t)$}\label{fig:soliton-propagating-M3}]{\includegraphics[scale=1]{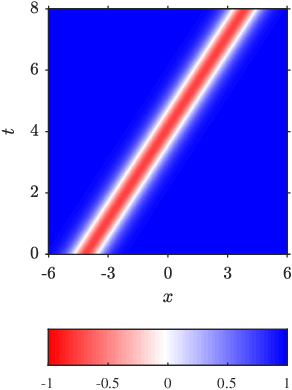}}
\caption{Propagating, one-soliton solution.\label{fig:soliton-propagating}}
\end{figure}
\end{center}
%%%%%%%%%%%%%%%%%%%%%%%%%%%%%%%%%%%%%%%%%%%%%%%%%%%%%%%%%%%%%%%%%%%%%%%%%%%%%%%%%%%%%%%%%%%%%%%%%%%%%%%%%%%%%%%%%%%%%%%%%%%%%%%%%%
%\begin{center}
%\begin{figure}[!ht]
%\hspace{-0.1cm}\subfigure[{$m_{1}(x,t)$}\label{fig:soliton-stationary-M1}]{\includegraphics[scale=1]{soliton_stationary_M1.eps}}
%\hspace{-0.3cm}\subfigure[{$m_{2}(x,t)$}\label{fig:soliton-stationary-M2}]{\includegraphics[scale=1]{soliton_stationary_M2.eps}}
%\hspace{-0.3cm}\subfigure[{$m_{3}(x,t)$}\label{fig:soliton-stationary-M3}]{\includegraphics[scale=1]{soliton_stationary_M3.eps}}
%\caption{Stationary, one-soliton solution.\label{fig:soliton-stationary}}
%\end{figure}
%\end{center}
%%%%%%%%%%%%%%%%%%%%%%%%%%%%%%%%%%%%%%%%%%%%%%%%%%%%%%%%%%%%%%%%%%%%%%%%%%%%%%%%%%%%%%%%%%%%%%%%%%%%%%%%%%%%%%%%%%%%%%%%%%%%%%%%%%

\subsection{Multi-soliton and breather-like solutions}\label{sec:4b}
By combining two or more one-soliton solutions, namely choosing $n>1$, and $n_{j}=1$ for all $j$, $\bar{n}=n$ in \eqref{eq:5.1}, one can easily construct multi-soliton solutions. In this respect, we point out that formulae \eqref{eq:final2} as well as \eqref{eq:final3} are particularly amenable to computer algebra, and allow to obtain explicit expressions, if not for the whole solution in terms of the three components of the magnetization, at least for many specific features of the magnetization dynamics, as shown in the present Section \ref{sec:4}.

For instance, if we choose $n=2$, $n_{1,2}=1$, $\bar{n}=2$, then from \eqref{eq:final3} -- and using the same notation as in this latter formula -- one gets a simple expression for the third component of the magnetization for the two-soliton solution (involving only determinants and traces of $(2\times2)$-matrices):
\begin{equation*}
m_3(x,t)=2\,{\left|\eta(x,t)\right|}^{2}-1\,,\,\,\mbox{ with }\,\,
\eta(x,t)=\tfrac{1-\mathrm{tr}\left(\tilde{N}\,A^{\dagger}\,\tilde{Q}\,A^{-1}\right)+\alpha\,\det\left(\tilde{N}\,\tilde{Q}\right)}{1+\mathrm{tr}\left(\tilde{N}\,\tilde{Q}\right)+\det\left(\tilde{N}\,\tilde{Q}\right)}\,,
\,\alpha=\frac{\det{A^{\dagger}}}{\det{A}}\,.
\end{equation*}
Observe that the same formula holds exactly even without the restriction $n_{j}=1$ for $j=1,2$, namely also in the case of a single two-pole soliton (see Section \ref{sec:4c}) with $n=1$, $n_{1}=2$, $\bar{n}=2$.

In Figure \ref{fig:solitons-02-opposite-direction} we show the head-on collision between two propagating solitons, obtained via \eqref{eq:onesoliton-pq} and \eqref{eq:normingc} with
\begin{align*}
v^{(1)} &= 1\,,&\omega^{(1)} &= 2\,,& x_{0}^{(1)} &= -5\,,& \varphi_{0}^{(1)} = 0\,,&\\
v^{(2)} &= -1\,,&\omega^{(2)} &= 2\,,& x_{0}^{(2)} &= 5\,,& \varphi_{0}^{(2)} = \tfrac{\pi}{2}\,,&
\end{align*}
where the superscript index in brackets identifies the soliton (note that, here and thereafter, the entries of the matrix $B$ are chosen to be ones). Observe the center-of-mass and phase shift due to the nonlinear interaction, see \cite{T, Fogedby1980b}. This phenomenon is particularly evident in
%Figure \ref{fig:solitons-02-propagating-vs-stationary}, showing the scattering between a propagating and a stationary soliton, obtained via %\eqref{eq:onesoliton-pq} and \eqref{eq:normingc} with
%\begin{align*}
%v^{(1)} &= 2\,,& \omega^{(1)} &= 3\,,& x_{0}^{(1)} &= -8\,,& \varphi_{0}^{(1)} &= 0\,,&\\
%v^{(2)} &= 0\,,& \omega^{(2)} &= 2\,,& x_{0}^{(2)} &= 0\,,& \varphi_{0}^{(2)} &= \tfrac{\pi}{2}\,.&
%\end{align*}
%Finally, in
Figure \ref{fig:solitons-03}, where we show the interaction of three propagating solitons,
obtained via \eqref{eq:onesoliton-pq} and \eqref{eq:normingc} with
\begin{align*}
v^{(1)} &= 1.7\,,& \omega^{(1)} &= 5\,,& x_{0}^{(1)} &= -7\,,& \varphi_{0}^{(1)} &= 0\,,&\\
v^{(2)} &= -0.25\,,& \omega^{(2)} &= 4\,,& x_{0}^{(2)} &= 0.25\,,& \varphi_{0}^{(2)} &= 0\,,&\\
v^{(3)} &= -1.8\,,& \omega^{(3)} &= 5.5\,,& x_{0}^{(3)} &= 10\,,& \varphi_{0}^{(3)} &= 0\,.&
\end{align*}
Figures \ref{fig:solitons-02-opposite-direction} and \ref{fig:solitons-03} well illustrate how, for a multi-soliton solution (featuring two or three solitons), the parameters $x_{0}$ and $\varphi_{0}$ in formula \eqref{eq:normingc} provide only an approximation of the initial position and phase, respectively, of each one of the solitons, becoming correct only as the distance that separates the solitons approaches infinity.
%%%%%%%%%%%%%%%%%%%%%%%%%%%%%%%%%%%%%%%%%%%%%%%%%%%%%%%%%%%%%%%%%%%%%%%%%%%%%%%%%%%%%%%%%%%%%%%%%%%%%%%%%%%%%%%%%%%%%%%%%%%%%%%%%%
\begin{center}
\begin{figure}[!ht]
\hspace{-0.0cm}\subfigure[{$m_{1}(x,t)$}\label{fig:solitons-02-opposite-direction-M1}]{\includegraphics[scale=1]{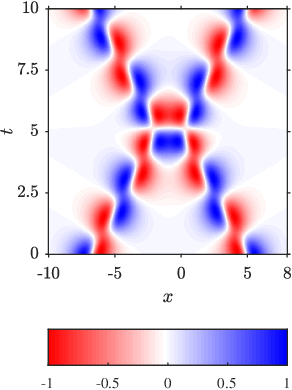}}
\hspace{+0.2cm}\subfigure[{$m_{2}(x,t)$}\label{fig:solitons-02-opposite-direction-M2}]{\includegraphics[scale=1]{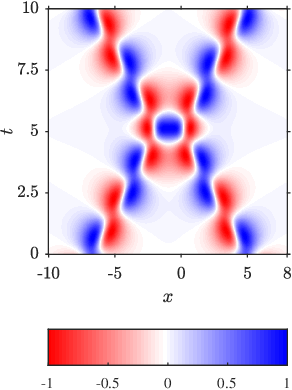}}
\hspace{+0.2cm}\subfigure[{$m_{3}(x,t)$}\label{fig:solitons-02-opposite-direction-M3}]{\includegraphics[scale=1]{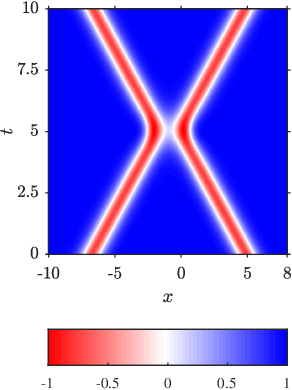}}
\caption{Head-on collision between two solitons propagating in opposite directions.\label{fig:solitons-02-opposite-direction}}
\end{figure}
\end{center}
%%%%%%%%%%%%%%%%%%%%%%%%%%%%%%%%%%%%%%%%%%%%%%%%%%%%%%%%%%%%%%%%%%%%%%%%%%%%%%%%%%%%%%%%%%%%%%%%%%%%%%%%%%%%%%%%%%%%%%%%%%%%%%%%%%
%\begin{center}
%\begin{figure}[!ht]
%\hspace{-0.1cm}\subfigure[{$m_{1}(x,t)$}\label{fig:solitons-02-propagating-vs-stationary-M1}]{\includegraphics[scale=1]{solitons_02_propagating_vs_stationary_M1.eps}}
%\hspace{-0.3cm}\subfigure[{$m_{2}(x,t)$}\label{fig:solitons-02-propagating-vs-stationary-M2}]{\includegraphics[scale=1]{solitons_02_propagating_vs_stationary_M2.eps}}
%\hspace{-0.3cm}\subfigure[{$m_{3}(x,t)$}\label{fig:solitons-02-propagating-vs-stationary-M3}]{\includegraphics[scale=1]{solitons_02_propagating_vs_stationary_M3.eps}}
%\caption{Scattering between a propagating and a stationary soliton. In $m_3$, observe the spatial shift experienced by the stationary soliton (in the %opposite direction with respect to the propagating one) after the interaction (see \cite{T}).\label{fig:solitons-02-propagating-vs-stationary}}
%\end{figure}
%\end{center}
%%%%%%%%%%%%%%%%%%%%%%%%%%%%%%%%%%%%%%%%%%%%%%%%%%%%%%%%%%%%%%%%%%%%%%%%%%%%%%%%%%%%%%%%%%%%%%%%%%%%%%%%%%%%%%%%%%%%%%%%%%%%%%%%%%
\begin{center}
\begin{figure}[!ht]
\hspace{-0.0cm}\subfigure[{$m_{1}(x,t)$}\label{fig:solitons-03-M1}]{\includegraphics[scale=1]{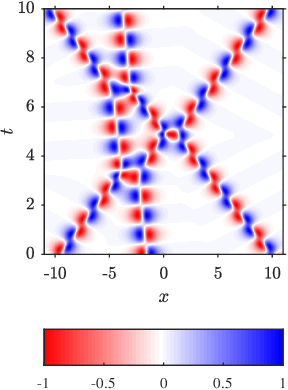}}
\hspace{+0.2cm}\subfigure[{$m_{2}(x,t)$}\label{fig:solitons-03-M2}]{\includegraphics[scale=1]{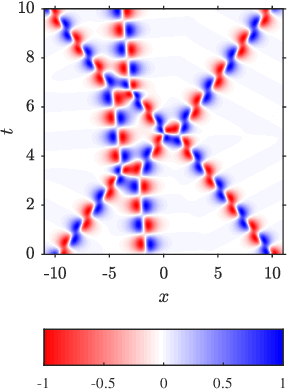}}
\hspace{+0.2cm}\subfigure[{$m_{3}(x,t)$}\label{fig:solitons-03-M3}]{\includegraphics[scale=1]{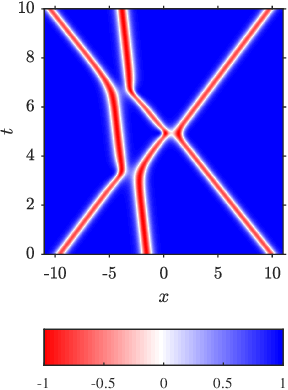}}
\caption{Interaction of three propagating soliton with different velocities. The solitons emerged unchanged from the interaction.\label{fig:solitons-03}}
\end{figure}
\end{center}
%%%%%%%%%%%%%%%%%%%%%%%%%%%%%%%%%%%%%%%%%%%%%%%%%%%%%%%%%%%%%%%%%%%%%%%%%%%%%%%%%%%%%%%%%%%%%%%%%%%%%%%%%%%%%%%%%%%%%%%%%%%%%%%%%%
Breather-like soliton solutions have been also reported in the literature (\textit{e.g.}, see \cite{MukVyaPan2016}), in terms of the observation of a periodic magnetization energy density, whereas here we exhibit the corresponding explicit behaviour of the magnetization dynamics in $x$ and $t$. Breather-like soliton solutions can be constructed out of two-soliton solutions, by creating two stationary, or two same-speed, propagating solitons close to each other (for comparison, observe that such a construction would be less immediate if one uses the solution formula in \cite{BianGuoLing2014}).

As breather-like solitons physically correspond to single states produced by two tangled stationary or propagating solitons (for this reason sometimes called ``tangled states'' or ``bound states''), one has to play also with the norming constants, namely with the components of the matrix $C$. It is worth observing that Corollary 2 in \cite{ZT} says that the norming constants for (\ref{eq:HF1}) can be given in terms of the norming constants featured by the Inverse Scattering Transform scheme for the focussing nonlinear Schr\"{o}dinger equation, all multiplied by the same complex factor, which is the nonlinear Schr\"{o}dinger transmission coefficient computed at the nonlinear Schr\"{o}dinger spectral parameter equal to zero; however, a rescaling of the norming constants is a non-trivial transformation of them, for the rescaling constants appear nonlinearly in the solution formula: as a result, ``tangled states'' (such as breather-like solitons) can be destroyed or created, and it is not forcedly the case that, via the gauge equivalence, a reflectionless breather-like solution of the focussing nonlinear Schr\"{o}dinger equation always corresponds to reflectionless breather-like solution of the Heisenberg ferromagnet equation (if it corresponds to a breather at all).

In the case of two stationary solitons ($v^{(1)}=v^{(2)}=0$), namely, in the case of two real eigenvalues $a^{1} = p_{1}$ and $a_{2} = p_{2}$, $p_1\neq p_2$, it is possible to show that, if the norming constants are chosen as follows
\begin{equation}\label{eq:normingc-sym}
C = \left(c_1,\,c_{2}\right) = 2\,\sqrt{\frac{(p_1+p_2)^2+(q_1-q_2)^2}{(p_1-p_2)^2+(q_1-q_2)^2}}\,\left(p_1,\,p_{2}\right)
\end{equation}
with $q_1=0$ and $q_2=0$, then a single, symmetrical, breather-like soliton solution is created, with $m_{3}$ characterized by two identical, localized minima oscillating in time around the origin with period
\begin{equation}\label{eq:breatherperiod}
\nu = \frac{2\,\pi}{4\,(p_1+p_2)\,(p_1-p_2)}\,.
\end{equation}
Figure \ref{fig:breather-stationary} shows an example of such a breather-like soliton, obtained via \eqref{eq:onesoliton-pq} and \eqref{eq:normingc-sym} with
\begin{equation*}
v^{(1)} = 0\,,\quad \omega^{(1)} = 0.8\,,
\quad\mbox{ and }\quad
v^{(2)} = 0\,,\quad \omega^{(2)} = 0.4\,,
\end{equation*}
thus entailing an oscillation in time with period $\nu \simeq 15.71$.

Propagating, breather-like solitons can be constructed in the same way as above, but assigning the same non-zero imaginary part to both the eigenvalues. For instance, Figure \ref{fig:breather-propagating} shows a propagating, breather-like soliton, moving with velocity $v=0.15$, obtained via \eqref{eq:onesoliton-pq} and \eqref{eq:normingc-sym} with
\begin{equation*}
v^{(1)} = 0.15\,,\quad \omega^{(1)} = 0.8\,,
\quad\mbox{ and }\quad
v^{(2)} = 0.15\,,\quad \omega^{(2)} = 0.4\,.
\end{equation*}
%%%%%%%%%%%%%%%%%%%%%%%%%%%%%%%%%%%%%%%%%%%%%%%%%%%%%%%%%%%%%%%%%%%%%%%%%%%%%%%%%%%%%%%%%%%%%%%%%%%%%%%%%%%%%%%%%%%%%%%%%%%%%%%%%%
\begin{center}
\begin{figure}[!ht]
\hspace{-0.0cm}\subfigure[{$m_{1}(x,t)$}\label{fig:breather-stationary-M1}]{\includegraphics[scale=1]{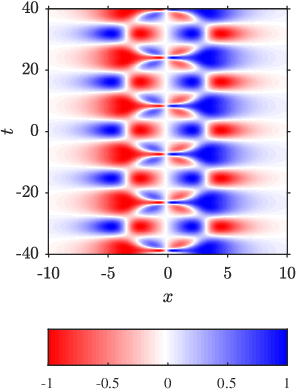}}
\hspace{+0.2cm}\subfigure[{$m_{2}(x,t)$}\label{fig:breather-stationary-M2}]{\includegraphics[scale=1]{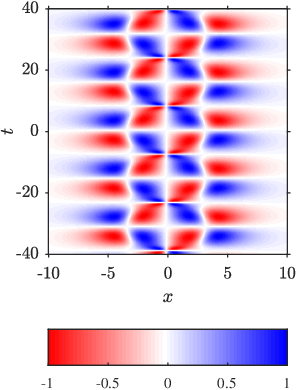}}
\hspace{+0.2cm}\subfigure[{$m_{3}(x,t)$}\label{fig:breather-stationary-M3}]{\includegraphics[scale=1]{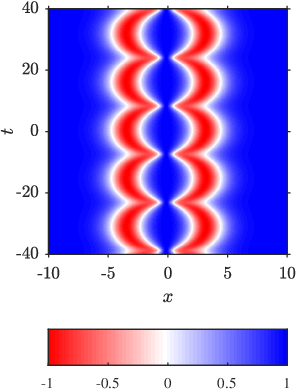}}
\caption{Stationary, breather-like soliton.\label{fig:breather-stationary}}
\end{figure}
\end{center}
%%%%%%%%%%%%%%%%%%%%%%%%%%%%%%%%%%%%%%%%%%%%%%%%%%%%%%%%%%%%%%%%%%%%%%%%%%%%%%%%%%%%%%%%%%%%%%%%%%%%%%%%%%%%%%%%%%%%%%%%%%%%%%%%%%
\begin{center}
\begin{figure}[!ht]
\hspace{-0.0cm}\subfigure[{$m_{1}(x,t)$}\label{fig:breather-propagating-M1}]{\includegraphics[scale=1]{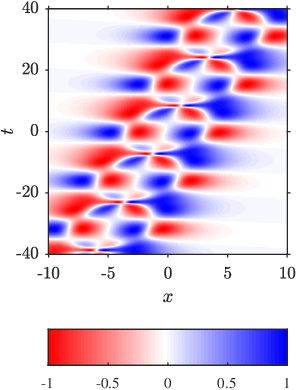}}
\hspace{+0.2cm}\subfigure[{$m_{2}(x,t)$}\label{fig:breather-propagating-M2}]{\includegraphics[scale=1]{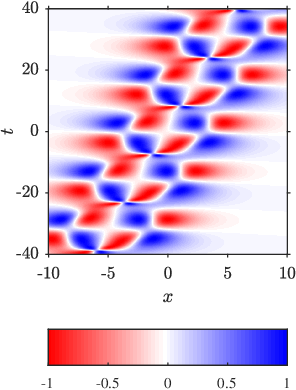}}
\hspace{+0.2cm}\subfigure[{$m_{3}(x,t)$}\label{fig:breather-propagating-M3}]{\includegraphics[scale=1]{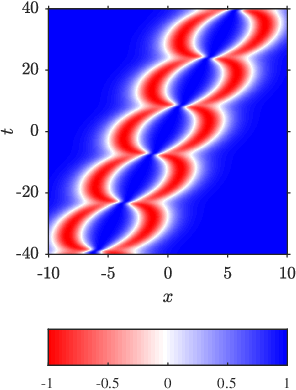}}
\caption{Propagating, breather-like soliton.\label{fig:breather-propagating}}
\end{figure}
\end{center}
%%%%%%%%%%%%%%%%%%%%%%%%%%%%%%%%%%%%%%%%%%%%%%%%%%%%%%%%%%%%%%%%%%%%%%%%%%%%%%%%%%%%%%%%%%%%%%%%%%%%%%%%%%%%%%%%%%%%%%%%%%%%%%%%%%
More generally, as said above, propagating or stationary, breather-like solitons can be had by creating two stationary, or two same-speed, propagating solitons close to each other. Indeed, the two solitons, associated to two different eigenvalues in the matrix $A$, and tangled to form a single, breather-like soliton, always maintain their ``individualities'', and, for instance, can be untangled or transformed into a different breather-like soliton by means of the interaction with another propagating soliton with different speed (see for instance Figure \ref{fig:interaction}): in other words, a single, breather-like soliton should always be regarded as a stable, periodic tangle of two interacting, but individual entities. As an example of this, in Figure \ref{fig:stationary-to-breather} a transition from two separated stationary solitons to a single, breather-like soliton is illustrated (for the sake of brevity, only the third component of the magnetization is given). Figure \ref{fig:stationary-to-breather-01} shows two stationary solitons, located at $x\simeq-7$ and $x\simeq7$, respectively, obtained via \eqref{eq:onesoliton-pq} and \eqref{eq:normingc} with
\begin{align*}
v^{(1)} &= 0\,,& \omega^{(1)} &= 0.8\,,& x_{0}^{(1)} &= -3\,,& \varphi_{0}^{(1)} &= 0\,,&\\
v^{(2)} &= 0\,,& \omega^{(2)} &= 0.4\,,& x_{0}^{(2)} &= 7\,,& \varphi_{0}^{(2)} &= 0\,.&
\end{align*}
Subsequently, only the values of $x_{0}^{(1)}$ and $x_{0}^{(2)}$ are modified, in order to reduce the (average) distance between the two solitons. In Figure \ref{fig:stationary-to-breather-02}, obtained by changing the above values of $x_{0}^{(1)}$ and $x_{0}^{(2)}$ into $x_{0}^{(1)}=0.9$ and $x_{0}^{(2)}=3.72$, the two solitons oscillate around $x\simeq-3.35$ and $x\simeq3.35$. In Figure \ref{fig:stationary-to-breather-03}, obtained by changing the above values of $x_{0}^{(1)}$ and $x_{0}^{(2)}$ into $x_{0}^{(1)}=1.6$ and $x_{0}^{(2)}=3$, the two solitons coalesced into a single, breather-like structure, formed by two entities oscillating around $x\simeq-2.75$ and $x\simeq2.75$.
%%%%%%%%%%%%%%%%%%%%%%%%%%%%%%%%%%%%%%%%%%%%%%%%%%%%%%%%%%%%%%%%%%%%%%%%%%%%%%%%%%%%%%%%%%%%%%%%%%%%%%%%%%%%%%%%%%%%%%%%%%%%%%%%%%
\begin{center}
\begin{figure}[!ht]
\hspace{-0.0cm}\subfigure[{$x_{0}^{(1)}=-3$, $x_{0}^{(2)}=7$}\label{fig:stationary-to-breather-01}]{\includegraphics[scale=1]{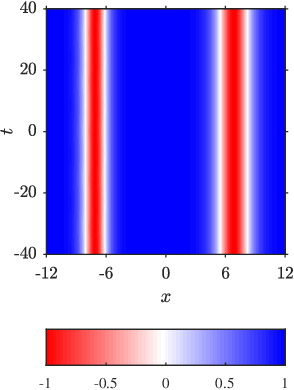}}
\hspace{+0.2cm}\subfigure[{$x_{0}^{(1)}=0.9$, $x_{0}^{(2)}=3.72$}\label{fig:stationary-to-breather-02}]{\includegraphics[scale=1]{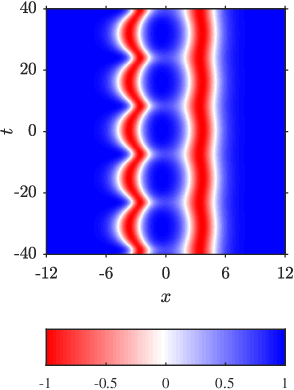}}
\hspace{+0.2cm}\subfigure[{$x_{0}^{(1)}=1.6$, $x_{0}^{(2)}=3$}\label{fig:stationary-to-breather-03}]{\includegraphics[scale=1]{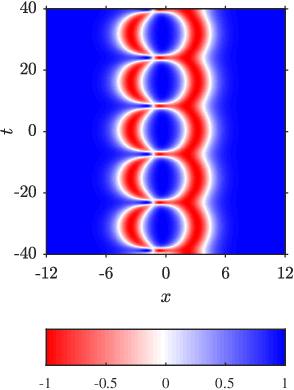}}
\caption{Transition from two stationary solitons to a pair of solitons forming a single stationary breather-like soliton (only $m_{3}(x,t)$ is shown).\label{fig:stationary-to-breather}}
\end{figure}
\end{center}
%%%%%%%%%%%%%%%%%%%%%%%%%%%%%%%%%%%%%%%%%%%%%%%%%%%%%%%%%%%%%%%%%%%%%%%%%%%%%%%%%%%%%%%%%%%%%%%%%%%%%%%%%%%%%%%%%%%%%%%%%%%%%%%%%%
\subsection{Multipole solutions}\label{sec:4c}
If $n_{j}>1$ for some $j$, then $A$ features a Jordan block of order $n_{j}$, and one has multipole soliton solutions.
%{\color{blue} Although multipole solutions of \eqref{eq:HF1} can be obtained by using the solution formula in \cite{BianGuoLing2014} (where multisoliton, multibreathers, and multipole solutions are collectively called high-order solitons), formula  to the best of our knowledge, the general, explicit expression provided here for multiple-pole solitons, which does not require the computation of any auxiliary parameters, is the first one to give their immediate classification, allowing an exact physical description. Furthermore}, let us clarify once more that, even if the Riemann-Hilbert problem for Zakharov-Shabat systems in the case of multiple-poles has been investigated for a long time (\textit{e.g.}, see \cite{ShchesnovichYang2003}), and although multipole solutions of the nonlinear Schr\"{o}dinger equation have been the subject of study for many years (\textit{e.g.}, see \cite{Olmedilla1986, GagnonStievenart1994, Schiebold2014}), nonetheless, as explained in the Introduction, the gauge equivalence \cite{ZT} does not automatically entail that from there one can easily and immediately recover a general, explicit, multipole-soliton solution formula for (\ref{eq:HF1}) {\color{blue}(see \cite{BianGuoLing2014})}.
An accurate analysis of the properties of the multipole soliton solutions, including their asymptotic behaviour, can be done in analogy to the study of the multipole solutions of the nonlinear Schr\"{o}dinger equation \cite{Olmedilla1986, GagnonStievenart1994, Schiebold2014} and of the Hirota equation \cite{DM4}, and is postponed to future investigation. It is worth reminding here that the first complete and rigorous asymptotic analysis of the multiple-pole solutions of the nonlinear Schr\"{o}dinger equation carried out in \cite{Schiebold2014}, about 30 years after their first systematic study in \cite{Olmedilla1986}, was made possible precisely due to the discovery of an explicit solution formula for the reflectionless case, which is the one given by the application of the matrix triplet technique (see \cite{DM0}).

If $n=1$, $n_{1}=2$, $\bar{n}=2$ in \eqref{eq:5.1}, then we have a single two-pole soliton solution. In this case, it is possible to show that, if the associated eigenvalue of $A$ is real ($a=p$), so that $A=\left(\begin{smallmatrix}p&1\\0&p\end{smallmatrix}\right)$, if $B$ is chosen as a vector of ones, and if the norming constants in $C$ are chosen as follows
\begin{equation}\label{eq:normingc-two-pole}
C = \left(c_1,\,c_{2}\right) =\left(4\,p^{2},\,4\,p\,[1+p\,(2\,x_0-1)]\right)\,e^{2\,p\,x_0-i\,\varphi_0}\,,
\end{equation}
then a single, symmetrical, two-pole soliton solution is created, with $m_{3}$ characterized by two minima, constituting two separated branches, that -- in analogy with the multipole solutions of the nonlinear Schr\"{o}dinger equation \cite{Schiebold2014} -- are expected to propagate in space at a velocity that varies logarithmically in time: the two minima are infinitely-apart from each other at $t=-\infty$; for $t<0$, they approach each other as $t$ increases; they interact once in $x_{0}$ at $t=0$; then they separate logarithmically from each other; and finally they are again infinitely-apart from each other at $t=\infty$. The phase on the $(\be_1,\be_2)$ plane at $t=0$ is $\varphi_{0}$. Figure \ref{fig:multipole-02} shows an example of such a solution, obtained via \eqref{eq:normingc-two-pole} with $p=\sqrt{2}$, $x_{0}=0$, and $\varphi_{0}=0$.
%%%%%%%%%%%%%%%%%%%%%%%%%%%%%%%%%%%%%%%%%%%%%%%%%%%%%%%%%%%%%%%%%%%%%%%%%%%%%%%%%%%%%%%%%%%%%%%%%%%%%%%%%%%%%%%%%%%%%%%%%%%%%%%%%%
\begin{center}
\begin{figure}[!ht]
\hspace{-0.0cm}\subfigure[{$m_{1}(x,t)$}\label{fig:multipole-02-M1}]{\includegraphics[scale=1]{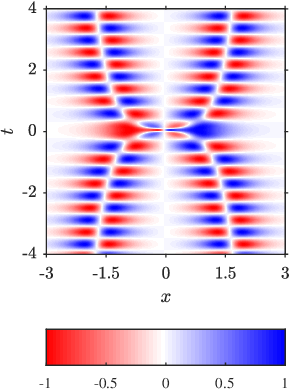}}
\hspace{+0.2cm}\subfigure[{$m_{2}(x,t)$}\label{fig:multipole-02-M2}]{\includegraphics[scale=1]{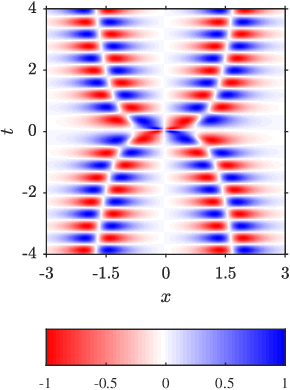}}
\hspace{+0.2cm}\subfigure[{$m_{3}(x,t)$}\label{fig:multipole-02-M3}]{\includegraphics[scale=1]{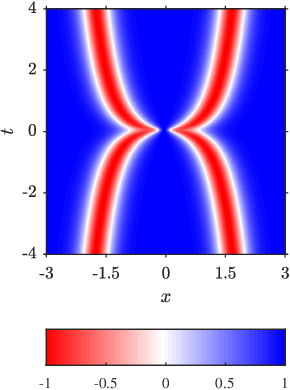}}
\caption{A multipole soliton solution with algebraic multiplicity $n_{1}=2$.\label{fig:multipole-02}}
\end{figure}
\end{center}
%%%%%%%%%%%%%%%%%%%%%%%%%%%%%%%%%%%%%%%%%%%%%%%%%%%%%%%%%%%%%%%%%%%%%%%%%%%%%%%%%%%%%%%%%%%%%%%%%%%%%%%%%%%%%%%%%%%%%%%%%%%%%%%%%%
Then the same technique can be generalized to any value of the algebraic multiplicity $n_{j}$. For instance, if $n=1$, $n_{1}=3$, $\bar{n}=3$ in \eqref{eq:5.1}, then we have a single three-pole soliton solution. Analogously to the previous case, it is possible to show that, if the associated eigenvalue of $A$ is real ($a=p$), so that $A=\left(\begin{smallmatrix}p&1&0\\0&p&1\\0&0&p\end{smallmatrix}\right)$, if $B$ is chosen as a vector of ones, and if the norming constants in $C$ are chosen as follows
\begin{equation}\label{eq:normingc-three-pole}
C^{T} = \left(\begin{array}{c}c_1\\c_{2}\\c_{3}\end{array}\right) =\left(\begin{array}{c}8\,p^{3}\\4\,p^{2}\,[3+p\,(4\,x_0-2)]\\8\,p^{2}\,x_0\,(x_0-1)+6\,p\,(2\,x_0-1)+3\end{array}\right)\,e^{2\,p\,x_0-i\,\varphi_0}\,,
\end{equation}
then a single, symmetrical, three-pole soliton solution is created, with $m_{3}$ characterized by three minima, constituting three separated branches, propagating in space at a velocity that varies logarithmically in time, and interacting in $x=x_0$ at $t=0$. Figure \ref{fig:multipole-03} shows an example of such a solution, obtained via \eqref{eq:normingc-three-pole} with $p=1$, $x_{0}=0$, and $\varphi_{0}=0$.
%%%%%%%%%%%%%%%%%%%%%%%%%%%%%%%%%%%%%%%%%%%%%%%%%%%%%%%%%%%%%%%%%%%%%%%%%%%%%%%%%%%%%%%%%%%%%%%%%%%%%%%%%%%%%%%%%%%%%%%%%%%%%%%%%%
\begin{center}
\begin{figure}[!ht]
\hspace{-0.0cm}\subfigure[{$m_{1}(x,t)$}\label{fig:multipole-03-M1}]{\includegraphics[scale=1]{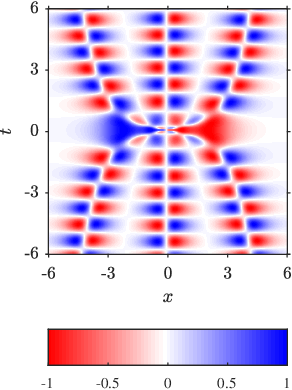}}
\hspace{+0.2cm}\subfigure[{$m_{2}(x,t)$}\label{fig:multipole-03-M2}]{\includegraphics[scale=1]{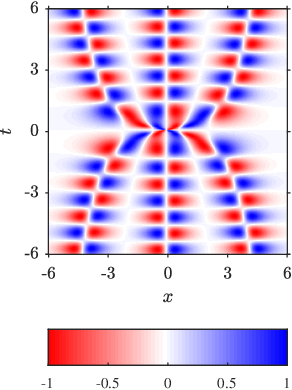}}
\hspace{+0.2cm}\subfigure[{$m_{3}(x,t)$}\label{fig:multipole-03-M3}]{\includegraphics[scale=1]{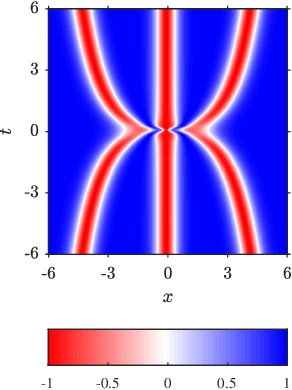}}
\caption{A multipole soliton solution with algebraic multiplicity $n_{1}=3$.\label{fig:multipole-03}}
\end{figure}
\end{center}
%%%%%%%%%%%%%%%%%%%%%%%%%%%%%%%%%%%%%%%%%%%%%%%%%%%%%%%%%%%%%%%%%%%%%%%%%%%%%%%%%%%%%%%%%%%%%%%%%%%%%%%%%%%%%%%%%%%%%%%%%%%%%%%%%%
A symmetrical four-pole soliton solution, with the four branches coalescing at $x=0$ for $t=0$, can be generated by choosing $n=1$, $n_{1}=4$, $\bar{n}=4$ in \eqref{eq:5.1}, $A$ a Jordan block of order 4, $B$ a vector of ones, and the norming constants in $C$ as follows
\begin{equation}\label{eq:normingc-four-pole}
C^{T} = \left(\begin{array}{c}c_1\\c_{2}\\c_{3}\\c_{4}\end{array}\right) =\left(\begin{array}{c}
16\,p^4\\
16\,p^3\,(2-p)\\
8\,p^2\,(3-4\,p)\\
p\,(8-24\,p)
\end{array}\right)\,,
\end{equation}
with $p=\tfrac{\sqrt{3}}{2}$.
%%%%%%%%%%%%%%%%%%%%%%%%%%%%%%%%%%%%%%%%%%%%%%%%%%%%%%%%%%%%%%%%%%%%%%%%%%%%%%%%%%%%%%%%%%%%%%%%%%%%%%%%%%%%%%%%%%%%%%%%%%%%%%%%%%
%\begin{center}
%\begin{figure}[!ht]
%\hspace{-0.1cm}\subfigure[{$m_{1}(x,t)$}\label{fig:multipole-04-M1}]{\includegraphics[scale=1]{multipole_04_M1.eps}}
%\hspace{-0.3cm}\subfigure[{$m_{2}(x,t)$}\label{fig:multipole-04-M2}]{\includegraphics[scale=1]{multipole_04_M2.eps}}
%\hspace{-0.3cm}\subfigure[{$m_{3}(x,t)$}\label{fig:multipole-04-M3}]{\includegraphics[scale=1]{multipole_04_M3.eps}}
%\caption{A multipole soliton solution with algebraic multiplicity $n_{1}=4$.\label{fig:multipole-04}}
%\end{figure}
%\end{center}
%%%%%%%%%%%%%%%%%%%%%%%%%%%%%%%%%%%%%%%%%%%%%%%%%%%%%%%%%%%%%%%%%%%%%%%%%%%%%%%%%%%%%%%%%%%%%%%%%%%%%%%%%%%%%%%%%%%%%%%%%%%%%%%%%%

Multipole soliton solutions can be made to ``propagate'' by adding a non-zero imaginary part to the eigenvalue $a$, so that the centre-of-mass, that is the geometric centre of the location of the minima of $m_3$, travels at a constant speed (given by four times the value of the imaginary part). An identical phenomenon is observed for the nonlinear Schr\"{o}dinger equation, see \cite{Schiebold2014}.

An example of interaction between a single soliton, a breather-like soliton and a two-pole soliton is shown in Figure \ref{fig:interaction}, corresponding to the following choice of the matrix triplet $(A,B,C)$:
\begin{footnotesize}
\begin{align*}
A &= \left(\begin{array}{ccccc}
p_1+i\,q_1&0&0&0&0\\
0&p_2+i\,q_2&0&0&0\\
0&0&p_3+i\,q_3&0&0\\
0&0&0&p_4+i\,q_4&1\\
0&0&0&0&p_4+i\,q_4\\
\end{array}\right)
\,,\quad
B = \left(\begin{array}{c}1\\1\\1\\1\\1\end{array}\right)\,,\\
C^{T} &= \left(\begin{array}{c}
2\,\sqrt{\frac{(p_1+p_2)^2+(q_1-q_2)^2}{(p_1-p_2)^2+(q_1-q_2)^2}}\,p_1\\
2\,\sqrt{\frac{(p_1+p_2)^2+(q_1-q_2)^2}{(p_1-p_2)^2+(q_1-q_2)^2}}\,p_2\\
2\,i\,p_3\,\left(\frac{p_3+i\,q_3}{p_3-i\,q_3}\right)\,e^{2\,(p_3+i\,q_3)\,x_{0}^{(3)}-i\,\varphi_{0}^{(3)}}\\
4\,p_{4}^{2}\,e^{2\,p_{4}\,x_0^{(4)}-i\,\varphi_0^{(4)}}\\
4\,p_{4}\,\left[1+p_{4}\,(2\,x_0^{(4)}-1)\right]\,e^{2\,p_{4}\,x_0^{(4)}-i\,\varphi_0^{(4)}}
\end{array}\right)\,,
\end{align*}
\end{footnotesize}
where $p_j$ and $q_{j}$, for j=1,2,3,4, are obtained via \eqref{eq:onesoliton-parameters} with
\begin{align*}
v^{(1)} &= 0\,,& \omega^{(1)} &= 3.6\,,& x_{0}^{(1)} &= 0\,,& \varphi_{0}^{(1)} &= 0\,,&\\
v^{(2)} &= 0\,,& \omega^{(2)} &= 1\,,& x_{0}^{(2)} &= 0\,,& \varphi_{0}^{(2)} &= 0\,,&\\
v^{(3)} &= 1.75\,,& \omega^{(3)} &= 3\,,& x_{0}^{(3)} &= -4\,,& \varphi_{0}^{(3)} &= \tfrac{3\,\pi}{4}\,,&\\
v^{(4)} &= 0\,,& \omega^{(4)} &= 2.9\,,& x_{0}^{(4)} &= 6\,,& \varphi_{0}^{(4)} &= 0\,.&
\end{align*}
Note that, with the above choice of values for the parameters, we have $q_{1}=q_{2}=q_{4}=0$, and $C_{1}=2\,\tfrac{p_1+p_2}{p_1-p_2}\,p_1$, $C_{2}=2\,\tfrac{p_1+p_2}{p_1-p_2}\,p_2$. In Figure \ref{fig:interaction}, observe that, along with a spatial and phase shift, the stationary breather-like soliton (associated to the parameters labelled 1 and 2 in the above table) experiences a change in its structure because of the interaction with the propagating soliton (associated to the parameters labelled 3 in the above table), whereas the two-pole soliton (associated to the parameters labelled 4 in the above table) appears to be only shifted in space and phase after the same interaction (one branch at a time, in chronological order of interaction).
%%%%%%%%%%%%%%%%%%%%%%%%%%%%%%%%%%%%%%%%%%%%%%%%%%%%%%%%%%%%%%%%%%%%%%%%%%%%%%%%%%%%%%%%%%%%%%%%%%%%%%%%%%%%%%%%%%%%%%%%%%%%%%%%%%
\begin{center}
\begin{figure}[!ht]
\hspace{-0.0cm}\subfigure[{$m_{1}(x,t)$}\label{fig:interaction-M1}]{\includegraphics[scale=1]{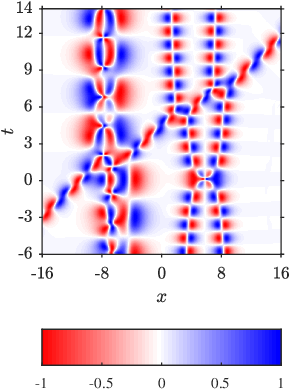}}
\hspace{+0.2cm}\subfigure[{$m_{2}(x,t)$}\label{fig:interaction-M2}]{\includegraphics[scale=1]{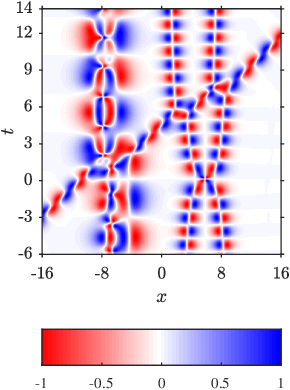}}
\hspace{+0.2cm}\subfigure[{$m_{3}(x,t)$}\label{fig:interaction-M3}]{\includegraphics[scale=1]{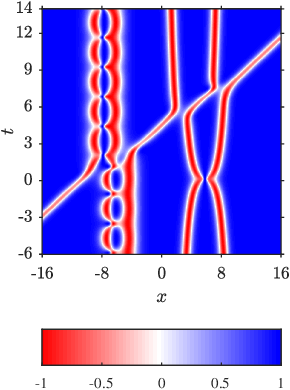}}
\caption{Interaction of a propagating soliton with a breather-like soliton and a multipole soliton with algebraic multiplicity 2.\label{fig:interaction}}
\end{figure}
\end{center}
%%%%%%%%%%%%%%%%%%%%%%%%%%%%%%%%%%%%%%%%%%%%%%%%%%%%%%%%%%%%%%%%%%%%%%%%%%%%%%%%%%%%%%%%%%%%%%%%%%%%%%%%%%%%%%%%%%%%%%%%%%%%%%%%%%
Aside, we report the remarkable fact that, when $x_{0}=0$ and the norming constants are given by \eqref{eq:normingc-two-pole} or \eqref{eq:normingc-three-pole}, or when the norming constants are as in \eqref{eq:normingc-four-pole}, then the two matrices $N$ and $Q$ commute.

\section{Conclusions and outlooks}\label{sec:5}
In this paper we have rigorously developed a novel approach for the IST for the classical, continuous Heisenberg ferromagnetic chain equation \eqref{eq:HF1}, making a consistent advancement on the theory developed in \cite{T, ZT, BianGuoLing2014}. Our working hypotheses (Assumptions \ref{HP1} and \ref{HP2}) are less restrictive than those assumed in \cite{T, ZT}. Hypotheses analogous to ours are employed also in \cite{BianGuoLing2014}, where the inverse scattering theory for \eqref{eq:HF1} is developed exploiting the gauge equivalence to the nonlinear Schr\"{o}dinger equation established in \cite{ZT} and by solving the corresponding Riemann-Hilbert problem. On the contrary, we have used directly the Marchenko equations to reconstruct the potential function. In our treatment, we have proved the analyticity properties of the eigenfunctions and the scattering data. In order to derive these results, we have used a new triangular representation of the Jost solutions (see equations \eqref{eq:final1} and \eqref{eq:limtaurho}), which eases the study of their asymptotic behaviour as well as of the scattering data in the far range of the spectral parameter $\zl$. In doing so, we fixed some existing issues in the literature. Moreover, by using the matrix triplet method, we have found a new, explicit multi-soliton solution formula for equation \eqref{eq:HF1}, which includes and allows the immediate classification of all the soliton solutions already known \cite{NakSas1974, LakRuiTho1976, Tjon, Wang2005, SalHas2009, BianGuoLing2014, ChenWan2014}, and, in particular, general, explicit expressions for the breather-like and multipole solutions (see Section \ref{sec:4}). This formula is particularly amenable of computer algebra and allows to easily read relevant physical information for the phenomena modeled by the corresponding solutions.

The work presented in this paper is part of a larger research programme recently launched and aimed at understanding integrable, continuous, $(1+1)$-dimensional models of ferromagnetism at the nanometer length scale, and especially at finding in closed-form and classifying all the non-topological, propagating, localized solutions of the one-dimensional, continuous Landau-Lifshitz equation with any kind of magnetic anisotropy (the continuous Heisenberg ferromagnet equation corresponds to the magnetically-isotropic Landau-Lifshitz equation). In particular, the results presented in this paper pave the way to a similar, immediate application of the IST to the Landau-Lifshitz equation with uni-axial (easy-axis and easy-plane) anisotropy (see \cite{BK, BKK, Mik, Zong}). Indeed, by adopting there the same triangular representation of the Jost solutions as considered here for the Heisenberg ferromagnet equation, we expect to overcome many difficulties currently featured by the IST machinery for the uni-axial Landau-Lifshitz equation, thereby establishing the asymptotic behavior of the scattering data for large values of the spectral parameter $\zl$, as well as providing explicit expressions for new classes of soliton solutions.

Finally, bearing in mind the connection illustrated in \cite{IaDuBoMoChHoAk2014} between low-dimensional magnetic-droplet solitons and current experiments on the magnetization dynamics in ferromagnetic nanowires, the results of the present paper are noteworthy in view of the potential applications in nanomagnetism and spintronics.

\section*{Acknowledgments}
This research has been partially supported by the University of Cagliari and the Regione Autonoma Sardegna in the framework of the ``Visiting Professor Call 2015'', which made possible the stay of MS at the Department of Mathematics and Computer Science of the University of Cagliari during the Spring Semester 2015, and by GNFM-INdAM (Gruppo Nazionale per la Fisica Matematica, National Group for Mathematical Physics -- Istituto Nazionale di Alta Matematica, National Institute of Advanced Mathematics). MS wishes to express his gratitude for the hospitality of the Department of Mathematics and Computer Science of the University of Cagliari, where part of this article has been written. MS wishes also to thank Northumbria University for allowing him to undertake the said visiting programme abroad.

FV wishes to express her gratitude for the hospitality of the Department of Mathematics and Information Sciences of Northumbria University during her visit in February 2016.

\appendix
\section{Proof of the Marchenko equations}\label{sec:A} %\eqref{eq:4.10}
In this appendix we derive the Marchenko integral equations assuming that all the discrete eigenvalues are simple (the general case when the discrete eigenvalues are not simple can be treated as in \cite{D}).
%To get this result, let us now write out \eqref{eq:4.1} in scalar form. We get
%\begin{subequations}\label{eq:4.6}
%\begin{align}
%{e^{-i\zl x}\overline{\phi}(x,\zl)}
%&={T(\zl)}{e^{-i\zl x}\psi(x,\zl)}
%-{e^{-2i\zl x}L(\zl)}{e^{i\zl x}\phi(x,\zl)},\label{eq:4.6a}\\
%{e^{i\zl x}\overline{\psi}(x,\zl)}
%&=-{e^{2i\zl x}R(\zl)}{e^{-i\zl x}\psi(x,\zl)}
%+{T(\zl)}{e^{i\zl x}\phi(x,\zl)},\label{eq:4.6b}\\
%{e^{-i\zl x}\psi(x,\zl)}
%&={T(\zl)^*}{e^{-i\zl x}\overline{\phi}(x,\zl)}
%+{e^{-2i\zl x}R(\zl)^*}{e^{i\zl x}\overline{\psi}(x,\zl)},\label{eq:4.6c}\\
%{e^{i\zl x}\phi(x,\zl)}
%&={e^{2i\zl x}L(\zl)^*}{e^{-i\zl x}\overline{\phi}(x,\zl)}
%+{T(\zl)^*}{e^{i\zl x}\overline{\psi}(x,\zl)}.\label{eq:4.6d}
%\end{align}
%\end{subequations}
Recalling (\ref{eq:3.11b}) with (\ref{eq:3.13b}),
$$
\bH^{up}(x)=\begin{pmatrix}H^{up}_1(x)&-H^{up}_2(x)^*\\H^{up}_2(x)&H^{up}_1(x)^*\end{pmatrix},\qquad
{\bH^{dn}}(x)=\begin{pmatrix}H^{dn}_{1}(x)^*&H^{dn}_{2}(x)\\
-H^{dn}_{2}(x)^*&H^{dn}_{1}(x)\end{pmatrix}\,,
$$
and setting
$$
T(\zl)=T_0(\zl)+\sum_j\,\frac{\theta_j}{\zl-ia_j}
$$
for some function $T_0(\zl)$ that is continuous in $\zl\in\overline{\C^+}$, is analytic in $\zl\in\C^+$, and tends to $e^{-i\za}$ as $\zl\to\infty$ from within $\overline{\C^+}$, we obtain from \eqref{eq:4.2} with \eqref{eq:4.4}
\begin{footnotesize}
\begin{align*}
&{\begin{pmatrix}-H^{up}_2(x)^*\\H^{up}_1(x)^*\end{pmatrix}}+{\int_0^\infty \mathrm{d}\xi\,
e^{-i\zl \xi}\begin{pmatrix}-K^{up}_2(x,x+\xi)^*\\ K^{up}_1(x,x+\xi)^*\end{pmatrix}}
={e^{-i\za}\begin{pmatrix}H^{dn}_{2}(x)\\ H^{dn}_{1}(x)\end{pmatrix}}
\nonumber\\ &+{T_0(\zl)e^{i\zl x}\phi(x,\zl)-e^{-i\za}\begin{pmatrix}
H^{dn}_{2}(x)\\ H^{dn}_{1}(x)\end{pmatrix}}
+{\sum_j\,\theta_j\frac{e^{i\zl x}\phi(x,\zl)
-e^{-a_jx}\phi(x,ia_j)}{\zl-ia_j}}\nonumber\\
&+i\sum_j\,\frac{c_{j}\,e^{-a_jx}}{\zl-ia_j}\,\Bigg\{\begin{pmatrix}H^{up}_1(x)\\H^{up}_2(x)
\end{pmatrix}+\int_0^\infty \mathrm{d}\xi\,e^{-a_j\xi}\begin{pmatrix}K^{up}_1(x,x+\xi)\\\
K^{up}_2(x,x+\xi)\end{pmatrix}\Bigg\}\nonumber\\
&-\left(\int_{-\infty}^\infty \mathrm{d}\zeta\,e^{-i\zl \zeta}\hat{R}(\zeta+2x)\right)\Bigg\{
\begin{pmatrix} H^{up}_1(x)\\ H^{up}_2(x)\end{pmatrix}+\int_0^\infty \mathrm{d}\xi\,
e^{i\zl \xi}\begin{pmatrix}K^{up}_1(x,x+\xi)\\ K^{up}_2(x,x+\xi)\end{pmatrix}\Bigg\},
\end{align*}
\end{footnotesize}
where $\hat{R}$ is defined as in Proposition \ref{Preflection}. Considering the limits as $\zl\to\pm\infty$, we obtain
\begin{equation}\label{eq:4.7}
\bH^{up}(x)={\bH^{dn}}(x)\,e^{i\za\zs_3},
\end{equation}
where $\tau(\zl)\to e^{i\za}$ as $\zl\to\pm\infty$. Using the identity
$$
-\frac{i}{\zl-ia_j}={\int_0^\infty \mathrm{d}\xi\,e^{-i\zl \xi}e^{-a_j\xi}}\,,
$$
and focussing in the above expression on the terms of the form $\int_0^\infty \mathrm{d}\xi\,e^{-i\zl \xi}[\ldots]$, stripping off the Fourier transform, we get
\begin{footnotesize}
\begin{align*}
\begin{pmatrix}-K^{up}_2(x,x+\xi)^*\\ K^{up}_1(x,x+\xi)^*\end{pmatrix}&=-\sum_j\,
c_j\,e^{-a_j(\xi+2x)}\,\begin{pmatrix}H^{up}_1(x)\\H^{up}_2(x)\end{pmatrix}\nonumber\\
&\quad-\int_0^\infty \mathrm{d}\zeta\,c_j\,e^{-a_j(\xi+\zeta+2x)}\,\begin{pmatrix}K^{up}_1(x,x+\zeta)\\
K^{up}_2(x,x+\zeta)\end{pmatrix}\nonumber\\
&\quad-\hat{R}(\xi+2x)\begin{pmatrix}H^{up}_1(x)\\ H^{up}_2(x)\end{pmatrix}
-\int_0^\infty \mathrm{d}\zeta\,\hat{R}(\zeta+\xi+2x)\begin{pmatrix}K^{up}_1(x,x+\zeta)\\K^{up}_2(x,x+\zeta)
\end{pmatrix}.
\end{align*}
\end{footnotesize}

Analogously, we have
%that the Riemann-Hilbert problem \eqref{eq:4.6c} implies
\begin{footnotesize}
\begin{align*}
&{\begin{pmatrix}H^{up}_1(x)\\H^{up}_2(x)\end{pmatrix}}+{\int_0^\infty \mathrm{d}\xi\,
e^{i\zl \xi}\begin{pmatrix}K^{up}_1(x,x+\xi)\\K^{up}_2(x,x+\xi)\end{pmatrix}}={e^{i\za}
\begin{pmatrix}H^{dn}_{1}(x)^*\\-H^{dn}_{2}(x)^*\end{pmatrix}}\nonumber\\
&+{T_0(\zl^*)^*e^{-i\zl x}\overline{\phi}(x,\zl)
-e^{i\za}\begin{pmatrix}H^{dn}_{1}(x)^*\\-H^{dn}_{2}(x)^*\end{pmatrix}}
+{\sum_j\,\theta_j^*\frac{e^{-i\zl x}\overline{\phi}(x,\zl)-e^{-a_j^*x}
\overline{\phi}(x,-ia_j^*)}{\zl+ia_j^*}}\nonumber\\
&-i\sum_j\,\frac{\overline{c}_j\,e^{-2a_j^*x}}{\zl+ia_j^*}\,\Bigg\{\begin{pmatrix}
-H^{up}_2(x)^*\\ H^{up}_1(x)^*\end{pmatrix}+\int_0^\infty \mathrm{d}\xi\,\sum_j\,e^{-a_j^*\xi}
\begin{pmatrix}-K^{up}_2(x,x+\xi)^*\\K^{up}_1(x,x+\xi)^*\end{pmatrix}\Bigg\}\nonumber\\
&+\left(\int_{-\infty}^\infty \mathrm{d}\zeta\,e^{i\zl \zeta}\hat{R}(\zeta+2x)^*\right)\Bigg\{
\begin{pmatrix}-H^{up}_2(x)^*\\H^{up}_1(x)^*\end{pmatrix}+\int_0^\infty \mathrm{d}\xi\,e^{-i\zl \xi}
\begin{pmatrix}-K^{up}_2(x,x+\xi)^*\\K^{up}_1(x,x+\xi)^*\end{pmatrix}\Bigg\}.
\end{align*}
\end{footnotesize}
Using the identity
$$
\frac{i}{\zl+ia_j^*}={\int_0^\infty \mathrm{d}\xi\,e^{i\zl \xi}e^{-a_j^*\xi}}\,,
$$
and focussing in the above expression on the terms of the form $\int_0^\infty \mathrm{d}\xi\,e^{i\zl \xi}[\ldots]$, stripping off the Fourier transform, we get
\begin{footnotesize}
\begin{align*}
\begin{pmatrix}K^{up}_1(x,x+\xi)\\K^{up}_2(x,x+\xi)\end{pmatrix}
&=-\sum_j\,\overline{c}_j\,e^{-a_j^*(\xi+2x)}\,\begin{pmatrix}-H^{up}_2(x)^*\\H^{up}_1(x)^*
\end{pmatrix}\nonumber\\
&\quad-\int_0^\infty \mathrm{d}\zeta\,\sum_j\,\overline{c}_j\,e^{-a_j^*(\xi+\zeta+2x)}\begin{pmatrix}
-K^{up}_2(x,x+\zeta)^*\\ K^{up}_1(x,x+\zeta)^*\end{pmatrix}\nonumber\\
&\quad+\hat{R}(\xi+2x)^*
\begin{pmatrix}-H^{up}_2(x)^*\\H^{up}_1(x)^*\end{pmatrix}+\int_0^\infty \mathrm{d}\zeta\,
\hat{R}(\zeta+\xi+2x)^*\begin{pmatrix}-K^{up}_2(x,x+\zeta)^*\\K^{up}_1(x,x+\zeta)^*\end{pmatrix}.
\end{align*}
\end{footnotesize}

Setting $y=x+\xi\ge x$, we obtain the system of coupled Marchenko integral equations
%\begin{footnotesize}
%\begin{subequations}\label{eq:4.8}
%\begin{align}
%&\begin{pmatrix}K^{up}_1(x,y)\\ K^{up}_2(x,y)\end{pmatrix}-\Omega(x+y)^*\begin{pmatrix}
%-H^{up}_2(x)^*\\H^{up}_1(x)^*\end{pmatrix}-\int_x^\infty \mathrm{d}\xi\,\Omega(\xi+y)^*\begin{pmatrix}
%-K^{up}_2(x,\xi)^*\\ K^{up}_1(x,\xi)^*\end{pmatrix}=0_{2\times1},\label{4.8a}\\
%&\begin{pmatrix}-K^{up}_2(x,y)^*\\ K^{up}_1(x,y)^*\end{pmatrix}+\Omega(x+y)\begin{pmatrix}
%H^{up}_1(x)\\H^{up}_2(x)\end{pmatrix}+\int_x^\infty \mathrm{d}\xi\,\Omega(\xi+y)\begin{pmatrix}
%K^{up}_1(x,\xi)\\ K^{up}_2(x,\xi)\end{pmatrix}=0_{2\times1},\label{4.8b}
%\end{align}
%\end{subequations}
%\end{footnotesize}
%where (see Proposition \ref{P4}) $\overline{c}_j=-(c_j)^*$ and
%\begin{equation}\label{eq:4.9}
% \Omega(x)=\hat{R}(x)+\sum_{j=1}^{n}\,c_j\,e^{-a_jx}\,.
%\end{equation}
%In other words, we obtain the Marchenko integral equations
\eqref{eq:4.10}:
\begin{equation*}
\bK^{up}(x,y)+\bH^{up}(x)\,\bO(x+y)+\int_x^\infty \mathrm{d}\xi\,\bK^{up}(x,\xi)\bO(\xi+y)=0_{2\times2}\,\,,
\end{equation*}
where
$\bO(x)=\begin{pmatrix}0&\Omega(x)\\-\Omega(x)^*&0\end{pmatrix}$.

Analogously, for $\bK^{dn}$ and $\bH^{dn}$, one can prove that
\begin{subequations}\label{eq:4.10dn}
\begin{equation}\label{eq:4.10dna}
\bK^{dn}(x,y)+\bH^{dn}(x)\,\overline{\bO}(x+y)+\int_{-\infty}^{x} \mathrm{d}\xi\,\bK^{dn}(x,\xi)\overline{\bO}(\xi+y)=0_{2\times2},
\end{equation}
with
\begin{equation}\label{eq:4.10dnb}
\overline{\bO}(x)=\begin{pmatrix}0&\overline{\Omega}(x)\\-\overline{\Omega}(x)^*&0\end{pmatrix}
\quad\mbox{ and }\quad
\overline{\Omega}(x)=\hat{L}(x)+\sum_{j=1}^{n}\,\overline{c}_j\,e^{-a_j^*x}\,,
\end{equation}
\end{subequations}
where $\hat{L}$ is defined as in Proposition \ref{Preflection}.

%\section{Alternative expression for the solution of \eqref{eq:HF1}}\label{sec:C}
%In this appendix we present the relation between the alternative Marchenko integral equation \eqref{eq:4.10dn} and the solution of the initial value problem \eqref{eq:HF1}.

We conclude this appendix by observing that an alternative expression for the solution of \eqref{eq:HF1} can be obtained as follows.
Setting
\begin{equation}\label{eq:Ldn}
\bK^{dn}(x,y)=\bH^{dn}(x)\,\overline{\bL}(x,y)\,,
\end{equation}
we arrive at the alternative Marchenko integral equation for $\overline{\bL}(x,y)$,
\begin{equation}\label{eq:4.11dn}
\overline{\bL}(x,y)+\overline{\bO}(x+y)+\int_{-\infty}^x \mathrm{d}\xi\,\overline{\bL}(x,\xi)\,\overline{\bO}(\xi+y)=0_{2\times2},
\end{equation}
which is the analogue for $\overline{\bL}(x,y)$ of equation \eqref{eq:4.11} for $\bL(x,y)$.

Let $\widetilde{\overline{\bK}}(x)$ and $\widetilde{\overline{\bL}}(x)$ be defined as follows (see equation \eqref{eq:notazione}):
\begin{equation}\label{eq:notazionedn}
\widetilde{\overline{\bK}}(x)=\int_{-\infty}^{x} \mathrm{d}\xi\,\bK^{dn}(x,\xi)\,,\quad
\widetilde{\overline{\bL}}(x)=\int_{-\infty}^{x} \mathrm{d}\xi\,\overline{\bL}(x,\xi)\,.
\end{equation}
Then, from \eqref{eq:2.1b} and \eqref{eq:3.12} we easily get the analogue for $\bH^{dn}(x)$ and $\widetilde{\overline{\bL}}(x)$ of relation \eqref{eq:3.17} for $\bH^{up}(x)$ and $\tilde{\bL}(x)$,
\begin{subequations}\label{eq:3.17dn}
\begin{equation}\label{eq:3.17dna}
I_2=\Phi(x,0)=\bH^{dn}(x)+\widetilde{\overline{\bK}}(x)=\bH^{dn}(x)\,\left[I_2+\widetilde{\overline{\bL}}(x)\right]\,,
\end{equation}
where
\begin{equation}\label{eq:3.17dnb}
\widetilde{\overline{\bL}}(x)={\bH^{dn}(x)}^{-1}\,\,\widetilde{\overline{\bK}}(x)\,.
\end{equation}
\end{subequations}
Equation \eqref{eq:3.17dn} allows us to establish the analogue for $\widetilde{\overline{\bL}}(x)$ of equation \eqref{eq:3.19} for $\tilde{\bL}(x)$, namely to write the solution of the initial value problem \eqref{eq:HF1} in the alternative form
\begin{equation}\label{eq:3.19dn}
\bm(x)\cdot\bzs =\bH^{dn}(x)\,\,\zs_3\,\,{\bH^{dn}(x)}^{-1}
=\left[I_2+{\widetilde{\overline{\bL}}(x)}^{\dagger}\right]\, \zs_3\, \bigg[I_2+\widetilde{\overline{\bL}}(x)\bigg]\,.
\end{equation}

\section{Explicit expressions for $\bm(x,t)$}\label{sec:B}
In the present appendix, we provide explicit expressions for $\bm(x,t)$ in terms of a matrix triplet in forms that are amenable to computer algebra and effective for actual numerical evaluation. In view of the following, we recall that, for inverting a generic invertible block matrix $U$ in the form
\begin{equation*}
U=\begin{pmatrix}U_{11}&U_{12} \\ U_{21}&U_{22} \end{pmatrix}\,,
\end{equation*}
whose diagonal blocks are both invertible, the rule is (see \emph{e.g.} \cite{GvL1983})
\begin{footnotesize}
\begin{equation*}
U^{-1}=\begin{pmatrix}U_{11}&U_{12} \\ U_{21}&U_{22} \end{pmatrix}^{-1}=
\begin{pmatrix}(U_{11}-U_{12}U^{-1}_{22}U_{21})^{-1}&-U_{11}^{-1}U_{12}(U_{22}-U_{21}U_{11}^{-1}U_{12})^{-1}\\ -U_{22}^{-1}U_{21}(U_{11}-U_{12}U_{22}^{-1}U_{21})^{-1}&(U_{22}-U_{21}U^{-1}_{11}U_{12})^{-1} \end{pmatrix}\,.
\end{equation*}
\end{footnotesize}
Applying this relation to the inverse matrix in the right-hand side of \eqref{eq:final1},
\begin{equation*}
\left[e^{2x\CA}+\CP(t)\right]^{-1}=\begin{pmatrix}e^{2xA}&N\\-Q(t)&e^{2xA^{\dagger}}\end{pmatrix}^{-1},
\end{equation*}
and recalling the structure of $\tilde{\bL}$ in \eqref{eq:Lstructure},
\begin{equation*}
\tilde{\bL}(x)
=\begin{pmatrix}{\tilde{L}_{11}(x;t)}&{\tilde{L}_{12}(x;t)}\\{\tilde{L}_{21}(x;t)}&{\tilde{L}_{22}(x;t)}\end{pmatrix}
=\begin{pmatrix}{{\tilde{L}}_1(x;t)}&-{{\tilde{L}_2}^*(x;t)}\\{{\tilde{L}}_2(x;t)}&{{\tilde{L}_1}^*(x;t)}\end{pmatrix}\,,
\end{equation*}
we obtain an explicit expression for the elements of $\tilde{\bL}(x;t)$ in terms of $(A,B,C)$ and $N$, and $Q(t)$:
\begin{small}
\begin{subequations}\label{eq:Lexplicit}
\begin{align}
\tilde{L}_{11}(x;t)&=\tilde{L}_1(x;t)=-C(t)\,e^{-2xA}\,N\left[e^{2xA^\dagger}+Q(t)\,e^{-2xA}\,N\right]^{-1}\,(A^{\dagger})^{-1}\,C(t)^{\dagger}\,,\\
\tilde{L}_{12}(x;t)&=-\tilde{L}^*_2(x;t)=-C(t)\,\left[e^{2xA}+N\,e^{-2xA^{\dagger}}\,Q(t)\right]^{-1}\,A^{-1}\,B\,,\\
\tilde{L}_{21}(x;t)&=\tilde{L}_2(x;t)=B^{\dagger}\,\left[e^{2xA^{\dagger}}+Q(t)\,e^{-2xA}\,N\right]^{-1}\,(A^{\dagger})^{-1}\,C(t)^{\dagger}\,,\\
\tilde{L}_{22}(x;t)&=\tilde{L}^*_1(x;t)=-B^{\dagger}\,e^{-2xA^{\dagger}}\,Q(t)\,\left[e^{2xA}+N\,e^{-2xA^{\dagger}}\,Q(t)\right]^{-1}\,A^{-1}\,B\,.
\end{align}
\end{subequations}
\end{small}
where $C(t)=C\,e^{-4itA^2}$ and $Q(t)=(e^{-4itA^2})^{\dagger}\,Q\,e^{-4itA^2}$ as in \eqref{eq:final1f}. Using the determinant lemma \eqref{eq:detlemma} and exploiting the fact that $N$ and $Q(t)$ satisfy the Lyapunov equations \eqref{eq:5.4e} and \eqref{eq:5.4f}, one can directly verify that indeed $\tilde{L}^*_{11}(x;t)=\tilde{L}_{22}(x;t)$ and $\tilde{L}^*_{12}(x;t)=-\tilde{L}_{21}(x;t)$.

\smallskip
Finally, from \eqref{eq:Lstructure} and \eqref{eq:solvingformula}, we have
\begin{footnotesize}
\begin{equation*}
\begin{pmatrix} m_3(x;t)&m_{-}(x;t) \\ m_{+}(x;t)&-m_3(x;t) \end{pmatrix}=\begin{pmatrix} 2\,\left|1+\tilde{L}_1(x;t)\right|^2-1&-2\,\left(1+{\tilde{L}_1}^*(x;t)\right)\,{\tilde{L}_2}^*(x;t) \\ -2\,\left(1+\tilde{L}_1(x;t)\right)\,\tilde{L}_2(x;t)&-2\,\left|1+\tilde{L}_1(x;t)\right|^2+1\end{pmatrix}\,,
\end{equation*}
\end{footnotesize}
with $m_{+}(x;t) = m_1(x;t)+im_2(x;t) = m_{-}(x;t)^{*}$, entailing (\ref{eq:final2}).

\smallskip
The expression for the elements of $\tilde{\bL}(x;t)$ can be further simplified. Using the determinant lemma \eqref{eq:detlemma} on \eqref{eq:Lexplicit}, we have
\begin{small}
\begin{subequations}\label{eq:Lexplicit1}
\begin{align}
\tilde{L}_1(x;t)
&=\frac{\det\left(e^{2xA^{\dagger}}N^{-1}e^{2xA}-(A^{\dagger})^{-1}Q(t)A\right)}{\det\left(e^{2xA^{\dagger}}N^{-1}\,e^{2xA}+Q(t)\right)}-1=\label{eq:Lexplicit1a}\\
&=\frac{\det\left(N^{-1}-(A^{\dagger})^{-1}e^{-2xA^{\dagger}}Q(t)e^{-2xA}A\right)}{\det\left(N^{-1}+e^{-2xA^{\dagger}}Q(t)e^{-2xA}\right)}-1\,,\label{eq:Lexplicit1b}\\
\tilde{L}_2(x;t)
&=\frac{\det\left(e^{2xA^{\dagger}}+Q(t)e^{-2xA}N+(A^{\dagger})^{-1}C(t)^{\dagger}B^{\dagger}\right)}{\det\left(e^{2xA^{\dagger}}+Q(t)e^{-2xA}N\right)}-1=\label{eq:Lexplicit1c}\\
&=\frac{\det\left(e^{2xA^{\dagger}}N^{-1}e^{2xA}+Q(t)+(A^{\dagger})^{-1}C(t)^{\dagger}B^{\dagger}N^{-1}e^{2xA}\right)}{\det\left(e^{2xA^{\dagger}}N^{-1}e^{2xA}+Q(t)\right)}-1=\label{eq:Lexplicit1d}\\
&=\frac{\det\left(N^{-1}+e^{-2xA^{\dagger}}Q(t)e^{-2xA}+e^{-2xA^{\dagger}}(A^{\dagger})^{-1}C(t)^{\dagger}B^{\dagger}N^{-1}\right)}{\det\left(N^{-1}+e^{-2xA^{\dagger}}Q(t)e^{-2xA}\right)}-1\,.\label{eq:Lexplicit1e}
\end{align}
\end{subequations}
\end{small}
Expressions \eqref{eq:Lexplicit1a} and \eqref{eq:Lexplicit1d} are convenient for computing numerically the solution for $x$ large and negative. Similarly, expressions \eqref{eq:Lexplicit1b} and \eqref{eq:Lexplicit1e} are convenient for computing numerically the solution for $x$ large and positive. Furthermore, we have the following compact alternative expressions
\begin{small}
\begin{align*}
\tilde{L}_1(x;t)&=\frac{\det\left(I_{\bar{n}}-(A^{\dagger})^{-1}\tilde{Q}A\tilde{N}\right)}{\det\left(I_{\bar{n}}+\tilde{Q}\tilde{N}\right)}-1=\frac{\det\left(I_{\bar{n}}-\tilde{N}(A^{\dagger})^{-1}\tilde{Q}A\right)}{\det\left(I_{\bar{n}}+\tilde{N}\tilde{Q}\right)}-1\,,\\
{\tilde{L}_1}^*(x;t)&=\frac{\det\left(I_{\bar{n}}-A^{\dagger}\tilde{Q}A^{-1}\tilde{N}\right)}{\det\left(I_{\bar{n}}+\tilde{Q}\tilde{N}\right)}-1=\frac{\det\left(I_{\bar{n}}-(\tilde{N}A^{\dagger})\tilde{Q}A^{-1}\right)}{\det\left(I_{\bar{n}}+\tilde{N}\tilde{Q}\right)}-1\,,\\
\tilde{L}_2(x;t)&=\frac{\det\left(I_{\bar{n}}+\tilde{Q}\tilde{N}+(A^{\dagger})^{-1}\tilde{C}^{\dagger}\tilde{B}^{\dagger}\right)}{\det\left(I_{\bar{n}}+\tilde{Q}\tilde{N}\right)}-1\,,\\
{\tilde{L}_2}^*(x;t)&=\frac{\det\left(I_{\bar{n}}+\tilde{N}\tilde{Q}+\tilde{B}\tilde{C}A^{-1}\right)}{\det\left(I_{\bar{n}}+\tilde{N}\tilde{Q}\right)}-1\,,
\end{align*}
\end{small}
where
\begin{equation*}
\tilde{B}=e^{-xA}\,B\,,\qquad
\tilde{C}\equiv \tilde{C}(x;t)=C(t)\,e^{-xA}=Ce^{-4itA^2}\,e^{-xA}\,,
\end{equation*}
and $\tilde{N}\equiv \tilde{N}(x)$ and $\tilde{Q}\equiv \tilde{Q}(x;t)$ are the following self-adjoint matrix functions
\begin{align*}
\tilde{N} \equiv \tilde{N}(x)=\tilde{N}^{\dagger}=e^{-xA}Ne^{-xA^{\dagger}}\,,\quad
\tilde{Q} \equiv \tilde{Q}(x;t)=\tilde{Q}^{\dagger}=e^{-xA^{\dagger}}Q(t)e^{-xA}
\end{align*}
satisfying the Lyapunov equations
\begin{align*}
A^{\dagger}\tilde{Q}+\tilde{Q}A =\tilde{C}^{\dagger}\tilde{C}\,,\quad
A\tilde{N}+\tilde{N}A^{\dagger} =\tilde{B}\tilde{B}^{\dagger}.
\end{align*}
Incidentally, we observe that the following non-trivial identity can be proved (after some efforts):
\begin{equation*}
\det\left(I_{\bar{n}}+\tilde{N}\tilde{Q}+\tilde{B}\tilde{C}A^{-1}\right)=\det\left(I_{\bar{n}}+\tilde{N}\tilde{Q}+A^{-1}\tilde{B}\tilde{C}\right).
\end{equation*}
Clearly the same identity is also true if we replace $\tilde{B},\tilde{C},\tilde{N}$ and $\tilde{Q}$ with the corresponding non-tilde matrices $B,C,N$ and $Q$.

A completely explicit expression for $\bm(x,t)$ can be given in terms of the matrix triplet $(A,\tilde{B},\tilde{C})$ and $\tilde{N}$ and $\tilde{Q}$:
\begin{small}
\begin{subequations}\label{eq:final3}
\begin{align}
m_{+}(x;t)&=-2\,\frac{\det\left(I_{\bar{n}}-(A^{\dagger})^{-1}\tilde{Q}A\tilde{N}\right)}{\det\left(I_{\bar{n}}+\tilde{Q}\tilde{N}\right)}\left[\frac{\det\left(I_{\bar{n}}+\tilde{Q}\tilde{N}+(A^{\dagger})^{-1}\tilde{C}^{\dagger}\tilde{B}^{\dagger}\right)}{\det\left(I_{\bar{n}}+\tilde{Q}\tilde{N}\right)}-1\right]\,,\\
m_{-}(x;t)&=-2\,\frac{\det\left(I_{\bar{n}}-\tilde{N}A^{\dagger}\tilde{Q}A^{-1}\right)}{\det\left(I_{\bar{n}}+\tilde{N}\tilde{Q}\right)}\left[\frac{\det\left(I_{\bar{n}}+\tilde{N}\tilde{Q}+\tilde{B}\tilde{C}A^{-1}\right)}{\det\left(I_{\bar{n}}+\tilde{N}\tilde{Q}\right)}-1\right]\,,\\
m_3(x;t)&=2\,\frac{\det\left(I_{\bar{n}}-(A^{\dagger})^{-1}\tilde{Q}A\tilde{N}\right)\det\left(I_{\bar{n}}-A^{\dagger}\tilde{Q}A^{-1}\tilde{N}\right)}{\det\left(I_{\bar{n}}+\tilde{Q}\tilde{N}\right)^2}-1\nonumber\\
&=2\,\frac{\det\left(A^{\dagger}-\tilde{Q}A\tilde{N}\right)\det\left((A^{\dagger})^{-1}-\tilde{Q}A^{-1}\tilde{N}\right)}{\det\left(I_{\bar{n}}+\tilde{Q}\tilde{N}\right)^2}-1\nonumber\\
&=2\,\frac{\det\left(A^{-1}-\tilde{N}(A^{\dagger})^{-1}\tilde{Q}\right)\det\left(A-\tilde{N}A^{\dagger}\tilde{Q}\right)}{\det\left(I_{\bar{n}}+\tilde{N}\tilde{Q}\right)^2}-1\,,
\end{align}
\end{subequations}
\end{small}
with $m_1(x;t)=\mathrm{Re}\left(m_{+}(x;t)\right)$ and $m_2(x;t)=\mathrm{Im}\left(m_{+}(x;t)\right)$.

%\printbibliography %to be used with BiBLaTeX
\bibliographystyle{plain}

\end{document}